%% file: projrefinewlp.tex
\PassOptionsToPackage{prologue, dvipsnames}{xcolor}
\documentclass[acmsmall,screen]{acmart}

\AtBeginDocument{%
	\providecommand\BibTeX{{%
			\normalfont B\kern-0.5em{\scshape i\kern-0.25em b}\kern-0.8em\TeX}}}

\usepackage{mathtools}
\usepackage{tikz}
\usepackage{booktabs}   
\usepackage{subcaption} 
\newsavebox{\tablebox}
\usetikzlibrary{quantikz}
\usepackage{mathpartir}

\setcopyright{none}
\renewcommand\footnotetextcopyrightpermission[1]{} 
\newcommand{\graycode}[1]{\colorbox{gray!20}{\small #1}}

\input{lstlisting}

\begin{document}
	
	\title{Refinement calculus of quantum programs with projective assertions}
	

	
\author{Yuan Feng}
\affiliation{
\institution{Centre for Quantum Software and Information, University of Technology Sydney}
\city{Sydney}
\country{Australia}
}
\email{yuan.feng@uts.edu.au}
\author{Li Zhou}
\affiliation{
\institution{State Key Laboratory of Computer Science, Institute of Software, Chinese Academy of Sciences}
\city{Beijing}
\country{China}
}
\email{li.zhou@ios.ac.cn}
\author{Yingte Xu}
\affiliation{
\institution{Max Planck Institute for Security and Privacy (MPI-SP)}           
\city{Bochum}
\country{Germany}
}
\email{lucianoxu@foxmail.com}
 
	\begin{abstract}
		
		Refinement calculus provides a structured framework for the progressive and modular development of programs, ensuring their correctness throughout the refinement process. This paper introduces a refinement calculus tailored for quantum programs. To this end, we first study the partial correctness of nondeterministic programs within a quantum while language featuring prescription statements. Orthogonal projectors, which are equivalent to subspaces of the state Hilbert space, are taken as assertions for quantum states. In addition to the denotational semantics where a nondeterministic program is associated with a set of trace-nonincreasing super-operators, we also present their semantics in transforming a postcondition to the weakest liberal postconditions and, conversely, transforming a precondition to the strongest postconditions. Subsequently, refinement rules are introduced based on these dual semantics, offering a systematic approach to the incremental development of quantum programs applicable in various contexts. To illustrate the practical application of the refinement calculus, we examine examples such as the implementation of a $Z$-rotation gate, the repetition code, and the quantum-to-quantum Bernoulli factory. Furthermore, we present \texttt{Quire}, a Python-based interactive prototype tool that provides practical support to programmers engaged in the stepwise development of correct quantum programs.
		
	\end{abstract}




	\maketitle

\input{pmymacro}
	\def\qassert{P}
	\def\qassertp{Q}
	
	\section{Introduction}
	
	Refinement calculus~\cite{dijkstra1976discipline,morgan1994programming,back1998refinement} is a formal method in software engineering used to specify and verify the correctness of software systems. Unlike Hoare logic~\cite{hoare1969axiomatic} or other static analysis methods~\cite{nielson1999principles},  refinement calculus does not verify a specific executable program. Instead, it specifies the desired behavior of a system through the use of abstract specifications, often expressed in terms of preconditions and postconditions. Preconditions define assumptions that must be satisfied before a program can be executed, and postconditions specify the expected results or properties that the final states should satisfy after execution of the program. The abstract specification is progressively refined into more concrete and detailed implementations by utilizing a set of rules that govern the steps of refinement. Each refinement step introduces additional details while maintaining correctness, resulting in a series of refined programs that converge to the final executable implementation. This approach helps to lay a rigorous foundation for software engineering in classical computing systems, ensuring the correctness, reliability, and maintainability of complex software~\cite{kourie2012correctness,mciver2020correctness}.

	By exploiting quantum parallelism caused by quantum superposition, quantum computing promises to provide unparalleled computing power~\cite{nielsen2002quantum,shor1994algorithms,grover1996fast,harrow2009quantum,childs2003exponential}. Developing reliable quantum programs is critical to effectively utilizing the power of quantum computers. However, due to the subtle and counterintuitive nature of quantum computing, the design of quantum programs is highly error-prone, thus requiring systematic approaches to design, analyze, and verify these programs. Current efforts related to the development and verification of quantum programs include mainly quantum Hoare logic~\cite{ying2012floyd,zhou2019applied,feng2021quantum,chadha2006reasoning,Kakutani:2009} and testing methods~\cite{li2020projection,liu2020quantum,huang2019statistical}. However, these techniques always assume that the program in question is fully developed and therefore cannot be applied in the early stages of program development.

	In this paper, we propose a refinement calculus for quantum programs. The predicates used to represent pre- and postconditions of quantum programs are taken to be subspaces (or equivalently, orthogonal projectors) of the associated Hilbert space. This design decision is motivated by the following considerations. Firstly, the collection of all subspaces within a finite-dimensional Hilbert space $\h$ constitutes a complete modular lattice. In contrast, the set of effects (positive operators with all eigenvalues not larger than 1), which are taken as quantum assertions in, say, ~\cite{d2006quantum,ying2012floyd} for verification of deterministic quantum programs, is merely a partially ordered set (CPO). As a result, utilizing projectors as assertions leads to a more elegant and richer theory, which can simplify the verification of quantum programs, especially those involving nondeterminism.
    In particular, this enables the utilization of quantum logic to aid in reasoning about our assertion logic, partially avoiding traditional matrix/linear algebraic calculations, which is of considerable importance for tooling and automation. Secondly, we observe that for the loop construct in our language, the weakest liberal preconditions and the strongest postconditions can be expressed as the meet and join of certain subspaces, respectively. For example, 
	$
			wlp.(\pwstm).Q  = \bigwedge_{n\geq 0} wlp.\while^n.Q
	$
	where $\while^n$ is the $n$-th approximation of $\pwstm$. Note that the sequence of subspaces $wlp.\while^n.Q$ is decreasing with respect to the set-inclusion order. Thus, their dimensions are also decreasing. As a result, the weakest liberal preconditions of quantum programs can be computed, and therefore their correctness verification can be done, effectively.
	Finally, taking projectors as quantum assertions does not significantly compromise expressiveness, contrary to initial impressions. In fact, it has the same power as effect assertions, provided that both pre- and postconditions are represented as projectors. More specifically, it has been demonstrated in~\cite{zhou2019applied} that for any quantum program $S$ and projectors $P$ and $Q$, the Hoare triple $\ass{P}{S}{Q}$ is semantically valid in the applied Hoare logic proposed in~\cite{zhou2019applied} if and only if it is semantically valid in the original quantum Hoare logic introduced in~\cite{ying2012floyd}. It is worth highlighting that this equivalence holds for both partial correctness and total correctness.

	To establish a robust theoretical foundation for the refinement calculus, we introduce a nondeterministic quantum language that expands upon the extensively researched quantum while language. This extension incorporates a construct known as \emph{prescription}, akin to the classical language in~\cite{morris1987theoretical}. A prescription consists of a pair of projectors defining the pre- and post-conditions of a fully executable program. Our main contributions in this paper include:

	\begin{enumerate}
		\item \emph{Formal semantics}. In addition to the denotational semantics associating a nondeterministic program with a set of trace-nonincreasing super-operators, we also present their semantics in terms of backward and forward predicate transformers, respectively. Notably, we discover unexpected applications of the Sasaki implication and conjunction~\cite{herman1975implication} in the weakest liberal preconditions and the strongest postconditions, respectively. This stands in contrast to the roles played by ordinary implication and conjunction in classical programs.
		
		\item \emph{Refinement calculus}. We introduce two types of refinement rules based on the weakest liberal precondition and the strongest postcondition semantics, respectively. These rules provide a structured approach to the incremental development of quantum programs that are applicable in various contexts. Moreover, we present refinement rules for assertion statements, useful in scenarios such as projection-based runtime testing and debugging of quantum programs~\cite{li2020projection}. A number of examples, including the implementation of a $Z$-rotation gate, the repetition code, and the quantum-to-quantum Bernoulli factory, have been examined to illustrate the practicality of the refinement calculus.
		
		\item \emph{Prototype implementation}.  To aid quantum programmers in the stepwise development of correct quantum programs, we develop \texttt{Quire}, a Python-based interactive prototype tool that implements the refinement rules proposed in this paper. In addition, this tool can check the well-formedness of operator terms and quantum programs, perform classical simulations of program execution, and verify whether specified prescriptions are satisfied by a fully implemented program.
	Furthermore, all the examples studied in this paper, excluding the quantum-to-quantum Bernoulli factory, which requires the formalization and analysis of parametric programs, have been successfully developed with the help of this tool.
	\end{enumerate}
	
 	\textbf{Related works}. It is well known that refinement calculus is closely related to Hoare logic. Essentially, both the refinement rules in the former and the proof rules in the latter can be derived from the weakest (liberal) precondition semantics of the programs under consideration. However, in refinement calculus, it is common to extend the language by incorporating \emph{specifications} (prescription statements in this paper), making it more comprehensive than Hoare logic, which exclusively deals with executable programs. The refinement calculus proposed in this paper provides a broader framework than the applied Hoare logic in~\cite{zhou2019applied} precisely for the same reason: we consider a richer language here and, as a result, we discover novel applications of Sasaki operators in the refinement rules.  
	
    The technique of stepwise refinement has previously been employed in~\cite{zuliani2007formal} to derive Grover’s algorithm from a probabilistic specification of the search problem. Throughout the refinement process, the rules from~\cite{morgan1994programming} for probabilistic programs, along with new ones addressing unitary transformations and quantum measurements, are applied. However, the specification language used is classical, and thus, it cannot describe the correctness of general quantum algorithms. Additionally, only a specific example of generating quantum programs from their specifications is presented, and a comprehensive refinement calculus tailored for quantum programs, like the one proposed in this paper, is absent in~\cite{zuliani2007formal}.
	
	Another related work is~\cite{tafliovich2009programming}, where the specification and analysis of quantum communication protocols in a predictive programming language~\cite{tafliovich2006quantum} are demonstrated. However, the considered specifications are limited to Boolean or probabilistic expressions with atomic propositions characterizing pure quantum states. Furthermore, the emphasis of~\cite{tafliovich2009programming} is on quantum communication, leading to refinement rules designed for message passing through classical and quantum channels. 
			
	Dynamic quantum logic was proposed and used in~\cite{baltag2011quantum,baltag2008dynamic,brunet2004dynamic,baltag2006lqp} to analyze the properties of quantum programs in an abstract way. Notably, the significant roles played by Sasaki implication and conjunction in this context have already been observed in the works of Baltag and Smets. However, in this paper, we take a step further by incorporating these notions into the analysis of quantum programs with unspecified components and proposing a refinement calculus based on them.

	Several verification tools have been developed for quantum programs. VOQC~\cite{hietala2021verified} is a verified optimizer for quantum circuits written in Coq. QWire~\cite{paykin2017qwire,Rand2018QWIREPF} and Qbricks~\cite{chareton2021automated} specialize in the proving of the properties of quantum circuits. In contrast, QHLProver~\cite{liu2019formal} and CoqQ~\cite{zhou2023coqq} employ quantum Hoare logic to verify quantum programs with complex structures like conditional branching and while loops. 
	
	Upon completion of this paper, we became aware of an independently developed work~\cite{peduri2023qbc}. In this work, an approach called Quantum Correctness by Construction (QbC) is proposed, which, similar to our refinement calculus, aims to assist in constructing quantum programs from their specifications. They also consider a quantum while language equipped with prescription statements. However, the distinction lies in the fact that the quantum predicates considered in~\cite{peduri2023qbc} are effects, while ours are projectors. As mentioned earlier, projectors enjoy much more favorable properties than effects, enabling us to provide significantly richer refinement rules for program generation. In particular, the existence of the strongest postconditions leads to a whole set of rules that cannot be obtained by using the weakest preconditions.

	\textbf{Structure of the paper}. The rest of this paper is organized as follows. We introduce in Sec.~\ref{sec:spec-lang} the target programming language of our analysis, and define its denotational semantics. The weakest liberal precondition and strongest postcondition semantics are given in Sec.~\ref{sec:wlpsp}. The main part of this paper is Sec.~\ref{sec:calculus} in which we propose a set of refinement rules based on the pre/postcondition semantics in Sec.~\ref{sec:wlpsp}. 	 
	 Illustrative examples are explored in Sec.~\ref{sec:case} to show the effectiveness of our refinement calculus. Sec.~\ref{sec:quire} is devoted to a Python-based prototype that implements the program refinement techniques developed in this paper. 
	 Finally, Sec.~\ref{sec:conclusion} concludes the paper and points out some directions for future study. 
  
    A brief introduction to the basic notions of quantum computing and quantum logic is presented in the Supplementary Material \ref{sec-preliminaries} to help understand the content of the paper. For the sake of readability, we omit all proofs in the main text. Interested readers may find the details in the Supplementary Material \ref{sec:deferred-proofs} as well.
	
	\section{A specification quantum language}
	\label{sec:spec-lang}
	The target programming language of our analysis is
an extension of the purely quantum while language defined in~\cite{zhou2019applied,feng2023abstract} with assertions and prescriptions.
	Let $\QVar$, ranged over by $q, r, \cdots$, be a finite set of (qubit-type) quantum variables. For any subset $W$ of $\QVar$, let
	$\h_W \define \bigotimes_{q\in W} \h_{q},
	$
	where $\h_{q}$ is the 2-dimensional Hilbert space associated with $q$ and its computational basis is denoted as $\{|0\>_q, |1\>_q\}$. As we use subscripts to distinguish Hilbert spaces with different quantum variables, their order in the tensor product is irrelevant. 		
	
	In this paper, we regard projectors as the \emph{qualitative} predicates for quantum states. We use the same symbol, such as $P$, to denote both a subspace and its corresponding projector. The intended meaning of these notations should be clear from the context. Consequently, a quantum state $|\psi\> \in P$ iff $P|\psi\> = |\psi\>$. Here, the former $P$ denotes a subspace, while the latter represents the corresponding projector. 
	Note that the collection of subspaces in $\h_V$ forms a complete lattice under the L\"{o}wner order $\le$ defined between linear operators such that $M\le N$ if and only if $N-M$ is positive. When applied to projectors, this order coincides with the set inclusion order between the corresponding subspaces; that is, $P\le Q$ (regarding as projectors) if and only if $P\subseteq Q$ (regarding as subspaces). Let $P_i$, $i\geq 0$, be a sequence of projectors.  
	Denote by $\bigwedge_{i\geq 0} P_i$ and $\bigvee_{i\geq 0} P_i$ the meet and join of $P_i$, respectively.
	 The satisfaction relation $\models$ between quantum states and projectors is defined as follows: Given a state $\rho$ and a projector $P$, $\rho \models P$ if and only if $\supp{\rho} \subseteq P$
where $\supp{\rho}$ denotes the support subspace of $\rho$; that is, the subspace spanned by eigenvectors of $\rho$ associated with non-zero eigenvalues. Note that $\supp{\rho} = \mathcal{N}(\rho)^\bot$ where $\mathcal{N}(\rho)\define\left\{|\psi\>\in \h_{V} : \<\psi|\rho|\psi\> = 0\right\}$ is the null subspace of $\rho$, and $Q^\bot$ denotes the ortho-complement of a subspace $Q$. For any finite dimensional Hilbert space $\h$, denote by $\d(\h)$ and $\mathcal{SO}(\h)$ the sets of partial density operators (positive operators with trace not larger than 1), and completely positive and trace-nonincreasing super-operators on $\h$, respectively.
	
	\begin{definition}\label{def:synsem}
		The syntax and denotational semantics of our language are given inductively as follows:
		\begin{equation*}
			\begin{array}{cclrl}
				& & \mbox{Syntax} & & \mbox{Semantics } \sem{S}\!: \dhv\ra  \dhv\\
				S & ::= & \sskip & \textit{(no-op)} & \left\{\id\right\}\\ 
				&|&  \abort& \textit{(abortion)} &  \{0\}\\
				&|&  \bar{q}:=0 & \textit{(initialisation)}& \left\{\mathit{Set}^{|0\>}_{\bar{q}}\right\}\\
				&| & 	\bar{q}\apply U &\textit{(unitary operation)}&\left\{\mathcal{U}_{\bar{q}}\right\}\\
				&|&   \assert{P[\bar{q}]}& \textit{(assertion)}& \left\{ \p_{\bar{q}}\right\}\\	
				&|&   [P,Q]_{\bar{q}}& {\textit{(prescription)}}& \left\{\e\in \mathcal{SO}(\h_{\bar{q}}): \forall \rho\models P, \e(\rho)\models Q\right\}\\	
				&|&   S_0\pcom{p} S_1&\textit{(prob. choice)}& p\sem{S_0} + (1-p)\sem{S_1}\\	
				&|&   S_0;S_1&\textit{(sequence)}& \sem{S_1}\circ \sem{S_0}\\
				&|&  \pmstm	& \textit{(conditional)}&\sem{S_1}\circ \p_{\bar{q}} +\sem{S_0}\circ \p^\bot_{\bar{q}}\\
				&|&  \pwstm	& \textit{(loop)}&  \sum_{k=0}^\infty \p^\bot_{\bar{q}} \circ
				\left(\sem{S}\circ \p_{\bar{q}}\right)^k
			\end{array}
		\end{equation*}
			where $S,S_0$ and $S_1$ are quantum programs, $\bar{q} \define q_1, \ldots, q_n$ a (ordered) tuple of distinct quantum variables from $\QVar$, $U$ a unitary operator on
			$\h_{\bar{q}}$, $P$ and $Q$ subspaces of $\h_{\bar{q}}$, and $p\in [0,1]$.
			Sometimes, we also use $\bar{q}$ to denote the (unordered) set $\{q_1,q_2,\dots,q_n\}$. Let $\qv(S)$ be the set of quantum variables that appear in $S$. 
			In the definition of denotational semantics, $\id$ and 0 are the identity and zero super-operators, respectively, $\mathit{Set}^{|0\>}_{\bar{q}}$ and $\u_{\bar{q}}$ with $|\bar{q}|=n$ are super-operators such that $\mathit{Set}^{|0\>}_{\bar{q}}(\rho) = \sum_{i=0}^{2^n-1}\quzi \rho\quiz$ and $\u_{\bar{q}}(\rho) = U_{\bar{q}} \rho U_{\bar{q}}^\dag$, and $\p_{\bar{q}}$ and $\p^\bot_{\bar{q}}$ are super-operators such that $\p_{\bar{q}}(\rho) = P_{\bar{q}} \rho P_{\bar{q}}$ and $\p^\bot_{\bar{q}}(\rho) = P^\bot_{\bar{q}} \rho P^\bot_{\bar{q}}$. Furthermore, operations such as composition $\circ$ and addition $+$ on individual super-operators are assumed to be extended to sets of super-operators in an element-wise way. For example, let $\supoprset$ and $\mathbb{F}$ be two sets of super-operators. Then $\supoprset \circ \e + \mathbb{F}\circ \f\define \{\e' \circ \e+ \f' \circ \f : \e'\in \supoprset, \f'\in \mathbb{F}\}.$ For clarification, we often use subscripts to emphasize the quantum system on which an operator is performed. For example, $P_W$ represents $P$ acting on system $W$. To simplify the notation, we do not distinguish between $P_W$ and its cylindrical extension $P_W\otimes I_{V\backslash W}$ to $\h_V$. We call a program \emph{executable} if it does not contain any prescription.
		\end{definition}

		Compared to the quantum while language considered in~\cite{zhou2019applied}, we have added two additional constructs here. Intuitively, $\assert P[\bar{q}]$ determines whether the current state is in the subspace $P$. If so, the program continues with the next command. Otherwise, it aborts without any output. On the contrary, $[P,Q]_{\bar{q}}$ is a nondeterministic construct that denotes the set of quantum programs containing only variables in $\bar{q}$ that, with the precondition $P$, establish the postcondition $Q$ upon termination.
	The following lemma justifies the semantic definition given in Definition.~\ref{def:synsem}.
		\begin{lemma}\label{lem:densemantics}
			For any programs $S$, the semantics $\sem{S}$ is a set of super-operators on $\h_{\qv(S)}$.
		\end{lemma}
		
		\begin{example}[A running example]\label{exam:pcircuit}
			We recall the algorithm in~\cite[Exercise 4.41]{nielsen2002quantum} designed to implement the $Z$-rotation gate $R_z(\theta) = \cos\frac{\theta}{2} I - i\sin\frac{\theta}{2} Z$ where $\cos\theta = 3/5$. The key component of this algorithm is a quantum circuit depicted on the left of Fig.~\ref{fig-rz} which, in the case that both measurement outcomes are 0, applies $R_z(\theta)$ to the last qubit. This happens with probability $5/8$. Otherwise, the circuit applies $Z$ to the last qubit, so we can apply another $Z$ to bring its state back to the original. 
			
			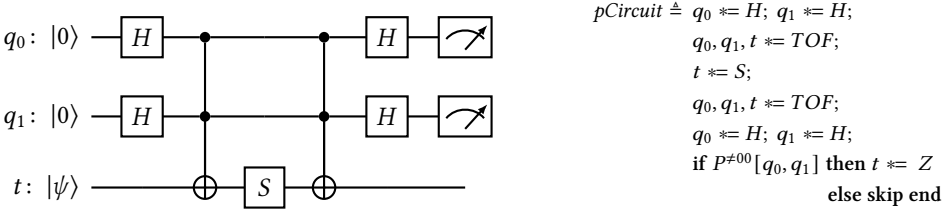
\begin{figure}[t]
				\centering
				\begin{minipage}[t]{0.5\textwidth}
					\centering
					\tikzset{
						my label/.append style={above right,xshift=0.5cm}
					}
					\begin{quantikz}[row sep=0.4cm,column sep=0.4cm]
						\lstick{$q_0\!:\  |0\>$} &\gate{H} & \ctrl{2}  & \qw& \ctrl{2} & \gate{H} &  \meter{} \\	  
						\lstick{$q_1\!:\  |0\>$} &\gate{H} & \ctrl{1}  & \qw & \ctrl{1} & \gate{H} &  \meter{} \\	
						\lstick{$t\!:\  |\psi\>$}  &\qw &\targ{}& \gate{S} & \targ{}&\qw  & \qw 
					\end{quantikz} 
				\end{minipage}
				\begin{minipage}[t]{0.48\textwidth}
					\centering
					\renewcommand{\arraystretch}{1.15}
					\footnotesize
					\begin{tabular}{rl}
						$\mathit{pCircuit}\define$
						&\hspace{-0.3cm} $q_0\apply H;\ q_1\apply H;$\\
						&\hspace{-0.3cm} $q_0,q_1,t\apply TOF;$\\
						&\hspace{-0.3cm} $t\apply S;$\\			
						&\hspace{-0.3cm} $q_0,q_1,t\apply TOF;$\\
						&\hspace{-0.3cm} $q_0\apply H;\ q_1\apply H;$\\
						&\hspace{-0.3cm} $\iif\ P^{\neq 00}[q_0,q_1]\ \then\ t\apply\ Z$\\
						&\hspace{1.5cm} $\eelse\ \sskip\ \pend$
					\end{tabular}
				\end{minipage}
				\caption{Quantum circuit (left) and program with correction (right) that implement $R_z(\theta)$ with probability 5/8.}\label{fig-rz}
			\end{figure}

			The algorithm for this key circuit plus the correction $Z$ gate when at least one measurement outcome is 1 can be written in our language shown on the right of Fig. \ref{fig-rz}, 
			where $P^{\neq 00} = {P^{00}}^\bot = |01\>\<01| \vee |10\>\<10| \vee |11\>\<11|$, $TOF$ is the Toffolli gate, which applies $X$ on the target (third) qubit if both the control qubits are in $|1\>$; otherwise, it does nothing. Note that $\mathit{Rz}$ does not contain any prescription, and so it is readily executable. 
			
			To compute the denotational semantics of $\mathit{pCircuit}$, we first note from Definition~\ref{def:synsem} that 
			\begin{align*}
				\sem{ \iif\ P^{\neq 00}[q_0,q_1]\ \then\ t\apply\ Z\ \eelse\ \sskip\ \pend} &= \sem{t\apply Z}\circ \p^{\neq 00}_{q_0, q_1} + \sem{\sskip}\circ \p^{00}_{q_0, q_1}\\
				&= \left\{\e^Z_t\circ \p^{\neq 00}_{q_0, q_1} + \p^{00}_{q_0, q_1}\right\}.
			\end{align*}
			Here, for a unitary gate $U$, we denote by $\e^U$ the corresponding super-operator such that $\e^U(\rho) = U\rho U^\dag$ for all $\rho$.
			Repeating using the rules in Definition~\ref{def:synsem}, we find
			\begin{equation}\label{eq:sempcircuit}
				\sem{\mathit{pCircuit}} = \left\{\left(\e^Z_t\circ \p^{\neq 00}_{q_0, q_1} + \p^{00}_{q_0, q_1}\right)\circ \e_{q_1}^H\circ\e_{q_0}^H\circ\e_{q_0,q_1,t}^{TOF}\circ \e_t^S\circ\e_{q_0,q_1,t}^{TOF}\circ\e_{q_1}^H\circ\e_{q_0}^H\right\}.
			\end{equation}
			We choose not to simplify the final result, since it is irrelevant to subsequent discussions.
		\end{example}		
  
		As usual, the correctness of a quantum program is expressed through a \emph{Hoare triple} $\ass{P}{S}{Q}$, which consists of a program $S$ and a pair of assertions $P$ and $Q$, called \emph{precondition} and \emph{postcondition} of $S$, respectively. Such a Hoare triple is said to be \emph{(partially) correct}, denoted $\models\ass{P}{S}{Q}$, if for any $\rho$ and $\e\in \sem{S}$,
		\begin{equation}\label{eq:defcor}
			\rho\models P \quad \mbox{implies}\quad \e(\rho)\models Q.
		\end{equation}	
        This means that starting from any quantum state supported in the subspace $P$, the final state after executing $S$ must be supported in the subspace $Q$.

		\begin{example}\label{exam:pcircuitcor} Let us revisit the running example. The correctness of the program $\mathit{pCircuit}$, disregarding its success probability, can be stated as for any state $|\psi\>\in \h_t$,
			\[
			\models \ass{|00\>_{q_0,q_1}\<00|\otimes |\psi\>_t\<\psi|}{\mathit{pCircuit}}{\left(P^{\neq 00}_{q_0,q_1}\otimes |\psi\>_t\<\psi|\right) \vee \left(P^{00}_{q_0,q_1}\otimes R_z(\theta)|\psi\>_t\<\psi|R_z(\theta)^\dagger\right)}.
			\]
			Intuitively, the precondition assumes that the input state of system $q_0$ and $q_1$ are both $|0\>$, while $t$ is in state $|\psi\>$. This can be seen from the fact that for any $\rho$, if $\rho \models|00\>_{q_0,q_1}\<00|\otimes |\psi\>_t\<\psi|$ then $\rho$ must be the pure state $|00\>_{q_0,q_1} |\psi\>_t$. Furthermore, the postcondition indicates that the output state has only two possibilities: either $q_0$ and $q_1$ are both $|0\>$ and the system $t$ is in state $R_z(\theta)|\psi\>$, thus implementing the desired rotation $R_z(\theta)$ successfully; or otherwise (either $q_0$ or $q_1$ is $|1\>$)  the system $t$ remains in the original state $|\psi\>$. To see this, suppose that the final state $\sigma$ satisfies the postcondition. Then $\sigma = \sigma_0 + \sigma_1$ such that 
			\[
			\sigma_0 \models P^{\neq 00}_{q_0,q_1}\otimes |\psi\>_t\<\psi| \qquad \mbox{and}\qquad \sigma_1 \models P^{00}_{q_0,q_1}\otimes R_z(\theta)|\psi\>_t\<\psi|R_z(\theta)^\dagger
			\]
			This in turn implies $\sigma_0$ is a product state and its reduced state on $t$ is $|\psi\>$, while  $\sigma_1$ is proportional to the pure state $|00\>_{q_0,q_1} R_z(\theta)|\psi\>_t$.

			By introducing the maximally entangled state $\Omega_{t,t'}\define |\omega\>_{t,t'}\<\omega|$, where $|\omega\> \define (|00\> + |11\>)/\sqrt{2}$, on $t$ and an auxiliary qubit $t'$, this correctness requirement can be simply stated as a \emph{single} formula
			\begin{equation}\label{eq:corpcircuit}
				\models \ass{|00\>_{q_0,q_1}\<00|\otimes \Omega_{t,t'}}{\mathit{pCircuit}}{\left(P^{\neq 00}_{q_0,q_1}\otimes  \Omega_{t,t'}\right) \vee \left(P^{00}_{q_0,q_1}\otimes \r_z(\theta)_t(\Omega_{t,t'})\right)}
			\end{equation}
			where $\r_z(\theta)$ is the super-operator corresponding to the unitary operator $R_z(\theta)$. To prove this, let $\sem{\mathit{pCircuit}} = \{\e\}$ as defined in~\eqref{eq:sempcircuit} and $\rho \models|00\>_{q_0,q_1}\<00|\otimes \Omega_{t,t'}$ be normalized. Then $\rho$ must be the pure state $|00\>_{q_0,q_1}|\omega\>_{t,t'}$. Note that all operations except the last one in $\e$ are unitary. We calculate 
			\begin{align*}
				|00\>_{q_0,q_1}|\omega\>_{t,t'} 
				\xrightarrow{\e^H_{q_0}}\xrightarrow{\e^H_{q_1}}\  & |++\>_{q_0,q_1}|\omega\>_{t,t'}\\
				\xrightarrow{\e^{TOF}_{q_0,q_1,t}}\xrightarrow{\e^{S}_{t}}\xrightarrow{\e^{TOF}_{q_0,q_1,t}}   \  & \frac{1}{2}\left[(|00\>+|01\>+|10\>)S|\omega\> + |11\>XSX|\omega\>\right]\\
				\xrightarrow{\e^H_{q_0}}\xrightarrow{\e^H_{q_1}} \ &\frac{1}{2}\left[(|++\>+|+-\>+|-+\>)S|\omega\> + |--\>XSX|\omega\>\right]\\
				=\ & \frac{1}{4}\left[|00\>(3S+XSX)|\omega\> + \left(|01\>+|10\>+|11\>\right)(S-XSX)|\omega\>\right]\\
				=\ & \frac{1-i}{4}\left[\sqrt{5}|00\>U|\omega\> + \left(|01\>+|10\>+|11\>\right)Z|\omega\>\right]
			\end{align*}
			where $U = diag\left(\frac{1+2i}{5}, \frac{-1+2i}{5}\right) = iR_{z}(\theta)$. Let $|\Phi\>$ be the last state in the above computation. Then 
			$$\p^{00}_{q_0, q_1}\left(|\Phi\>\<\Phi|\right) = \left|\frac{1-i}{4}\right|^2 \cdot 5 \cdot |00\>_{q_0,q_1}\<00| \otimes U|\omega\>_t\<\omega| U^\dag= \frac{5}{8} |00\>_{q_0,q_1}\<00| \otimes \r_z(\theta)_t(\Omega_{t,t'})$$
			while
			$\e^Z_t\circ \p^{\neq 00}_{q_0, q_1}\left(|\Phi\>\<\Phi|\right) = \frac{3}{8} |\phi\>_{q_0,q_1}\<\phi| \otimes \Omega_{t,t'}$
			where $|\phi\> = (|01\>+|10\>+|11\>)/\sqrt{3}$. It is then easy to see that $\e(\rho)=\e^Z_t\circ \p^{\neq 00}_{q_0, q_1}(\rho) + \p^{00}_{q_0, q_1}(\rho)$ indeed satisfies the postcondition of Eq.\eqref{eq:corpcircuit}. 
		\end{example}

		Note that from~\cite{zhou2019applied}, $\models \ass{P}{S}{Q}$ if and only if for any $\rho$,
		\begin{equation}\label{eq:corequiv}
			\tr(P\rho) \leq \inf \left\{\tr(Q\e(\rho)) + \tr(\rho) - \tr(\e(\rho)): \e\in \sem{S}\right\}.
		\end{equation}
            This gives a close relationship between quantum Hoare logics where quantum predicates are given as projectors or effects. In the following, we present two more equivalent conditions for partial correctness that are useful for discussions in the next section.
		\begin{lemma}\label{lem:equivpcorrectness}
			For any program $S$ and projectors $P$ and $Q$,
			\[
			\models \ass{P}{S}{Q} \quad \mbox{iff} \quad\forall \e\in \sem{S},\ P\le \mathcal{N}(\e^\dag(Q^\bot))\quad \mbox{iff}\quad \forall \e\in \sem{S},\ \supp{\e(P)}\le Q
			\]
			where $\e^\dag$ denotes the dual super-operator of $\e$; that is, $\tr\left(\e^\dag(M)\rho\right) = \tr\left(M\e(\rho)\right)$ for all $M$ and $\rho$.
		\end{lemma}

		\section{Weakest liberal precondition and strongest postcondition}	\label{sec:wlpsp}
		
		We have defined the correctness of a quantum program in terms of a Hoare triple and demonstrated by an example how to prove it using the denotational semantics of our target language. However, in practical verification scenarios, we often encounter situations where we are provided with only a precondition $P$ or a postcondition $Q$ (but not both). In such cases, our aim is to compute the \emph{strongest postcondition} of $P$ or the \emph{weakest precondition} of $Q$ such that $\models\ass{P}{S}{Q}$ holds. This section is dedicated to addressing these scenarios. Moreover, the results obtained lay the groundwork for the refinement calculus that we will be developing in the following two sections.

		\begin{definition}\label{def:wlp}
			Given a quantum program $S$ and a postcondition $Q$, the \emph{weakest liberal precondition} of $Q$ with respect to $S$, if it exists, is the (unique) quantum assertion $wlp.S.Q$ such that (1) $\models\ass{wlp.S.Q}{S}{Q}$; and (2) for any $P$ such that $\models\ass{P}{S}{Q}$, it holds $P\le wlp.S.Q$.
		\end{definition}
		Intuitively, the first clause requires that $wlp.S.Q$ is indeed a precondition of $Q$ with respect to $S$, while the second clause stipulates that it is the weakest according to the partial order $\le$. It is easy to check from Lemma~\ref{lem:equivpcorrectness} that in our setting $wlp.S.Q$ exists for all programs $S$ and postconditions $Q$; specifically,
		\begin{equation}\label{eq:wlpabstract}
			wlp.S.Q = \bigwedge \left\{\mathcal{N}(\e^\dag(Q^\bot)): \e\in \sem{S}\right\}.
		\end{equation}
		The next lemma gives the explicit form of the weakest liberal precondition for each construct that appears in our target language. 
		\begin{lemma}\label{lem:wlp} The weakest liberal preconditions of quantum programs in our language can be given inductively as follows. For any projector $R$,
			\begin{enumerate}
				\item $wlp.S.I = I$, $wlp.\abort.R =I$, and $wlp.\sskip.R  = R$; 
				\item $wlp.(\bar{q}:=0).R  = E(\<0|_{\bar{q}}R|0\>_{\bar{q}})$ where $E(M)$ denotes the eigenspace of $M$ associated with eigenvalue 1 (we let $E(M)=0$, the null space, if $M$ does not have eigenvalue 1);
				\item $wlp.(\bar{q}\apply U).R  =  U^\dag_{\bar{q}}RU_{\bar{q}}$;
				\item $wlp.(\assert P[\bar{q}]).R  = P^\bot\vee (P\wedge R)= P \rightsquigarrow R\  \mbox{(the Sasaki implication)} $;
				\item $wlp.[P,Q]_{\bar{q}}.R  = \begin{cases}
					I & \mbox{ if } R = I\\
					P & \mbox{ if } Q\le R \mbox{ and } R\neq I\\
					0 & \mbox{o.w. } 
				\end{cases}$
				\item $wlp.(S_0\pcom{p} S_1).R  = wlp.S_0.R \wedge wlp.S_1.R\	\mbox{where} \ 0<p<1$;
				\item $wlp.(S_0;S_1).R  =  wlp.S_0.(wlp.S_1.R)$;
				\item $wlp.(\pmstm).R = (P \rightsquigarrow wlp.S_1.R)\wedge (P^\bot \rightsquigarrow wlp.S_0.R)$;
				\item\label{lemitem:wlprn} $wlp.(\pwstm).R  = \bigwedge_{n\geq 0}R_n$ where $R_0 = I$ and $$R_{n+1}  =  \left(P\rightsquigarrow wlp.S.R_n\right) \wedge\left(P^\bot\rightsquigarrow R\right).$$			\end{enumerate}
		\end{lemma}
		
		Note that every quantum state satisfies the predicate $I$, whereas the quantum state 0 satisfies every predicate $R$. Therefore, $wlp.\abort.R =I$ implies that the $\abort$ program can establish any postcondition $R$ starting from any initial state. This aligns with the concept of \emph{partial} correctness. To understand $wlp.(\bar{q}:=0).R  = E(\<0|_{\bar{q}}R|0\>_{\bar{q}})$, let us examine two simple examples. First, let $R=|0\>_{\bar{q}}\<0|$. Then $E(\<0|_{\bar{q}}R|0\>_{\bar{q}})=I$, indicating that from \emph{any} initial state, the final state of $\bar{q}$ after execution of $\bar{q}:=0$ is $|0\>$. Next, let $R=|+\>_{\bar{q}}\<+|$. Then $E(\<0|_{\bar{q}}R|0\>_{\bar{q}})=E(\frac{1}{2})=0$, indicating that from \emph{no} initial state, the final state of $\bar{q}$ is $|+\>$.

		The more intriguing case arises with the $\assert$statements, where we discover rather unexpected applications of the Sasaki implication. This contrasts to the role played by the ordinary implication in classical programs, where $wlp.(\assert{B}).p = \neg B\vee p = B\ra p$ for any Boolean guard $B$ and classical predicate $p$.
		Note that for any pure state $|\psi\>$, $|\psi\>\<\psi|\models P^\bot\vee (P\wedge R)$ if and only if $|\psi\> = |\psi_1\> + |\psi_2\>$ for some $|\psi_1\>$ and $|\psi_2\>$ such that (1) $P|\psi_1\> = 0$ and (2) $|\psi_2\> = P|\psi_2\> \in R$.
		The second part, $|\psi_2\>$, of $|\psi\>$ undergoes the program $\assert P[\bar{q}]$ without any change ($|\psi_2\> = P|\psi_2\>$), and satisfies the postcondition ($P|\psi_2\>\in R$). As for the first part, $|\psi_1\>$, the program $\assert P[\bar{q}]$ yields no results, thus the postcondition is also fulfilled according to the definition of partial correctness. 
		
		For $wlp.[P,Q]_{\bar{q}}.R$, the cases where $R=I$ or $Q\le R$ but $R\neq I$ are easy to understand. When $Q\not \le R$, first note that we can always find a quantum program $S$ in $\sem{[P,Q]_{\bar{q}}}$ that fails to establish the postcondition $R$. This program can, for example, transform any proper initial state into a pure state that is in $Q$ but not in $R$. Consequently, for \emph{no} initial state, all final states after executing a program in $\sem{[P,Q]_{\bar{q}}}$, including $S$, lie in $R$.
		Therefore, we are forced to set $wlp.[P,Q]_{\bar{q}}.R=0$ in this case.
		
		Recall that for classical branching construct: 
		\[wlp.(\iif\ B\ \then\ S_1\ \then\ S_0\ \pend).p = \left(B\ra wlp.S_1.p\right)\wedge \left(\neg B\ra wlp.S_0.p\right).\]
		The weakest liberal precondition for quantum conditional branching has a similar structure (and also a similar explanation). However, the implications inside the two conjuncts have to be replaced by the Sasaki implication, because branching in quantum programs is caused by measurements (or, by $\assert$statements). In fact, it can be easily checked that 
		\[
		\sem{\pmstm} = \sem{\assert{P[\bar{q}]}; S_1} + \sem{\assert{P^\bot[\bar{q}]}; S_0}.
		\]
		On the other hand, one may wonder why the conjunction is not replaced by the Sasaki conjunction $\doublecap$. The reason is that the quantum programs in our language follow a `quantum data, classical control' paradigm~\cite{selinger2004towards}. 
		This means that although different branches are caused by quantum measurements, the way in which they are combined to form the semantics of the entire program remains classical. Thus, we use the \emph{classical conjunction} $\wedge$ to connect the two subformulas $P \rightsquigarrow wlp.S_1.R$ and $P^\bot \rightsquigarrow wlp.S_0.R$ obtained from the individual branches.

		A dual notion of the \emph{strongest postcondition} can also be defined for quantum programs in our language.
		
		\begin{definition}
			Given a quantum program $S$ and a quantum assertion $P$, the strongest postcondition of $P$ with respect to $S$, if it exists, is the (unique) quantum assertion $sp.S.P$ such that (1) it is a valid postcondition: $\models\ass{P}{S}{sp.S.P}$; and (2) it is the strongest: for any $Q$ such that $\models\ass{P}{S}{Q}$, it holds $sp.S.P \le Q$.
		\end{definition}
		
		Again, it is easy to check from Lemma~\ref{lem:equivpcorrectness} that in our setting, $sp.S.P$ does exist for all programs $S$ and preconditions $P$; specifically,
		\begin{equation}\label{eq:spcabstract}
			sp.S.P = \bigvee \left\{\supp{\e(P)}: \e\in \sem{S}\right\}.
		\end{equation}
		The next lemma gives the explicit form of the strongest postcondition semantics for each construct that appears in our target language. 
		\begin{lemma}\label{lem:spost} The strongest postconditions of quantum program in our language can be given inductively as follows. For any projector $R$,
			\begin{enumerate}
				\item $sp.S.0 = 0$, $sp.\abort.R =0$, and $sp.\sskip.R=R$;
				\item $sp.(\bar{q}:=0).R  = |0\>_{\bar{q}}\<0|\otimes \supp{\tr_{\bar{q}}(R)}$;
				\item $sp.(\bar{q}\apply U).R  =  U_{\bar{q}}RU^\dag_{\bar{q}}$;
				\item $sp.(\assert P[\bar{q}]).R  = P\wedge (P^\bot\vee R) = P \doublecap R\ \mbox{(the Sasaki conjunction)} $;
				\item $sp.[P,Q]_{\bar{q}}.R  = \begin{cases}
					0 & \mbox{ if } R = 0\\
					Q & \mbox{ if } R\le P \mbox{ and } R\neq 0\\
					I & \mbox{o.w. } 
				\end{cases}$
				\item $sp.(S_0\pcom{p} S_1).R  =  sp.S_0.R \vee  sp.S_1.R$ where $0<p<1$;
				\item $sp.(S_0;S_1).R  =  sp.S_1.(sp.S_0.R)$;
				\item $sp.(\pmstm).R = (sp.S_1.(P\doublecap R)) \vee (sp.S_0.(P^\bot\doublecap R))$;
				\item $sp.(\pwstm).R  = \bigvee_{n\geq 0}  \left(P^\bot\doublecap R_n\right)$ where $R_0=0$ and
				$R_{n+1} = R \vee sp.S.(P\doublecap R_n).$
			\end{enumerate}
		\end{lemma}
		
		The strongest postcondition semantics is relatively easier to comprehend compared to the weakest precondition semantics, as it aligns with the forward manner of program execution. Notably, the Sasaki conjunction assumes a significant role in determining the strongest postconditions, akin to the role played by the Sasaki implication in the weakest preconditions. On the other hand, thanks to the Sasaki conjunction, the strongest postconditions for quantum programs bear a striking resemblance to those for classical programs, with the ordinary conjunctions being replaced by the Sasaki conjunctions.

		From Lemmas~\ref{lem:wlp} and~\ref{lem:spost}, or more abstractly, from Eqs.~\eqref{eq:wlpabstract} and~\eqref{eq:spcabstract}, it is easy to observe that both $wlp.S$ and $sp.S$ are Scott-continuous. In particular, they are monotonic with respect to $\le$.

	\section{Refinement calculus for partial correctness}
        \label{sec:calculus}
		
		In this section, we turn to construct a refinement calculus to facilitate stepwise development of quantum programs. The soundness of the refinement rules relies on the pre/post-condition semantics developed in the previous section.

		For any quantum programs $S$ and $S'$, we say that $S$ is refined by $S'$ in partial correctness, denoted $S\le S'$, if for any projectors $P$ and $Q$, $\models \ass{P}{S}{Q}$ implies $\models \ass{P}{S'}{Q}$.
		Intuitively, this means that $S'$ satisfies all the correctness specifications that $S$ satisfies. In simpler terms, $S'$ maintains the correctness of $S$. Notably, if $S = [P,Q]_{\bar{q}}$, then the relation $S\le S'$ implies that $S'$, which only operates on qubits in $\bar{q}$, achieves the postcondition $Q$ if starting with precondition $P$. Let $\equiv$ be the equivalence relation induced by $\le$. 
		The following theorem establishes a close connection between refinement relation and the concepts of weakest liberal preconditions and strongest postconditions.
		\begin{theorem}\label{thm:refinewlpspc} For any quantum programs $S$ and $S'$, 
			\[
			S\le S' \qquad \mbox{iff} \qquad \forall Q,\ wlp.S.Q\le wlp.S'.Q\qquad \mbox{iff} \qquad \forall P,\ sp.S.P\ge sp.S'.P.
			\]
		\end{theorem}
		
		To facilitate refinement reasoning, we first present a set of structural rules in Fig.~\ref{fig:pstructurerules} demonstrating that the refinement relation forms a preorder and is preserved by all program constructs (thus, the induced equivalence $\equiv$ also maintains this property). More rules will be presented in subsequent sections.

		{
			\renewcommand{\arraystretch}{1.3}
			\begin{figure}
				\footnotesize
				\begin{flushleft}\fbox{\textsf{Structural rules}}\end{flushleft}
				\vspace{1em}
				\begin{tabular}{llll}
					\textsc{(Stru-Ref)} & $S \equiv S$
					\qquad &
					\textsc{(Stru-Tran)} & \begin{tabular}{l}$
						\inferrule* [rightstyle={\footnotesize \sc}]{
							S_0 \le S_1
							\qquad 
							S_1\le S_2
						}{
							S_0 \le S_2
						}
						$	\end{tabular}\vspace{0.5em}
					\\ 
					\textsc{(Stru-Seq)} & \begin{tabular}{l}$
						\inferrule* [rightstyle={\footnotesize \sc}]{
							S_0 \le T_0\qquad S_1\le T_1
						}{
							S_0;S_1 \le T_0;T_1
						}
						$	\end{tabular}							
					\qquad &
					\textsc{(Stru-Cond)} & \begin{tabular}{l}$
						\inferrule* [rightstyle={\footnotesize \sc}]{
							S_0 \le T_0\qquad S_1\le T_1
						}{
							\pmstm \le \con{P[\bar{q}]}{T_1}{T_0}
						}
						$	\end{tabular}
					\vspace{0.5em}
					\\ 
					\textsc{(Stru-Pcho)} & \begin{tabular}{l}$
						\inferrule* [rightstyle={\footnotesize \sc}]{
							S_0 \le T_0\qquad S_1\le T_1
						}{
							S_0\pcom{p} S_1\le T_0\pcom{p} T_1
						}
						$	\end{tabular}
					\qquad &
					\textsc{(Stru-While)} & \begin{tabular}{l}$
						\inferrule* [rightstyle={\footnotesize \sc}]{
							S \le T
						}{
							\whilestm{P[\bar{q}]}{S} \le \whilestm{P[\bar{q}]}{T}
						}
						$	\end{tabular}
				\end{tabular}
				\caption{Structure rules for the refinement relation}
				\label{fig:pstructurerules}
			\end{figure}
		}

		\subsection{Refinement rules for prescriptions}
		
		This section is devoted to refinement rules where prescriptions are refined by different program constructs in which the subprograms may also include prescriptions. We first present a general condition for a program to refine a prescription, which can be proven directly from Theorem~\ref{thm:refinewlpspc}.
		
		\begin{theorem}\label{thm:generalref} 
			For any projectors $P, Q$ and quantum program $S$,
			\[
			[P,Q]\le S\quad \mbox{iff}\quad  P \le wlp.S.Q \quad \mbox{iff}\quad sp.S.P \le Q.
			\]		
		\end{theorem}
		
		Given that we consider both pre- and postcondition semantics in this paper, there are naturally two types of refinement rules. One is based on weakest liberal preconditions, while the other is based on strongest postconditions. These rules are presented in Fig.~\ref{fig:prulespres}. We say $\vdash_{\textsc{Wpc}}
		[P,Q]_{\bar{q}}\le S$ if the refinement relation $[P,Q]_{\bar{q}}\le S$ can be derived from the structural rules in Fig.~\ref{fig:pstructurerules}, the common rules, and the rules based on the weakest liberal preconditions in Fig.~\ref{fig:prulespres}.
		Similarly, let $\vdash_{\textsc{Spc}}
		[P,Q]_{\bar{q}}\le S$ if the refinement relation $[P,Q]_{\bar{q}}\le S$ can be derived from the structural rules in Fig.~\ref{fig:pstructurerules}, the common rules, and the rules based on the strongest postconditions in Fig.~\ref{fig:prulespres}. To help understand how to facilitate the refinement rules, let us examine some examples.
		
		{
			\renewcommand{\arraystretch}{1.2}
			\begin{figure}
				\footnotesize
				\begin{flushleft}\fbox{\textsf{Common rules}}\end{flushleft}
				\smallskip
				\begin{tabular}{llcl}
					\textsc{(Comm-Pres-All)} & $[0,Q]_{\bar{q}}$ & \le & $S$ \qquad\qquad\qquad $\ \ [P,I]_{\bar{q}}\quad  \le \quad S$
					\\
					\textsc{(Comm-Pres-Abort)} & $[P,Q]_{\bar{q}}$ & \le & $ \abort	$
					\\     
					\textsc{(Comm-Pres-Skip)} & $[P,P]_{\bar{q}}$ & \le & $ \sskip$
					\\
					\textsc{(Comm-Pres-Cons)} & $[P,Q]_{\bar{q}}$ & \le & $ [R,T]_{\bar{q}}$\quad if $P\le R$ and $T\le Q$
					\\
					\textsc{(Comm-Pres-Seq)} & $[P,Q]_{\bar{q}}$ & \le & $ [P,R]_{\bar{q}};[R,Q]_{\bar{q}}$
				\end{tabular}
				\vspace{1em}
				\begin{flushleft}\fbox{\textsf{Rules based on weakest liberal preconditions}}\end{flushleft}
				\vspace{1em}
				\begin{tabular}{llcl}
					\textsc{(Wpc-Pres-Init)} & $[E(\<0|_{\bar{q}} Q|0\>_{\bar{q}}),Q]_{\bar{q}}$ & \le & $\bar{q}:=0$
					\\
					\textsc{(Wpc-Pres-Unit)} & $[U_{\bar{q}}^\dag Q U_{\bar{q}},Q]_{\bar{q}}$& \le & $\bar{q}\apply U$
					\\
					\textsc{(Wpc-Pres-Ass)} & $[P\rightsquigarrow Q, Q]_{\bar{q}} $ & \le & $ \assert P[\bar{q}]$
					\\
					\textsc{(Wpc-Pres-Pcho)} &  $[P,Q]_{\bar{q}}$ & $\equiv$ & $ [P_1,Q]_{\bar{q}}\pcom{p}[P_2,Q]_{\bar{q}}$ if $0<p<1$ and $P= P_1\wedge P_2$
					\\
					\textsc{(Wpc-Pres-Cond)} &  $[(P\rightsquigarrow P_1) \wedge (P^\bot \rightsquigarrow P_0),Q]_{\bar{q}}$ & $\equiv$ & $\iif\ P[\bar{q}]\ \then\ [P_1, Q]_{\bar{q}}\ \eelse\ [P_0,Q]_{\bar{q}}\ \pend$
					\\
					\textsc{(Wpc-Pres-While)} & $[(P\rightsquigarrow Q) \wedge (P^\bot \rightsquigarrow R),R]_{\bar{q}}$ & \le & $ \while\ P[\bar{q}]\ \ddo\ [Q, (P\rightsquigarrow Q) \wedge (P^\bot \rightsquigarrow R)]_{\bar{q}}\ \pend$
				\end{tabular}
				\vspace{1em}
				\begin{flushleft}\fbox{\textsf{Rules based on strongest postconditions}}\end{flushleft}
				\vspace{1em}
				\begin{tabular}{llcl}
					\textsc{(Spc-Pres-Init)} & $[P,|0\>_{\bar{q}}\<0|\otimes \supp{\tr_{\bar{q}}(P)}]_{\bar{q}}$ & \le & $  \bar{q}:=0$
					\\ 
					\textsc{(Spc-Pres-Unit)} & $[P, U_{\bar{q}} P U^\dag]_{\bar{q}}$ & \le & $  \bar{q}\apply U$
					\\
					\textsc{(Spc-Pres-Ass)} & $[P, Q\doublecap P]_{\bar{q}} $ & \le & $  \assert Q[\bar{q}]$
					\\
					\textsc{(Spc-Pres-Pcho)} &  $[P,Q]_{\bar{q}}$ & $\equiv$ & $ [P,Q_1]_{\bar{q}}\pcom{p}[P,Q_2]_{\bar{q}}$ if $0<p<1$ and $Q_1\vee Q_2= Q$
					\\
					\textsc{(Spc-Pres-Cond)} &  $[P,Q]_{\bar{q}}$
					& $\le$ & $  \iif\ R[\bar{q}]\ \then\ [R\doublecap P, Q]_{\bar{q}}\ \eelse\ [R^\bot\doublecap P, Q]_{\bar{q}}\ \pend$
					\\
					\textsc{(Spc-Pres-While)} & $[Inv, P^\bot \doublecap Inv]_{\bar{q}}$ & \le & $  \while\ P[\bar{q}]\ \ddo\ [P\doublecap Inv, Inv]_{\bar{q}}\ \pend$
				\end{tabular}
				\smallskip
				\caption{Refinement rules for prescriptions. }
				\label{fig:prulespres}
			\end{figure}
		}

		\begin{example}\label{exam:set1}
			Suppose that our goal is to initialize a qubit $q$ in $|1\>$; that is, our specification is $[I, |1\>_q\<1|]_q$. We have at least three different approaches to achieve this, depending on which program construct we choose to use first.
			\begin{enumerate}
				\item \emph{Sequential composition}. 				
				We can utilize rule (\textsc{Comm-Pres-Seq}) to break down the specification into the sequential composition of two sub-ones. For example, we derive from this rule that
				$
				[I, |1\>_q\<1|]_q \le [I, |0\>_q\<0|]_q; [|0\>_q\<0|, |1\>_q\<1|]_q.
				$
				For the first part, (\textsc{Wpc-Pres-Init}) or (\textsc{Spc-Pres-Init}) tells us 
				$[I, |0\>_q\<0|]_q \le q:=0$. For the second part, we use either (\textsc{Wpc-Pres-Unit}) or (\textsc{Spc-Pres-Unit}) to derive
				$[|0\>_q\<0|, |1\>_q\<1|]_q \le q \apply X$. Finally, we have
				\[
				[I, |1\>_q\<1|]_q \le q:=0; q\apply X
				\]
				as desired, employing \textsc{(Stru-Seq)}.
				
				\item \emph{Conditional branching}. Note that 
				$
				\left(P^0 \rightsquigarrow |0\>\<0|\right) \wedge  \left(P^{\neq 0} \rightsquigarrow |1\>\<1|\right)= I
				$ where $P^0 \define |0\>\<0|$ and $P^{\neq 0} \define |1\>\<1|$.
				Then we have 
				$$
				[I, |1\>_q\<1|]_q \le \iif\ P^0[q]\ \then\ [|0\>_q\<0|, |1\>_q\<1|]_q\ \eelse\ [|1\>_q\<1|, |1\>_q\<1|]_q\ \pend
				$$
				from (\textsc{Wpc-Pres-Cond}), and so 
				\[
				[I, |1\>_q\<1|]_q \le \iif\ P^0[q]\ \then\ q\apply X\ \eelse\ \sskip\ \pend
				\]
				from (\textsc{Wpc-Pres-Unit}), (\textsc{Comm-Pres-Skip}), and  (\textsc{Stru-Cond}).
				
				Similarly, we can obtain the same implementation using (\textsc{Spc-Pres-Cond}), (\textsc{Spc-Pres-Unit}),  (\textsc{Comm-Pres-Skip}), and (\textsc{Stru-Cond}).
				
				\item \emph{While loop}. Note that
				$
				\left(P^0 \rightsquigarrow I\right) \wedge  \left(P^{\neq 0} \rightsquigarrow |1\>\<1|\right) = I$. Then we have 
				$
				[I, |1\>_q\<1|]_q \le \while\ \allowbreak P^0[q]\ \ddo\ [I, I]_q\ \pend$
				from (\textsc{Wpc-Pres-While}), and so 
				\[
				[I, |1\>_q\<1|]_q \le \while\ P^0[q]\ \ddo\ q\apply H\ \pend
				\]
				from (\textsc{Wpc-Pres-Unit}) and (\textsc{Stru-While}). Note that we can refine $[I, I]_q$ with $q\apply U$ for any single-qubit unitary $U$, or even any single-qubit quantum program $S$. This is due to the fact that only \emph{partial} correctness is considered here, and termination is not our concern.
				
				Again, we can derive the same implementation using (\textsc{Spc-Pres-While}), (\textsc{Spc-Pres-Unit}), and (\textsc{Stru-While}).
			\end{enumerate}
		\end{example}
		
		\begin{theorem}\label{thm:psoundpre}
			All the proof rules presented in Figs.~\ref{fig:pstructurerules} and~\ref{fig:prulespres} are sound; that is,
			\[
			\mbox{If}\quad \vdash_{\textsc{Wpc}}
			[P,Q]_{\bar{q}}\le S\quad \mbox{or}\quad \vdash_{\textsc{Spc}}
			[P,Q]_{\bar{q}}\le S\quad \mbox{then}\quad \models \ass{P}{S}{Q}. 
			\]
		\end{theorem}
		
		Conversely, both types of refinement rules are (relatively) complete, which means that if a Hoare triple $\ass{P}{S}{Q}$ with $\qv(S) = \bar{q}$ is semantically correct, then the program $S$, whether or not it is fully executable, can be obtained from the specification $[P,Q]_{\bar{q}}$ by applying rules from Fig.~\ref{fig:pstructurerules} and either type of refinement rules from Fig.~\ref{fig:prulespres}.
		\begin{theorem}\label{thm:pcompletepre}
			The proof rules presented in Figs.~\ref{fig:pstructurerules} and~\ref{fig:prulespres} are relatively complete; that is,
			\[
			\mbox{If}\quad \models \ass{P}{S}{Q} \quad \mbox{then}\quad \vdash_{\textsc{Wpc}}
			[P,Q]_{\bar{q}}\le S\quad \mbox{and}\quad \vdash_{\textsc{Spc}}
			[P,Q]_{\bar{q}}\le S. 
			\]
		\end{theorem}

		\subsection{Refinement rules for assertion statements}
		
		This section is dedicated to refinement rules where assertion statements are involved, motivated by~\cite{back1988calculus}. For simplicity, we write $\{P\}$ for $\assert P$, following the notation from~\cite{back1988calculus}.
		Similarly to the previous section, we first present a general theorem. 
		\begin{theorem}\label{thm:generalassert} 
			For any projectors $P, Q$ and quantum program $S$,
			\begin{align*}
				\{P\}; S \equiv \{P\}; S; \{Q\} & \quad \mbox{iff}\quad sp.S.P \le Q \ \qquad \mbox{iff}\quad P \le wlp.S.Q\qquad \mbox{iff}\quad \models \ass{P}{S}{Q}\\
				S; \{Q\} \equiv \{P\}; S; \{Q\}& \quad \mbox{iff}\quad sp.S.P^\bot \le Q^\bot  \quad \mbox{iff}\quad P^\bot \le wlp.S.Q^\bot\quad \mbox{iff}\quad \models \ass{P^\bot}{S}{Q^\bot}.
			\end{align*}
		\end{theorem}
		
		A direct corollary of the above theorem is that $\{P\};  S \equiv \{P\}; S; \{sp.S.P\}$.
		However, the seemingly corresponding statement $S; \{Q\} \equiv \{wlp.S.Q\}; S; \{Q\}$ for $wlp$ is not true. To see this, let $Q\not\in \{0, I\}$ and $\sem{S} = \{Set_\psi\}$ where $Set_\psi$ sets the quantum system in a normalized state $|\psi\>$ neither in $Q$ nor in $Q^\bot$. Then $wlp.(S; \{Q\}).0 = wlp.S.Q^\bot =0$ while $$wlp.\left(\{wlp.S.Q\}; S; \{Q\}\right).0 = wlp.S.Q \rightsquigarrow wlp.S.Q^\bot = 0\rightsquigarrow 0 = I. $$
		In fact, from Theorem~\ref{thm:generalassert}, the correct equivalence reads as follows:
		$$S; \{Q\} \equiv \left\{\left(wlp.S.Q^\bot\right)^\bot\right\}; S; \{Q\}.$$
		
	\begin{figure}
			\footnotesize
			{
				\renewcommand{\arraystretch}{1.5}
				\begin{flushleft}\fbox{\textsf{Common rules}}\end{flushleft}
				\begin{tabular}{ll}
					\textsc{(Comm-Asst-True)} & $S\equiv \{I\}; S\equiv S; \{I\}$ \vspace{0.5em}\\
					\textsc{(Comm-Asst-Post)} & $\{P\}; S \equiv \{P\}; S; \{Q\}$\quad if \quad $sp.S.P \le Q$ \quad or\quad $P \le wlp.S.Q$\vspace{0.5em}\\					
					\textsc{(Comm-Asst-Pre)} & $S; \{Q\} \equiv \{P\}; S; \{Q\}$\quad if \quad $sp.S.P^\bot \le Q^\bot$\quad or\quad $P^\bot \le wlp.S.Q^\bot$ \vspace{0.5em}\\										
					\textsc{(Comm-Asst-Seq)} &  \begin{tabular}{l}$
						\inferrule* [rightstyle={\footnotesize \sc}]{
							\{P\};  S_1 \equiv \{P\}; S_1; \{Q\}\qquad \{Q\};  S_2 \equiv \{Q\}; S_2; \{R\}
						}{
							\{P\};  S_1 ; S_2  \equiv \{P\}; S_1; \{Q\}; S_2; \{R\} 
						}$	\end{tabular}
				\end{tabular}
				\vspace{1em}
				\begin{flushleft}\fbox{\textsf{Rules based on weakest liberal preconditions}}\end{flushleft}
				\vspace{1em}
				\begin{tabular}{ll}
					\textsc{(Wpc-Asst-Pres)} & $\{P\}; [Q,R] \equiv \{P\}; [P\rightsquigarrow Q, R]$ \vspace{0.5em}\\
										\textsc{(Wpc-Asst-Cond)} & {\scriptsize \begin{tabular}{l}$
												\inferrule* [rightstyle={\footnotesize \sc}]{
														S_1; \{R\} \equiv \{P\rightsquigarrow Q\}; S_1; \{R\}
														\qquad 
														S_0; \{R\}\equiv \{P^\bot\rightsquigarrow Q\}; S_0; \{R\}
													}{
														\iif\ P[\bar{q}]\ \then\ S_1\ \eelse\ S_0\ \pend; \{R\}
														\equiv \{Q\}; \iif\ P[\bar{q}]\ \then\ \{P\rightsquigarrow Q\}; S_1\ \eelse\ \{P^
														\bot\rightsquigarrow Q\}; S_0\ \pend; \{R\}
													}
												$	\end{tabular} }\vspace{0.5em}
										\\ 
										\textsc{(Wpc-Asst-While)} & { \begin{tabular}{l}$
												\inferrule* [rightstyle={\footnotesize \sc}]{
														S; \{Q\} \equiv \{P\rightsquigarrow Q\}; S; \{Q\}
													}{
														\pwstm; \{P^\bot \rightsquigarrow Q\}
														\equiv \{Q\}; \while\ P[\bar{q}]\ \ddo\ \{P\rightsquigarrow Q\}; S; \{Q\}\ \pend; \{P^\bot \rightsquigarrow Q\}
													}
												$ 	\end{tabular}}
				\end{tabular}
				\vspace{1em}
				\begin{flushleft}\fbox{\textsf{Rules based on strongest postconditions}}\end{flushleft}
				
				\vspace{1em}
				
				\begin{tabular}{ll}
					\textsc{(Spc-Asst-Post)} & $	\{P\}; \{Q\} \equiv \{P\}; \{Q\};\{Q\doublecap P\} 
					\equiv  \{P\doublecap Q\}; \{P\}; \{Q\}$\\
					\textsc{(Spc-Asst-Pres)} & $[Q,R]; \{P\} \equiv [Q, P\doublecap R]; \{P\}$\vspace{0.5em}\\
					\textsc{(Spc-Asst-Cond)}& \hspace{-1em} \begin{tabular}{l}$
						\inferrule* [rightstyle={\footnotesize \sc}]{
							\{P\doublecap Q\};  S_1 \equiv \{P\doublecap Q\}; S_1; \{R\}
							\qquad 
							\{P^\bot\doublecap Q\};  S_0 \equiv \{P^\bot\doublecap Q\}; S_0; \{R\}
						}{
							\{Q\}; \iif\ P[\bar{q}]\ \then\ S_1\ \eelse\ S_0\ \pend
							\equiv \{Q\}; \iif\ P[\bar{q}]\ \then\ \{P\doublecap Q\}; S_1\ \eelse\ \{P^
							\bot\doublecap Q\}; S_0\ \pend; \{R\}
						}
						$	\end{tabular}\vspace{0.5em}
					\\ 
					\textsc{(Spc-Asst-While)} &  \begin{tabular}{l}$
						\inferrule* [rightstyle={\footnotesize \sc}]{
							\{P\doublecap Inv\};  S \equiv \{P\doublecap Inv\}; S; \{Inv\}
						}{
							\{Inv\}; \pwstm
							\equiv\  \{Inv\}; \while\ P[\bar{q}]\ \ddo\ \{P\doublecap Inv\}; S; \{Inv\}\ \pend; \{P^\bot\doublecap Inv\}
						}
						$ 	\end{tabular}
				\end{tabular}
			}
			\caption{Refinement rules for assertion statements}
			\label{fig:prulesass}
		\end{figure}
		As in the previous section, there are naturally two types of refinement rules presented in Fig.~\ref{fig:prulesass}: one based on the weakest liberal preconditions and the other on strongest postconditions. Note that all the rules in Fig.~\ref{fig:prulesass} are expressed in terms of the equivalence relation $\equiv$, meaning that appropriate assertions can be freely added or removed between subprograms as long as the rules are adhered to. This can be particularly useful in scenarios such as projection-based run-time testing and debugging of quantum programs~\cite{li2020projection}. Specifically, let us revisit the loop program $S\define \while\ P^0[q]\ \ddo\ q\apply H\ \pend$ in Example~\ref{exam:set1}. Note that $P^0\doublecap I = P^0$ and $\{P^0\}; q\apply H \equiv \{P^0\}; q\apply H; \{I\}$ from rule \textsc{(Comm-Asst-Post)}. Thus, 
		\begin{align*}
			S \equiv \{I\}; S
			&\equiv \{I\}; \while\ P^0[q]\ \ddo\ \{P^0\}; q\apply H; \{I\}\ \pend; \{|1\>_q\<1|\}\\
			&\equiv \{I\}; \while\ P^0[q]\ \ddo\ \{P^0\}; q\apply H\ \pend; \{|1\>_q\<1|\}
		\end{align*}
		where the first and third equivalences follow from \textsc{(Comm-Asst-True)} while the second from \textsc{(Spc-Asst-While)}.

		\begin{theorem}\label{thm:psoundass}
			All the proof rules presented in Fig.~\ref{fig:prulesass} are sound.
		\end{theorem}
		
		To conclude this section, we remark that some refinement rules applicable to classical programs do not hold in the quantum setting. A notable example is as follows: in classical programming, if $p \rightarrow p'$, then $\{p\} \le \{p'\}$ where $p$ and $p'$ are classical predicates (first-order logic formulas)~\cite{back1988calculus}. This classical rule allows the elimination of assertions from a program; that is, $S[\{p\}] \le S$, since $p \rightarrow \top$ and $\{\top\} \equiv \sskip$. However, this rule is no longer valid for quantum programs. To illustrate this, note that 
		$
		wlp.\{|0\>_q\<0|\}.|+\>_q\<+|  = |1\>_q\<1|
		$ and
		$wlp.\{I\}.|+\>_q\<+|  = |+\>_q\<+|.
		$
		Since $|1\>_q\<1|\not \le |+\>_q\<+|$, we have $\{|0\>_q\<0|\}\not \le \{I\}$ from Theorem~\ref{thm:refinewlpspc}.

	\section{Case studies}
		\label{sec:case}
  
          To illustrate the effectiveness of the refinement rules proposed in this paper, we employ them to implement a $Z$-rotation gate, the repetition code, and the quantum-to-quantum Bernoulli factory.

	\subsection{Implementing a Z-rotation gate}\label{subsec:zrotation}
        
        We have shown in Examples~\ref{exam:pcircuit} and~\ref{exam:pcircuitcor} that the circuit in Fig.~\ref{fig-rz} (or the program $\mathit{pCircuit}$) either implements $R_z(\theta)$ (when both measurement outcomes are 0) or does nothing (otherwise). It is then evident that the following program
		\begin{align*}
			\mathit{Rz}\quad\define\quad & q_0, q_1:=0;\ q_0\apply X;\\
			& \while\ P^{\neq 00}[q_0,q_1]\ \ddo\\
			& \quad q_0, q_1:=0;\\
			& \quad \mathit{pCircuit}\\
			&\pend
		\end{align*}
	    using $\mathit{pCircuit}$ as a subprogram, will realize $R_z(\theta)$ whenever it terminates.
		
		In the following, we will see how the refinement calculus presented in this paper can be used to derive $\mathit{Rz}$ in a stepwise manner (so that the correctness is guaranteed). We start with the specification (prescription) $S_0 \define \left[\Omega_{t,t'}, \r_z(\theta)_t(\Omega_{t,t'})\right]$ and try to refine it until we get $\mathit{Rz}$. The process consists of several steps.
		
    For the first step, from rules \textsc{(Comm-Pres-Seq)}, \textsc{(Spc-Pres-Init)}, and \textsc{(Comm-Pres-Cons)} we have
    \begin{align*}
        S_0 &\le \left[\Omega_{t,t'}, P^{\neq 00}_{q_0,q_1}\otimes \Omega_{t,t'}\right]; \left[P^{\neq 00}_{q_0,q_1}\otimes \Omega_{t,t'}, \r_z(\theta)_t(\Omega_{t,t'})\right]\\
        & \le q_0, q_1:=0;\ q_0\apply X;\ \left[P^{\neq 00}_{q_0,q_1}\otimes \Omega_{t,t'}, \r_z(\theta)_t(\Omega_{t,t'})\right]
    \end{align*}
    since 
    $
    |0\>_{q_0,q_1}\<0|\otimes \supp{\tr_{q_0, q_1} \Omega_{t,t'}} = P^{00}_{q_0,q_1}\otimes \Omega_{t,t'}
    $
    while
    \[
    X_{q_0}\left(P^{00}_{q_0,q_1}\otimes \Omega_{t,t'}\right)X_{q_0}^\dag = |10\>_{q_0,q_1}\<10|\otimes \Omega_{t,t'} \le P^{\neq 00}_{q_0,q_1}\otimes \Omega_{t,t'}.
    \]
	
    We now introduce the while loop. For this purpose, we take
			\[
			Inv \define \left(P^{\neq 00}_{q_0,q_1}\otimes \Omega_{t,t'}\right) \vee \left(P^{00}_{q_0,q_1}\otimes \r_z(\theta)_t(\Omega_{t,t'})\right)
			\]
			as the loop invariant which captures the intuition that if the measurement outcomes are both 0, then $\r_z(\theta)$ has been successfully applied on $t$; otherwise, the state of $t$ remains unchanged. Then from rule \textsc{(Spc-Pres-While)},
			\[
			\left[Inv, P^{00}_{q_0,q_1}\doublecap Inv\right] \le \while\ P^{\neq 00}[q_0,q_1]\ \ddo\ \left[P^{\neq 00}_{q_0,q_1}\doublecap Inv, Inv\right]\ \pend.
			\]
			Note that $P^{\neq 00}_{q_0,q_1}\otimes \Omega_{t,t'} \le Inv$ and $P^{00}_{q_0,q_1}\doublecap Inv = P^{00}_{q_0,q_1}\otimes \r_z(\theta)_t(\Omega_{t,t'})\le \r_z(\theta)_t(\Omega_{t,t'}).$
			We have from \textsc{(Comm-Pres-Cons)} that
			\[
			\left[P^{\neq 00}_{q_0,q_1}\otimes \Omega_{t,t'}, \r_z(\theta)_t(\Omega_{t,t'})\right] \le  \while\ P^{\neq 00}[q_0,q_1]\ \ddo\ \left[P^{\neq 00}_{q_0,q_1}\doublecap Inv, Inv\right]\ \pend.
			\]
	
 Next, we refine the loop body $\left[P^{\neq 00}_{q_0,q_1}\doublecap Inv, Inv\right]$. This can be done as follows. 
			\begin{align*}
				\left[P^{\neq 00}_{q_0,q_1}\otimes \Omega_{t,t'}, Inv\right] &\le \left[P^{\neq 00}_{q_0,q_1}\otimes \Omega_{t,t'}, P^{00}_{q_0,q_1}\otimes \Omega_{t,t'}\right]; \left[P^{00}_{q_0,q_1}\otimes \Omega_{t,t'}, Inv\right] & \textsc{(Comm-Pres-Seq)}\\
				&\le q_0,q_1 :=0; \left[P^{00}_{q_0,q_1}\otimes \Omega_{t,t'}, Inv\right] & \textsc{(Spc-Pres-Init)}\\
				&\le q_0,q_1 :=0; \mathit{pCircuit} & \mbox{Eq.~\eqref{eq:corpcircuit} and Theorem~\ref{thm:generalref}}
			\end{align*}

    Putting the above three steps together, we derive $S_0\le \mathit{Rz}$ as desired using the structural rules in Fig.~\ref{fig:pstructurerules}.

        \subsection{Repetition code}
        \label{subsec: rep3}

  	Quantum repetition codes are one of the simplest classes of quantum error correction codes designed to protect quantum information from noise in quantum communication channels or during quantum computation. A typical three-qubit repetition code encodes one logical qubit state $\alpha|0\> + \beta|1\>$ into a three-physical-qubit state $\alpha|000\> + \beta|111\>$ (depicted in the left dashed box in Fig.~\ref{fig-repetition} for a circuit implementation of the encoding process).
		Subsequently, all the three qubits pass through a noisy quantum channel, which flips the qubit (that is, applies a Pauli-$X$ operator on it) nondeterministically.
		For simplicity, we assume that at most one of the three qubits is flipped, without knowing which one. Interestingly, by applying a properly designed error correction procedure, which can be expressed in our target language (refer to the right part of Fig.~\ref{fig-repetition}), the error can be detected and corrected perfectly.

  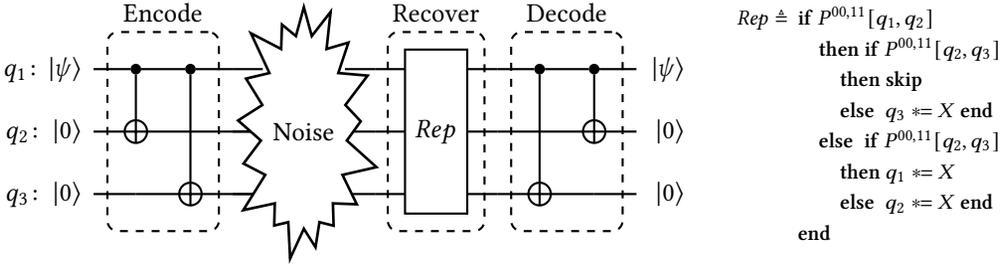
\begin{figure}
    \begin{minipage}[t]{0.7\textwidth}
    \centering
      \tikzset{
            my label/.append style={above right,xshift=0.3cm}
          }
          \begin{quantikz}[row sep=0.4cm,column sep=0.4cm]
            \lstick{$q_1\!:\,  |\psi\>$} & \ctrl{1}\gategroup[3,steps=2,style={dashed,rounded
            corners,inner sep=3pt}]{Encode} & 
            \ctrl{2} & \qw & \gate[3,style={starburst,line
                width=1pt,inner xsep=-4pt,inner ysep=-5pt},
              label style=black]{\text{Noise}} & \qw  & \gate[3,disable auto
              height]{Rep}\gategroup[3,steps=1,style={dashed,rounded
              corners,inner sep=3pt}]{Recover} & \qw & \ctrl{2}\gategroup[3,steps=2,style={dashed,rounded
              corners,inner sep=3pt}]{Decode} & \ctrl{1} & \qw & & \lstick{$|\psi\>$}\\
            \lstick{$q_2\!:\  |0\>$} & \targ{} & \qw     & \qw & \qw & \qw &  & \qw & \qw & \targ{} & \qw & & \lstick{$|0\>$} \\
            \lstick{$q_3\!:\  |0\>$} & \qw     & \targ{} & \qw & \qw & \qw &  & \qw & \targ{} & \qw & \qw & & \lstick{$|0\>$}
            \end{quantikz}   
          \end{minipage}
          \begin{minipage}[t]{0.28\textwidth}
            \centering
            \renewcommand{\arraystretch}{1.15}
            \footnotesize
            \begin{tabular}{rl}
              $\mathit{Rep}\define$ &\hspace{-0.3cm}  $\iif\ P^{00,11}[q_1,q_2]$\\
              &\hspace{-0.3cm} $\quad\then\ \iif\ P^{00,11}[q_2,q_3]$\\
              &\hspace{-0.3cm} $\quad\quad\then\ \sskip$\\
              &\hspace{-0.3cm} $\quad\quad\eelse\ \; q_3 \apply X\ \pend$\\
              &\hspace{-0.3cm} $\quad\eelse\ \; \iif\ P^{00,11}[q_2,q_3]$\\
              &\hspace{-0.3cm} $\quad\quad\then\ q_1 \apply X$\\
              &\hspace{-0.3cm} $\quad\quad\eelse\ \; q_2 \apply X \ \pend$\\
              &\hspace{-0.3cm} $\pend$
            \end{tabular}
          \end{minipage}       
    \caption{Repetition code to resist noise (left) and the program $Rep$ for recovery (right). The projector that appears in the guard is defined as $P^{00,11} \triangleq |00\>\<00| \vee |11\>\<11|$. }\label{fig-repetition}
  \end{figure}
	
		To describe the correctness of $\mathit{Rep}$ using one single Hoare triple, we employ the same technique as in the previous section to introduce an auxiliary qubit $a$. Let 
		\begin{align*}
			|e_0\>  \define (|000,0\> + |111,1)/\sqrt{2} &\qquad |e_1\>  \define (|100,0\> + |011,1)/\sqrt{2} \\
			|e_2\>  \define (|010,0\> + |101,1)/\sqrt{2}  &\qquad |e_3\>  \define (|001,0\> + |110,1)/\sqrt{2} 
		\end{align*}
		be the output states after the initial maximally entangled state $(|0,0\> + |1,1\>)_{q_1, a}/\sqrt{2}$ undergoes the encoding process and the noisy channel, corresponding to the four possibilities of error occurring: no error, bit-flip error on the first, second, and third qubit, respectively. 
		That is, $|e_i\> = X_{q_i}|e_0\>$ for $i=1,2,3$. For simplicity, we omit the qubit subscripts and assume the fixed order $q_1, q_2, q_3, a$, with a comma between the three principle qubits and the auxiliary one for clarity. Let $P_e = \bigvee_{i=0}^{3}|e_i\>\<e_i|$.

		With these notations, the correctness of $\mathit{Rep}$ can be represented as $\models \ass{P_e}{\mathit{Rep}}{|e_0\>\<e_0|}$.
		In the following, we will use the refinement calculus presented in this paper to derive $\mathit{Rep}$ from the specification (prescription) $S_0 \triangleq [P_e, |e_0\>\<e_0|].$
		
		Let $P^{00,11} \triangleq |00\>\<00| \vee |11\>\<11|$ and $P^{01,10} \triangleq \left(P^{00,11}\right)^\bot = |01\>\<01| \vee |10\>\<10|$.
		First, from \textsc{(Spc-Pres-Cond)} we have
		$
		S_0 \le \iif\ P^{00,11}[q_1,q_2]\ \then\ \left[P^{00,11}_{q_1,q_2} \doublecap P_e, |e_0\>\<e_0|\right]\ \eelse 
		\left[P^{01,10}_{q_1,q_2} \doublecap P_e, |e_0\>\<e_0|\right]\ \pend
		$
		and furthermore,
			\begin{align*}
				\left[P^{00,11}_{q_1,q_2} \doublecap P_e, |e_0\>\<e_0|\right]\le
				\iif\ P^{00,11}[q_2,q_3]\ &\then\ \left[P^{00,11}_{q_2,q_3} \doublecap \left(P^{00,11}_{q_1,q_2} \doublecap P_e\right), |e_0\>\<e_0|\right]\\
				&\eelse \ \
				\left[P^{01,10}_{q_2,q_3} \doublecap\left(P^{00,11}_{q_1,q_2} \doublecap P_e\right), |e_0\>\<e_0|\right]\ \pend
				\\
				 \left[P^{01,10}_{q_1,q_2} \doublecap P_e, |e_0\>\<e_0|\right]\le
				 \iif\ P^{00,11}[q_2,q_3]\ &\then\ \left[P^{00,11}_{q_2,q_3} \doublecap \left(P^{01,10}_{q_1,q_2} \doublecap P_e\right), |e_0\>\<e_0|\right]\\
				 &\eelse \ \
				 \left[P^{01,10}_{q_2,q_3} \doublecap\left(P^{01,10}_{q_1,q_2} \doublecap P_e\right), |e_0\>\<e_0|\right]\ \pend.
			\end{align*}
		We further refine using \textsc{(Wpc-Pres-Unit)} to obtain, say,
		\begin{align*}
			\left[P^{01,10}_{q_2,q_3} \doublecap \left(P^{00,11}_{q_1,q_2} \doublecap P_e\right), |e_0\>\<e_0|\right]&\equiv
			\left[|e_3\>\<e_3|, |e_0\>\<e_0|\right]\equiv q_3 \apply X.
		\end{align*}
		Similarly, the other branches can be refined into $\sskip$, $q_1\apply X$, and $q_2\apply X$, respectively. Finally, combining all of the above steps using the structural rule (\textsc{Stru-Cond}) leads to $S_0 \le \mathit{Rep}$.

  		\subsection{Quantum-to-quantum Bernoulli factory}
		
		Consider a quantum program $\qcoin$ that produces the state $|\psi_p\> = \sqrt{p}|0\>+\sqrt{1-p}|1\>$ where $p$ is an unknown parameter in the range $[0,1]$. Since measuring $|\psi_p\>$ in the computational basis produces the Bernoulli distribution $\mathrm{B}(p)$, the state $|\psi_p\>$ can be regarded as a quantum analogy of a classical coin with bias $p$ towards 0.
		The quantum-to-quantum Bernoulli factory problem aims to determine all complex functions $h: [0,1]\ra \mathbb{C}$ so that the quantum state $|h(p)\> \define \frac{1}{\sqrt{1+|h(p)|^2}}(h(p)|0\> + |1\>)$, corresponding to the Bernoulli distribution $\mathrm{B}(|h(p)|^2)$, can be produced using $\qcoin$ as a subroutine. 
		It is discovered in~\cite{JZS18} that the necessary and sufficient condition for $h$ to be implementable is that $h(p)$ belongs to the field generated by $\sqrt{\frac{p}{1-p}}$ and $\mathbb{C}$ in the sense that
		\begin{equation}\label{eq:hpfunc}
		h(p) = \frac{g_1(p)}{g_2(p)}\sqrt{\frac{p}{1-p}} + \frac{g_3(p)}{g_4(p)}
		\end{equation}
		where $g_i(p)$, $1\leq i\leq 4$, are polynomials of $p$ with coefficients in $\mathbb{C}$. Note also that $|\psi_p\> = \left|\sqrt{\frac{p}{1-p}}\right\>$.

			{
			\renewcommand{\arraystretch}{1.3}
			\begin{figure}
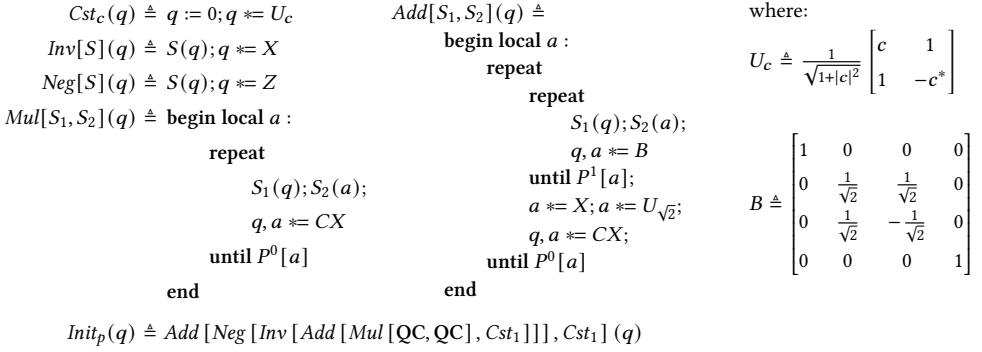

				\footnotesize
				\begin{tabular}{lll}
					\begin{tabular}{rl}
						$\mathit{Cst}_c(q)\triangleq$ & \hspace{-0.3cm} $q:=0; q\apply U_c$\\
						$\mathit{Inv}[S](q)\triangleq$ & \hspace{-0.3cm} $S(q); q \apply X$\\
						$\mathit{Neg}[S](q)\triangleq$ & \hspace{-0.3cm} $S(q); q \apply Z$ \\
						$\mathit{Mul}[S_1, S_2](q)\triangleq$ & \hspace{-0.3cm} $\blocal\ a :$ \\
						&\hspace{-0.3cm} $\qquad\irepeat$\\
						&\hspace{-0.3cm} $\qquad\qquad S_1(q); S_2(a);$\\
						&\hspace{-0.3cm} $\qquad\qquad q,a\apply CX$\\
						&\hspace{-0.3cm} $\qquad\iuntil\ P^{0}[a]$\\
						&\hspace{-0.3cm} $\pend$
					\end{tabular}
					\quad& 
                    {
                    \renewcommand{\arraystretch}{1.03}
					\begin{tabular}{rl}
						&\hspace{-1cm} $\mathit{Add}[S_1, S_2](q) \triangleq$\\
						& \hspace{-0.3cm} $\blocal\ a :$ \\
						& \hspace{-0.3cm} $\qquad \irepeat$\\
						& \hspace{-0.3cm} $\qquad\qquad \irepeat$\\
						& \hspace{-0.3cm} $\qquad\qquad\qquad S_1(q); S_2(a);$\\
						& \hspace{-0.3cm} $\qquad\qquad\qquad q,a\apply B$\\
						& \hspace{-0.3cm} $\qquad\qquad\iuntil\ P^{1}[a];$ \\ 
						& \hspace{-0.3cm} $\qquad\qquad a\apply X; a \apply U_{\sqrt{2}};$ \\
						& \hspace{-0.3cm} $\qquad\qquad q,a\apply CX;$\\
						& \hspace{-0.3cm} $\qquad \iuntil\ P^{0}[a]$\\
						& \hspace{-0.3cm} $\pend$
					\end{tabular}
                    }
					\qquad& 
					\begin{tabular}{l}
                        where: \\
						$U_{c}\define
						\frac{1}{\sqrt{1+|c|^2}}\begin{bmatrix}
							c & 1 \\
							1 & -c^*
						\end{bmatrix}
						$\\ \\
						$B\define
						 \begin{bmatrix}
							1 & 0 & 0 & 0 \\
							0 & \frac{1}{\sqrt{2}} & \frac{1}{\sqrt{2}} & 0 \\
							0 & \frac{1}{\sqrt{2}} & -\frac{1}{\sqrt{2}} & 0 \\
							0 & 0 & 0 & 1
						\end{bmatrix}
						$\\ \\
					\end{tabular}
				\end{tabular}
				\begin{tabular}{l}
					$\mathit{Init_p}(q) \triangleq \mathit{Add}\left[\mathit{Neg}\left[\mathit{Inv}\left[\mathit{Add}\left[\mathit{Mul}\left[\qcoin, \qcoin\right], \mathit{Cst}_1\right]\right]\right], \mathit{Cst}_1\right](q)$\hspace{3.56cm}
				\end{tabular}
				\caption{Constructs for quantum-to-quantum Bernoulli factory.
				}
				\label{fig:q2q-bernoulli}
			\end{figure}
		}

		Our aim of this section is to implement the quantum program which can produce $|h(p)\>$ for any given function $h$ of the form in Eq.~\eqref{eq:hpfunc}.
		Obviously, one way to achieve this is to implement subroutines that realize all the basic field operators. Specifically, we will construct the following programs: $\mathit{Cst}_c$ for constant $c\in \mathbb{C}$, $\mathit{Neg}$ for negation $-$, $\mathit{Inv}$ for inversion ${}^{-1}$, $\mathit{Mul}$ for multiplication $*$, and $\mathit{Add}$ for addition $+$. These programs, together with an additional program $\mathit{Init_p}$ that produces $|p\>$, are depicted in Fig.~\ref{fig:q2q-bernoulli}, with the help of the following extra program constructs. 
		The correctness of these programs will be seen and formally proved through the refinement processes below.
		
		\begin{enumerate}
			\item The \textit{block command} $\blocal\ \bar{a} : S\ \pend$, which allows the use of local variables $\bar{a}$. It is required that qubits $\bar{a}$ are initially not entangled with other qubits and are they discarded after executing $S$; 
			\item The \textit{program scheme} $S[A_1,\ldots, A_n]$, which invokes subroutines $A_1,\ldots, A_n$. This construct is indispensable for the current example as the implementation details of the initially given $\qcoin$ are unknown, and it can only be called as a black box. Furthermore, the programs realizing field operators like $\mathit{Add}$ are all implemented with subroutines.
			\item The \emph{repeat-until} construct as a syntactic sugar
			\[
			\irepeat\ S\ \iuntil\ P[\bar{q}]\ \triangleq\ S; \while\ P^\bot[\bar{q}]\ \ddo\ S\ \od.
			\]
			
		\end{enumerate}
		The first two constructs have been explored in detail in~\cite{ying2016foundations}. We also adopt the convention of defining a family of programs on different sets of qubits; for example, we denote $S(q)$ the one-qubit program $S$ acting on $q$. Then $S(a)$ means the same program, but with $q$ replaced by $a$. Consequently, in general, we will have a program of the form $P[A_{11},\cdots, A_{n1}](\bar{q})$ with subroutines $A_{11},\cdots, A_{n1}$ and acting on qubits $\bar{q}$. First, we present some refinement rules for the new constructs, illustrated in Fig. \ref{fig:block-scheme-rules}. The soundness of these rules are easy to check.

		To show how these subroutines can be used to build a program to implement any given function, let us consider the case where $h(p) = 2/(p+1)$. Then from the flow graph

  \vspace{0.15cm}
\begin{tikzpicture}[x=0.75pt,y=0.75pt,yscale=-1,xscale=1]
  \path (40,63);

  \draw (125,0) node [anchor=north west][inner sep=0.75pt]   [align=left] {$\displaystyle Init_{p}$};
  \draw [shift={(124,0)}, rotate = 0] [line width=0.08]  (0,-2) -- (31,-2) -- (31,15) -- (0,15) -- cycle ;
  \draw    (155,7) -- (175,7) ;
  \draw [shift={(177,7)}, rotate = 180] [fill={rgb, 255:red, 0; green, 0; blue, 0 }  ][line width=0.08]  [draw opacity=0] (6,-3) -- (0,0) -- (6,3) -- cycle    ;
  \draw (178,0) node [anchor=north west][inner sep=0.75pt]   [align=left] {$\displaystyle p$};

  \draw (128,30) node [anchor=north west][inner sep=0.75pt]   [align=left] {$\displaystyle Cst_{1}$};
  \draw [shift={(127,30)}, rotate = 0] [line width=0.08]  (0,-2) -- (28,-2) -- (28,15) -- (0,15) -- cycle ;
  \draw    (155,37) -- (175,37) ;
  \draw [shift={(177,37)}, rotate = 180] [fill={rgb, 255:red, 0; green, 0; blue, 0 }  ][line width=0.08]  [draw opacity=0] (6,-3) -- (0,0) -- (6,3) -- cycle    ;
  \draw (178,30) node [anchor=north west][inner sep=0.75pt]   [align=left] {$\displaystyle 1$};

  \draw    (190,7) -- (205,22) ;
  \draw    (190,37) -- (205,22) ;
  \draw    (205,22) -- (248,22) ;
  \draw [shift={(250,22)}, rotate = 180] [fill={rgb, 255:red, 0; green, 0; blue, 0 }  ][line width=0.08]  [draw opacity=0] (6,-3) -- (0,0) -- (6,3) -- cycle    ;
  \draw (211,7) node [anchor=north west][inner sep=0.75pt]   [align=left] {$\displaystyle Add$};
  \draw [shift={(210,7)}, rotate = 0] [line width=0.08]  (0,-2) -- (28,-2) -- (28,15) -- (0,15) -- cycle ;
  \draw (251,15) node [anchor=north west][inner sep=0.75pt]   [align=left] {$\displaystyle p+1$};
  \draw    (281,22) -- (321,22) ;
  \draw [shift={(323,22)}, rotate = 180] [fill={rgb, 255:red, 0; green, 0; blue, 0 }  ][line width=0.08]  [draw opacity=0] (6,-3) -- (0,0) -- (6,3) -- cycle    ;
  \draw (289,7) node [anchor=north west][inner sep=0.75pt]   [align=left] {$\displaystyle Inv$};
  \draw [shift={(288,7)}, rotate = 0] [line width=0.08]  (0,-2) -- (23,-2) -- (23,15) -- (0,15) -- cycle ;
  \draw (324,15) node [anchor=north west][inner sep=0.75pt]   [align=left] {$\displaystyle 1/( p+1)$};

  \draw (316,45) node [anchor=north west][inner sep=0.75pt]   [align=left] {$\displaystyle Cst_{2}$};
  \draw [shift={(315,45)}, rotate = 0] [line width=0.08]  (0,-2) -- (28,-2) -- (28,15) -- (0,15) -- cycle ;
  \draw    (343,52) -- (363,52) ;
  \draw [shift={(365,52)}, rotate = 180] [fill={rgb, 255:red, 0; green, 0; blue, 0 }  ][line width=0.08]  [draw opacity=0] (6,-3) -- (0,0) -- (6,3) -- cycle    ;
  \draw (366,45) node [anchor=north west][inner sep=0.75pt]   [align=left] {$\displaystyle 2$};

  \draw    (378,22) -- (393,37) ;
  \draw    (378,52) -- (393,37) ;
  \draw    (393,37) -- (436,37) ;
  \draw [shift={(438,37)}, rotate = 180] [fill={rgb, 255:red, 0; green, 0; blue, 0 }  ][line width=0.08]  [draw opacity=0] (6,-3) -- (0,0) -- (6,3) -- cycle    ;
  \draw (399,22) node [anchor=north west][inner sep=0.75pt]   [align=left] {$\displaystyle Mul$};
  \draw [shift={(398,22)}, rotate = 0] [line width=0.08]  (0,-2) -- (28,-2) -- (28,15) -- (0,15) -- cycle ;
  \draw (439,30) node [anchor=north west][inner sep=0.75pt]   [align=left] {$\displaystyle 2/( p+1)$};
\end{tikzpicture}
		
		\noindent we see how the program $\mathit{Mul}[\mathit{Inv}[\mathit{Add}[\mathit{Init}_p,\mathit{Cst}_1]] , \mathit{Cst}_2]$ implements the desired quantum coin.

		In the following, we will employ the refinement calculus to derive these basic subroutines; the whole program is shown in the Supplementary Material \ref{app:q2q-bernoulli}. More specifically, we need to show
			\begin{align*}
				&[I,|\widetilde{c}\>]_q \le \mathit{Cst}_c(q) &&
				[I,|1/h\>]_q \le \mathit{Inv}[F_h](q) &&
				[I,|-h\>]_q \le \mathit{Neg}[F_h](q) \\
				&[I, |h_1*h_2\>]_q \le \mathit{Mul}[F_{h_1}, F_{h_2}](q) &&
				[I,|h_1+h_2\>]_q \le \mathit{Add}[F_{h_1}, F_{h_2}](q) &&
				[I,|p\>]_q \le \mathit{Init}_p(q)
			\end{align*}
			where $F_h$ is a subroutine to produce the state $|h\>$ that has been realized; that is, $[I,|h\>]_q \le F_h(q)$. Furthermore, we add a tilde over $c$ in $|\widetilde{c}\>$ to emphasize that it represents $\frac{1}{|c|^2 +1}(c|0\>+|1\>)$. This avoids confusion when $c=0$ or 1.
			For simplicity, sometimes we slightly abuse the notation by using a pure state $|\psi\>$ to denote its corresponding projector $|\psi\>\<\psi|$.

		{
			\renewcommand{\arraystretch}{1.3}
			\begin{figure}
				\footnotesize
				\begin{tabular}{ll}
					\textsc{(Stru-Local)} & \begin{tabular}{l}$
						\inferrule* [rightstyle={\footnotesize \sc}]{
							S \le T
						}{
							\blocal\ \bar{a} : S\ \pend \le \blocal\ \bar{a} : T\ \pend
						}
						$	\end{tabular}\vspace{0.5em} \\
					\textsc{(Stru-Scheme)} & \begin{tabular}{l}$
						\inferrule* [rightstyle={\footnotesize \sc}]{
							\forall i, \bar{q}_i,\ A_{i1}(\bar{q}_i) \le A_{i2}(\bar{q}_i)
						}{
							P[A_{11},\cdots, A_{n1}](\bar{q}) \le P[A_{12},\cdots, A_{n2}](\bar{q})
						}
						$	\end{tabular}\vspace{0.5em} \\
						\textsc{(Stru-Repeat)} & \begin{tabular}{l}$
							\inferrule* [rightstyle={\footnotesize \sc}]{
								S \le T
							}{
								\repstm{S}{P[\bar{q}]} \le \repstm{T}{P[\bar{q}]}
							}
							$	\end{tabular}
				\end{tabular}\\
				\vspace{1em}
				\begin{tabular}{llcl}
					\textsc{(Pres-Local)} & $[P,Q]_{\bar{q}}$ & $\equiv$ & $ \blocal\ \bar{a} : [P, Q]_{\bar{q},\bar{a}}\ \pend$ 	 \qquad if $\bar{q}\cap \bar{a} = \emptyset$						\\
					\textsc{(Pres-Repeat)} & $[P,Q]_{\bar{q}}$ & \le & $ \irepeat\ [P, (R^\bot \rightsquigarrow P) \wedge (R \rightsquigarrow Q)]_{\bar{q}}\ \iuntil\ R[\bar{q}]$
				\end{tabular}
				\caption{Refinement rules for block command, program scheme, and repeat-until constructs. 
				}
				\label{fig:block-scheme-rules}
			\end{figure}
		}

		First, note that 
		$U_c|0\> \propto c|0\>+|1\> \propto |\widetilde{c}\>$, $X|h(p)\> \propto |0\> + h(p)|1\>
		\propto |1/h(p)\>$,  and $Z|h(p)\> \propto h(p)|0\> - |1\>
		\propto |-h(p)\>$. Here we use the symbol $\propto$ to ignore the unimportant normalization factor. We can easily construct these subroutines by employing (\textsc{Comm-Pres-Seq}) and (\textsc{Spc-Pres-Unit}):
		\begin{align*}
			\left[I, |\widetilde{c}\>\right]_q &\quad \le \quad \left[I, |0\>\right]_q;\ \left[|0\>, |\widetilde{c}\>\right]_q
			&&\hspace{-5em}\le\quad q:=0;\ q \apply U_c\\
			\left[I, |1/h\>\right]_q &\quad\le \quad \left[I, |h\>\right]_q;\ \left[|h\>, |1/h\>\right]_q
			&&\hspace{-5em} \le\quad F_h(q);\ q \apply X\\
			[I,|-h\>]_q &\quad\le \quad[I,|h\>]_q;\ [|h\>, |-h\>]_q &&\hspace{-5em}\le\quad F_h(q);\ q \apply Z.
		\end{align*}

		Next, we derive
        {\small
		\begin{align*}
			\left[I, |h_1*h_2\>\right]_q &\equiv \blocal\ a : [I, |h_1*h_2\>_q]_{q,a}\ \pend&\hspace{-7em}(\textsc{Pres-Local})\\
			&\le \blocal\ a : [I, |h_1*h_2, 0\>]_{q,a}\ \pend &\hspace{-7em}(\textsc{Comm-Pres-Cons})\\
			&\le \blocal\ a : \irepeat\ [I, (P^1_a \rightsquigarrow I) \wedge (P^{0}_a \rightsquigarrow |h_1*h_2,0\>)]_{q,a}\ \iuntil\ P^{0}[a]\ \pend \\
			& &\hspace{-7em} (\textsc{Pres-Repeat})\\
			&\le \blocal\ a : \irepeat\ [I, P^1_a \vee |h_1*h_2,0\>]_{q,a}\ \iuntil\ P^{0}[a]\ \pend &\hspace{-7em} \mbox{(Simplification)}
		\end{align*}
        }
		
        \noindent Note that $CX|h_1, h_2\> \propto (h_1|0\>+h_2|1\>)|1\> + (h_1h_2|0\>+|1\>)|0\> \le P^1_a \vee |h_1*h_2,0\>$. Then
		\begin{align*}
			[I, P^1_a \vee |h_1*h_2,0\>]_{q,a} 
			&\le [I, |h_1\>_q]_{q,a};\ [|h_1\>_q, |h_1, h_2\>]_{q,a};\ [|h_1, h_2\>, CX|h_1, h_2\>]_{q,a}\\
			&&\hspace{-11em}  (\textsc{Comm-Pres-Seq}), (\textsc{Comm-Pres-Cons})\\
			&\le F_{h_1}(q);\ F_{h_2}(a);\ q,a\apply CX &\hspace{-11em} (\textsc{Spc-Pres-Unit})
		\end{align*}
		and finally we obtain $\mathit{Mul}$ as shown in Fig. \ref{fig:q2q-bernoulli} by putting them all together.
		
		The case of $\mathit{Add}$ is similar to $\mathit{Mul}$. From the fact that $h_1+h_2 = \frac{h_1+h_2}{\sqrt{2}}* \sqrt{2}$, we have
        {\small
		\begin{align*}
			\left[I, |h_1+h_2\>\right]_q 
			&\le \blocal\ a : \irepeat\ \left[I, \left|\frac{h_1+h_2}{\sqrt{2}}, \widetilde{\sqrt{2}}\right\>\right]_{q,a}; 
			q,a\apply CX\ \iuntil\ P^{0}[a]\ \pend.
		\end{align*}
        }
		
        \noindent Furthermore, as
		$B|h_1, h_2\> \propto \left(h_1h_2|0\> + \frac{h_1-h_2}{\sqrt{2}}|1\>\right)|0\> + \left(\frac{h_1+h_2}{\sqrt{2}}|0\> + |1\>\right)|1\> \le P^1_a \vee \left|\frac{h_1+h_2}{\sqrt{2}},1\right\>$ and  $|\mbox{\footnotesize $\widetilde{\sqrt{2}}$}\> = U_{\sqrt{2}}|0\>$, we compute 
		{\small
			\begin{align*}
				\left[I, \left|\frac{h_1+h_2}{\sqrt{2}}, \widetilde{\sqrt{2}}\right\>\right]_{q,a}
				&\le \left[I, \left|\frac{h_1+h_2}{\sqrt{2}}, 1\right\>\right]_{q,a};\  \left[\left|\frac{h_1+h_2}{\sqrt{2}}, 1\right\>, \left|\frac{h_1+h_2}{\sqrt{2}}, 0\right\>\right]_{q,a};\ \left[\left|\frac{h_1+h_2}{\sqrt{2}}, 0\right\>, \left|\frac{h_1+h_2}{\sqrt{2}}, \widetilde{\sqrt{2}}\right\>\right]_{q,a} \\
				&\le \irepeat\ \left[I, \left(P^0_a \rightsquigarrow I\right) \wedge \left(P^{1}_a \rightsquigarrow \left|\frac{h_1+h_2}{\sqrt{2}},1\right\>\right)\right]_{q,a}\ \iuntil\ P^{1}[a];\ a\apply X;\ a \apply U_{\sqrt{2}} \\
				&\equiv \irepeat\ \left[I, P^0_a \vee \left|\frac{h_1+h_2}{\sqrt{2}},1\right\>\right]_{q,a}\ \iuntil\ P^{1}[a];\ a\apply X;\ a \apply U_{\sqrt{2}}
			\end{align*}
		}
		while 
					\begin{align*}
			\left[I, P^0_a \vee \left|\frac{h_1+h_2}{\sqrt{2}},1\right\>\right]_{q,a}
			&\le [I, |h_1\>_q]_{q,a};\ [|h_1\>_q, |h_1,h_2\>]_{q,a};\ [|h_1,h_2\>, B|h_1, h_2\>]_{q,a}\\
			&\le F_{h_1}(q);\ F_{h_2}(a);\ q,a\apply B.
		\end{align*}
		These lead to $\mathit{Add}$ shown in Fig. \ref{fig:q2q-bernoulli}.

		Finally, note that $\sqrt{\frac{p}{1-p}}\xrightarrow[]{\mathit{Mul}\ \sqrt{\frac{p}{1-p}}}\frac{p}{1-p}\xrightarrow[]{\mathit{Add}\ 1}\frac{1}{1-p}\xrightarrow[]{\mathit{Inv}}1-p\xrightarrow[]{\mathit{Neg}}p-1\xrightarrow[]{\mathit{Add}\ 1}p$. We can derive $\mathit{Init}_p$ as follows via (\textsc{Stru-Scheme})
		\begin{align*}
			[I,|p\>]_q 
			&\le \mathit{Add}\left[F_{p-1}, F_1\right]
			\le \mathit{Add}\left[\mathit{Neg}\left[F_{1-p}\right], \mathit{Cst}_1\right]
			\le \mathit{Add}\left[\mathit{Neg}\left[\mathit{Inv}\left[F_{1/(1-p)}\right]\right], \mathit{Cst}_1\right](q) \\
			&\le \mathit{Add}\left[\mathit{Neg}\left[\mathit{Inv}\left[\mathit{Add}\left[F_{p/1-p}, F_1\right]\right]\right], \mathit{Cst}_1\right](q) \\
			&\le \mathit{Add}\left[\mathit{Neg}\left[\mathit{Inv}\left[\mathit{Add}\left[\mathit{Mul}\left[F_{\sqrt{p/1-p}}, F_{\sqrt{p/1-p}}\right], \mathit{Cst}_1\right]\right]\right], \mathit{Cst}_1\right](q) \\
			&\le \mathit{Add}\left[\mathit{Neg}\left[\mathit{Inv}\left[\mathit{Add}\left[\mathit{Mul}\left[\qcoin, \qcoin\right], \mathit{Cst}_1\right]\right]\right], \mathit{Cst}_1\right](q)
		\end{align*}
		 since the given $\qcoin$ implements $F_{\sqrt{p/(1-p)}}$.
   
		\section{Prototype Implementation: \texttt{Quire}}
        \label{sec:quire}
		To showcase the practical utility of our theoretical framework, we present a prototype implementation named \texttt{Quire}, a Python-based \textbf{qu}antum \textbf{i}nteractive \textbf{r}efinement \textbf{e}ngine. \texttt{Quire} performs various tasks, including checking the well-formedness of operator terms and quantum programs, computing classical simulations of program executions, verifying adherence to specified prescriptions, and aiding in the stepwise refinement of programs in the sense of partial correctness. This tool can be accessed via the GitHub repository: \url{https://github.com/LucianoXu/Quire}.

        \begin{figure}
          \includegraphics*[width=\textwidth]{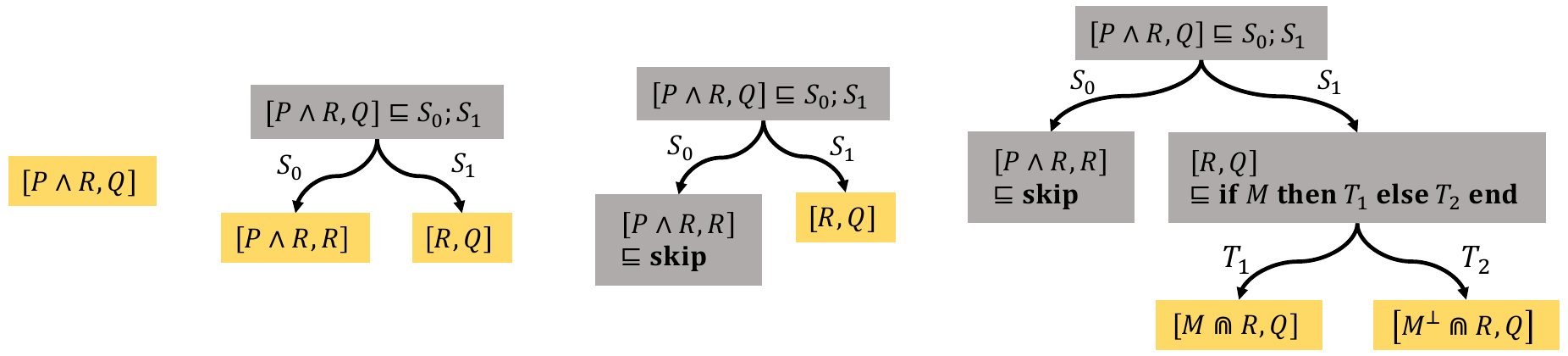}
          \vspace{-0.3cm}
          \caption{An illustration demonstrating the stepwise refinement of a random example. The abstract syntax tree expands at each step, with unresolved prescriptions (i.e., the goals) highlighted as yellow boxes at the leaf nodes.}
          \label{fig:refine illustration}
        \end{figure}

\begin{figure}

		\begin{tikzpicture}

			\node[draw=gray, fill=gray!20, dashed, rounded corners, inner xsep=5pt, inner ysep =0pt] (setup) {
			  \begin{minipage}{0.96\textwidth}
				\lstset{style=Quire-tiny}
				\begin{lstlisting}
// Setup for the refinement. Constant definitions be utilised later on. 
Def Peq := [|00$\>$] $\vee$ [|11$\>$]. // [|v>] is the syntax for the projector |v><v|
Def Pe0 := 0.5 [|0000$\>$ + |1111$\>$]. Def Pe1 := 0.5 [|1000$\>$ + |0111$\>$].
Def Pe2 := 0.5 [|0100$\>$ + |1011$\>$]. Def Pe3 := 0.5 [|0010$\>$ + |1101$\>$].
Def Pe := Pe0 $\vee$ Pe1 $\vee$ Pe2 $\vee$ Pe3.
				\end{lstlisting}
				\end{minipage}
        };
  
			\node[draw=gray, fill=gray!20, dashed, rounded corners, inner xsep=5pt, inner ysep =0pt, below= 5pt of setup ] (start) {
			  \begin{minipage}{0.96\textwidth}
				\lstset{style=Quire-tiny}
				\begin{lstlisting}
Refine Rep : < Pe[q1 q2 q3 a], Pe0[q1 q2 q3 a] >.
				\end{lstlisting}
				\end{minipage}
        };

        \node[draw=gray, fill=yellow!10, dashed, rounded corners, inner xsep=5pt, inner ysep =0pt, below=20pt of start.south west, anchor=west] (S1) {
      \begin{minipage}{0.575\textwidth}
				\lstset{style=Quire-tiny}
				\begin{lstlisting}
Goal (1/1)
< Pe[q1 q2 q3 a], Pe0[q1 q2 q3 a] >
      \end{lstlisting}
      \end{minipage}
			};

			\node[draw=gray, fill=gray!20, dashed, rounded corners, inner xsep=5pt, inner ysep =0pt, right= 5pt of S1] (C1) {
			  \begin{minipage}{0.345\textwidth}
				\lstset{style=Quire-tiny}
				\begin{lstlisting}
Step If Peq[q1 q2].
				\end{lstlisting}
				\end{minipage}
        };

        \node[draw=gray, fill=yellow!10, dashed, rounded corners, inner xsep=5pt, inner ysep =0pt, below=5pt of S1] (S2) {
      \begin{minipage}{0.575\textwidth}
				\lstset{style=Quire-tiny}
				\begin{lstlisting}
Goal (1/2)
< (Peq[q1 q2] $\Cap$ Pe[q1 q2 q3 a]), Pe0[q1 q2 q3 a] >
Goal (2/2)
< ((Peq[q1 q2]^$\bot$) $\Cap$ Pe[q1 q2 q3 a]), Pe0[q1 q2 q3 a] >
      \end{lstlisting}
      \end{minipage}
			};

      \node[draw=gray, fill=gray!20, dashed, rounded corners, inner xsep=5pt, inner ysep =0pt, right= 5pt of S2] (C2) {
			  \begin{minipage}{0.345\textwidth}
				\lstset{style=Quire-tiny}
				\begin{lstlisting}
Step If Peq[q2 q3].
WeakenPre Pe0[q1 q2 q3 a].
				\end{lstlisting}
				\end{minipage}
        };

      \node[draw=gray, fill=yellow!10, dashed, rounded corners, inner xsep=5pt, inner ysep =0pt, below=5pt of S2] (S3) {
      \begin{minipage}{0.575\textwidth}
				\lstset{style=Quire-tiny}
				\begin{lstlisting}
Goal (1/3)
< Pe0[q1 q2 q3 a], Pe0[q1 q2 q3 a] >
Goal (2/3)
< ((Peq[q2 q3]^$\bot$) $\Cap$ (Peq[q1 q2] $\Cap$ Pe[q1 q2 q3 a])), Pe0[q1 q2 q3 a] >
Goal (3/3)
< ((Peq[q1 q2]^$\bot$) $\Cap$ Pe[q1 q2 q3 a]), Pe0[q1 q2 q3 a] >
      \end{lstlisting}
      \end{minipage}
			};

      \node[draw=gray, fill=gray!20, dashed, rounded corners, inner xsep=5pt, inner ysep =0pt, right= 5pt of S3] (C3) {
			  \begin{minipage}{0.345\textwidth}
				\lstset{style=Quire-tiny}
				\begin{lstlisting}
Step skip.
WeakenPre Pe3[q1 q2 q3 a]. Step X[q3].
        \end{lstlisting}
				\end{minipage}
        };

      \node[draw=gray, fill=yellow!10, dashed, rounded corners, inner xsep=5pt, inner ysep =0pt, below=5pt of S3] (S4) {
        \begin{minipage}{0.575\textwidth}
          \lstset{style=Quire-tiny}
          \begin{lstlisting}
Goal (1/1)
< ((Peq[q1 q2]^$\bot$) $\Cap$ Pe[q1 q2 q3 a]), Pe0[q1 q2 q3 a] >     
        \end{lstlisting}
        \end{minipage}
        };

        \node[draw=gray, fill=gray!20, dashed, rounded corners, inner xsep=5pt, inner ysep =0pt, below= of C3] (C4) {
          \begin{minipage}{0.345\textwidth}
          \lstset{style=Quire-tiny}
          \begin{lstlisting}
Step If Peq[q2 q3].
WeakenPre Pe1[q1 q2 q3 a]. Step X[q1].
WeakenPre Pe2[q1 q2 q3 a]. Step X[q2].
          \end{lstlisting}
          \end{minipage}
          };

          \node [below= -5pt of S4] (vdots) {
            \begin{minipage}{0.575\textwidth}
              $$
              \vdots
              $$
            \end{minipage}
            };
          
          \node [draw=gray, fill=yellow!10, dashed, rounded corners, inner xsep=5pt, inner ysep =0pt, below= 2pt of vdots] (S5) {
            \begin{minipage}{0.575\textwidth}
              \lstset{style=Quire-tiny}
              \begin{lstlisting}
Goal Clear.
              \end{lstlisting}
            \end{minipage}
          };

          \node [draw=gray, fill=gray!20, dashed, rounded corners, inner xsep=5pt, inner ysep =0pt, below= 15pt of S5.south west, anchor=west] (end) {
            \begin{minipage}{0.96\textwidth}
              \lstset{style=Quire-tiny}
              \begin{lstlisting}
End.
              \end{lstlisting}
            \end{minipage}
          };
          

        \end{tikzpicture}
        \caption{Illustration of the interactive refinement process for the repetition code. 
        Gray boxes represent \texttt{Quire} commands, and the corresponding responses are in yellow boxes. In refinement mode, unsolved prescriptions appear as individual goals in the output. 
        The user suggests the next step through commands, \texttt{Quire} verifies their legality, and responds with the result, either a change in goals or rejection of the suggestion. 
        This entire procedure mirrors that in Sec.~\ref{subsec: rep3}.
        }
        \label{fig:rep3refine}
      \end{figure}

        As depicted in Fig. \ref{fig:refine illustration}, the refinement process in \texttt{Quire} resembles the growth of an abstract syntax tree. In addition to component programs, refinement relations $S \sqsubseteq S'$ are also represented as nodes in the syntax tree. The refinement begins with a single prescription as the root, and the syntax tree expands through the application of refinement rules. Consequently, unsolved prescriptions are shown as leaf nodes in the tree, managed by \texttt{Quire} as \emph{goals}. Once the leaf nodes no longer contain goals, the refinement process concludes and the history of stepwise refinement is accurately preserved within the syntax tree. We demonstrate the functionality and application of \texttt{Quire} using the repetition code example from Sec. \ref{subsec: rep3}, as depicted in Fig. \ref{fig:rep3refine}. Another example implementing the $Z$-rotation 
        gate from Sec.~\ref{subsec:zrotation} is provided in the Supplementary Material \ref{sec:Rz in Quire}, and comprehensive documentation is available in the Supplementary Material \ref{sec: doc_quire}.

	   We highlight two closely related tools, \texttt{CorC}~\cite{runge2019tool} and \texttt{CoqQ}~\cite{zhou2023coqq}, from which we draw design methodologies. \texttt{CorC} implements correctness by construction for classical programs, with the refinement procedure modeled and manipulated as a tree in its graphic editor. \texttt{CoqQ} is a tool embedded in the \texttt{Coq} proof assistant, facilitating the formalization and verification of quantum programs. It enhances readability and usability by allowing the use of Dirac notations. In contrast, we aim to strike a balance between user-friendliness, expressiveness, and automation in the implementation of \texttt{Quire}. Our examination of \texttt{Quire} covers several aspects, demonstrating that quantum refinement calculus can practically assist in designing correct quantum programs.

        As a showcase of refinement, \texttt{Quire} underscores the significance of automation and usability. We opt for Python and adhere to variable-free expressions to ensure that the tool is user-friendly and easy to develop. With these constrained expressions, we are able to achieve completeness in deciding correct refinement by leveraging the lattice structure of projective predicates and employing numerical calculations in Python. In this regard, \texttt{Quire} assists users by handling the burden of verification.

		Due to the same design choices, the prototype implementation \texttt{Quire} exhibits limitations in expressiveness, reliability, and performance. Tools built on theorem provers like CoqQ~\cite{zhou2023coqq} can reason in higher-order logic, but \texttt{Quire} restricts expressions to be variable-free. The Python foundation of \texttt{Quire} means that its methods, calculations, and rules lack a constructive basis, increasing the likelihood of bugs. Moreover, reliability concerns arise from numerical calculations, including issues with floating point precision and the linear algebra algorithms that we implemented. Given that we employ matrix methods in \texttt{NumPy} as the calculation back-end, the time and space complexity grows exponentially with the number of qubits.

		\section{Conclusion and Future Works}
				\label{sec:conclusion}
		In this paper, we explore the refinement calculus of non-deterministic quantum programs, employing projectors as state assertions and defining program properties in terms of partial correctness. We introduce refinement rules based on the weakest liberal precondition and the strongest postcondition, respectively, offering a structured approach to the incremental development of quantum programs applicable in various contexts. To demonstrate the real-world utility of these rules, we showcase their practical application through instances such as realizing a $Z$ rotation gate, constructing the repetition code, and developing the quantum-to-quantum Bernoulli factory. Furthermore, we present an interactive prototype tool crafted in Python, serving as a hands-on support system for programmers immersed in the iterative process of creating accurate quantum programs.

	For future work, there are three natural directions to extend our current results: (1) considering effects or even sets of effects as quantum state assertions. This extension requires the generalization of the qualitative relation of a state satisfying a projector assertion to the quantitative expectation of a state satisfying an effect (or a set of effects). (2) Defining program properties in terms of total correctness. While we anticipate that most refinement rules proposed in the current paper remain valid in this stronger setting, the incorporation of a notion of ranking functions is necessary to ensure (probabilistic) termination to deal with while loops.
      (3) Quire still needs to be improved in various aspects to become a practical tool. The primary focus is on expressiveness, with the aim of facilitating the use of free variables, arrays of quantum variables, and intricate constructs of quantum registers. Maintaining automation necessitates careful considerations, which may include the incorporation of the decision procedure of Dirac notations and a constructive type system, as well as existing works for automated reasoning about quantum logic.

\section*{Acknowledgment} We thank Prof Mingsheng Ying for inspiring discussions. This work is partially supported by the Australian Research Council (Grant No DP220102059).		        
		
		\bibliography{ref}
		
		\appendix
\newpage
	\noindent{\Large \textbf{Supplementary material and deferred proofs}}
	
		\section{Preliminaries}
  \label{sec-preliminaries}
				
		This section provides essential preliminaries necessary to grasp the content of this paper. We will start with a brief introduction to basic notations in quantum computing, followed by an overview of quantum logic, with a specific focus on the Sasaki implication and conjunction. These concepts are crucial in the refinement calculus developed in this paper. For an in-depth exploration of relevant backgrounds, we recommend referring to~\cite[Chapter 2]{nielsen2002quantum} for quantum computing and~\cite{Birkhoff1936TheLO} for quantum logic.

\subsection{Foundations of quantum computing}
\label{sec:preliminary-foundation}
According to von Neumann's formalism of quantum mechanics
\cite{vN55}, any quantum system with finite degrees of freedom is associated with a finite-dimensional Hilbert space $\h$ called its \emph{state space}. When the dimension of $\h$ is $2$, such a system is called a \emph{qubit}, the analogy of bit in classical computing. A {\it pure state} of the system is described by a normalized vector in $\h$. 
When the system is in state $|\psi_i\>$\footnote{In this paper, we follow the tradition in the quantum computing community to denote vectors in $\h$ in the Dirac form like $|\psi\>$.} with probability $p_i$, $i\in I$, it is said to be in a \emph{mixed} state, represented by the \emph{density operator} $\sum_{i\in I} p_i|\psi_i\>\<\psi_i|$ on $\h$. Here $\<\psi|$ is the \emph{adjoint} vector of $|\psi\>$ and $|\psi\>\<\psi|$ denotes the \emph{outer product} of $|\psi\>$ with itself. When an orthonormal basis of $\h$ is taken and $|\psi\>$ is written as a column vector under this basis, then $\<\psi|$ is merely the transpose and complex conjugate of $|\psi\>$, and $|\psi\>\<\psi|$ is the (matrix) product of $|\psi\>$ and $\<\psi|$. 

Obviously, a density operator is positive and has trace 1. 
In this paper, we follow Selinger's convention~\cite{selinger2004towards} to  regard \emph{partial density operators}, i.e. positive operators with traces not greater than 1 as (normalized) quantum states. Intuitively, a partial density operator $\rho$ denotes a legitimate quantum state $\rho/\tr(\rho)$ that is obtained with probability $\tr(\rho)$. Denote by $\dh$ and $\lh$ the sets of partial density operators and linear operators on $\h$, respectively. The state space of a composite system (e.g., a quantum system consisting of multiple qubits) is the tensor product of the state spaces
of its components. For any $\rho$ in $\d(\h_1 \otimes \h_2)$, the partial traces
$\tr_{\h_1}(\rho)$ and $\tr_{\h_2}(\rho)$ are
the reduced quantum states of $\rho$ on $\h_2$ and $\h_1$, respectively. Here, the \emph{partial trace} $\tr_{\h_2}$ is a linear function from $\l(\h_1\otimes \h_2)$ to $\l(\h_1)$ such that for any $|\psi_i\>, |\phi_i\> \in \h_i$, $i=1,2$,
$$\tr_{\h_2}(|\psi_1\>\<\phi_1|\otimes |\psi_2\>\<\phi_2|) = 
\<\phi_2|\psi_2\> |\psi_1\>\<\phi_1|$$
where $\<\phi_2|\psi_2\>$ is the inner product of $|\phi_2\>$ and $|\psi_2\>$.

The \emph{evolution} of a closed quantum system is described by a unitary
operator on its state space: If the states of the system at
$t_1$ and $t_2$ are $\rho_1$ and $\rho_2$, respectively, then
$\rho_2=U\rho_1U^{\dag}$ for some unitary $U$. Here $U^\dag$ is the \emph{adjoint} operator of $U$; that is, $U^\dag$ is the \emph{unique} operator such that $\<\psi|U^\dag |\phi\> = \<\phi|U |\psi\>^*$ and $z^*$ denotes the conjugate of complex number $z$. Typical unitary operators used throughout this paper include Pauli-$X$, $Y$, and $Z$, Hadamard $H$, and controlled-NOT (CNOT) operator $CX$ and two special gates $U_{c}$ ($c \in \mathbb{C}$) and $B$, represented in the matrix form (with respect to the computational basis $\{|0\>, |1\>\}$) respectively as follows:
\[
X\define \begin{bmatrix}
	0 & 1 \\
	1 & 0 
\end{bmatrix},\ Y\define \begin{bmatrix}
	0 & -i \\
	i & 0 
\end{bmatrix},\ Z\define \begin{bmatrix}
	1 & 0 \\
	0 & -1
\end{bmatrix},\ U_{c}\define \frac{1}{\sqrt{1+|c|^2}}\begin{bmatrix}
	c & 1 \\
	1 & -c^*
\end{bmatrix},
\]
and
\[H\define \frac{1}{\sqrt{2}} \begin{bmatrix}
	1 & 1 \\
	1 & -1
\end{bmatrix},\ CX\define \begin{bmatrix}
	1 & 0 & 0 & 0 \\
	0 & 1 & 0 & 0 \\
	0 & 0 & 0 & 1 \\
	0 & 0 & 1 & 0
\end{bmatrix},\ 
B\define \begin{bmatrix}
	1 & 0 & 0 & 0 \\
	0 & \frac{1}{\sqrt{2}} & \frac{1}{\sqrt{2}} & 0 \\
	0 & \frac{1}{\sqrt{2}} & -\frac{1}{\sqrt{2}} & 0 \\
	0 & 0 & 0 & 1
\end{bmatrix}
.\]
A (projective) quantum {\it measurement} $\m$ is described by a
collection $\{P_i : i\in O\}$ of projectors (hermitian operators with eigenvalues being either 0 or 1) in the state space $\h$, where $O$ is the set of measurement outcomes. It is required that the
measurement operators $P_i$'s satisfy the completeness equation
$\sum_{i\in O}P_i = I_\h$, the identity operator on $\h$. If the system was in state $\rho$ before measurement, then the probability of observing the outcome $i$ is given by
$p_i=\tr(P_i\rho),$ and the state of the post-measurement system
becomes $\rho_i = P_i\rho P_i/p_i$ whenever $p_i>0$. Sometimes we use
a hermitian operator $M$ in $\lh$ called \emph{observable} to represent a projective measurement. To be specific, let 
$
M=\sum_{m\in \mathit{spec}(M)}mP_m
$ 
where $\mathit{spec}(M)$ is the set of eigenvalues of $M$, and $P_m$ the projector onto the eigenspace associated with $m$. Then the projective measurement determined by $M$ is $\{P_m : m\in \mathit{spec}(M)\}$. Note that by the linearity of the trace function,
the expected value of outcomes when $M$ is measured on state $\rho$ is calculated as
$
\sum_{m\in \mathit{spec}(M)} m\cdot  \tr(P_m \rho) = \tr(M\rho).
$

Finally, the dynamics that can occur in a (not necessarily closed) physical system are described by a completely positive and trace-nonincreasing super-operator. Here a linear operator $\e$ from $\l(\h_1)$ to $\l(\h_2)$ is called a \emph{super-operator}, and it is said to be (1) \emph{positive} if it maps positive operators on $\h_1$ to positive operators on $\h_2$; (2) \emph{completely positive} if $\mathcal{I}_\h\otimes \e$ is positive for all finite dimensional Hilbert space $\h$, where $\mathcal{I}_\h$ is the identity super-operator on $\lh$; (3) \emph{trace-preserving} (resp. \emph{trace-nonincreasing}) if 
$\tr(\e(A)) = \tr(A)$ (resp. $\tr(\e(A)) \leq \tr(A)$ for any positive operator $A\in \l(\h_1)$. For simplicity, in this paper all super-operators are assumed to be completely positive and trace-nonincreasing unless otherwise stated. Typical examples of trace-preserving super-operators include the unitary transformation $\e_U(\rho)\define U\rho U^\dag$ and the state transformation $
 \e_\m(\rho) \define \sum_{i\in O} p_i \rho_i = \sum_{i\in O} P_i \rho P_i 
 $ caused by a measurement $\{P_i : i\in O\}$, when all the post-measurement states are taken into account. 
Note that for each $i$, the evolution $
\e_i(\rho) \define  p_i \rho_i = P_i \rho P_i 
$
is also a super-operator; it is in general trace-nonincreasing. Denote by $\mathcal{SO}(\h)$ the set of super-operators on $\h$.

	\subsection{Introduction to quantum logic}
	Birkhoff and von Neumann proposed quantum logic in~\cite{Birkhoff1936TheLO}, and since then there has been a lot of effort exploring the proof systems for Hilbert subspaces. However, existing works have not leveraged the power of quantum logic to reason about quantum programs. This work,  for the first time, uses the Sasaki implication and Sasaki conjunction (a.k.a. Sasaki hook and Sasaki projection) in quantum refinement calculus, instead of the direct calculation of subspaces. In this section, we give a brief introduction to quantum logic.

	Given a finite dimensional Hilbert space $\h$, let $\s(\h)$ be the set orthogonal projectors (operators with eigenvalues being either 0 or 1) on $\h$. Note that there is a one-to-one correspondence between the projectors in $\h$ and the subspaces in $\h$. In this paper, we do not distinguish these two notions for the sake of simplicity. It is well known that the poset 
	$$\left(\s(\h), \le, \wedge, \vee, \bot\define\{0\}, \top\define\s(\h), {}^\bot\right)$$ forms an atomic, completely atomistic, complete, orthocomplemented, orthomodular, and non-distributive\footnote{Suppose $\l$ is a lattice. $\l$ is called bounded if it has the least element $\bot$ and the greatest element $\top$. An element $a\in \l$ is an atom if $a \gt \bot$ and for any $b$, $b \le a$ implies $b = \bot$ or $b = a$. $\l$ is atomic if for every $b \gt \bot$, there exists an atom $a$ such that $a \le b$. $\l$ is completely atomistic if every $b \gt \bot$, $b = \bigvee_i a_i$ where $a_i$ range over all atoms less equal to $b$. $\l$ is complete if every subset of $\l$ has both a greatest lower bound and a least upper bound. $\l$ is modular if it satisfies the modular law, that is, for all $a,b,c$, if $a\le c$ then $a \vee (b \wedge c) = (a \vee b) \wedge c$. $\l$ is orthocomplemented if for any $a\in \l$, it holds $a^\bot \vee a = 1 $ and $a^\bot \wedge a = 0$ (complement law), $a^{\bot\bot}=a$ (involution law), and if $a \le b$ then $b^\bot \le a^\bot$ (order-reversing). $\l$ is orthomodular if it satisfies the orthomodular law, i.e., for all $a,b$, if $a\le b$ then $a \vee (a^\bot \wedge b) = b$. It is evident that distributivity implies modularity, which, in turn, implies orthomodularity. } 
	lattice \cite{Re13} where the meet $\wedge$ and the join $\vee$ are defined as the intersection and sum (union) of subspaces, respectively, and ${}^\bot$ denotes the orthogonal complement in $\h$. Furthermore, it is a modular lattice given that we have assumed that $\h$ is finite-dimensional.
	Finally, the L\"{o}wner order $\le$ between projectors coincides with the set inclusion order between the corresponding subspaces; that is, $P\le Q$ (regarding as projectors) iff $P\subseteq Q$ (regarding as subspaces). 
	
	\begin{definition}[Sasaki implication and conjunction]
		Given projectors $P,Q$, the Sasaki implication $\rightsquigarrow$ and the Sasaki conjunction $\doublecap$ are defined as:
		\[
		P \rightsquigarrow Q = P^\bot\vee(P\wedge Q)\qquad \mbox{and}\qquad P \doublecap Q = P \wedge (P^\bot \vee Q).
		\]
	\end{definition}
\newcommand{\oA} {{A^\bot}}
\newcommand{\oB} {{B^\bot}}
\newcommand{\oC} {{C^\bot}}
\newcommand{\si}{\;\!\!\!\rightsquigarrow\;\!\!\!}
\renewcommand{\sc}{\doublecap}
	Sasaki implication can be regarded as the quantum extension of implication since it satisfies Birkhoff-von Neumann requirement $P \rightsquigarrow Q = I$ if and only if $P \sqsubseteq Q$, and the compatible import-export law, i.e., if $P$ commutes with $Q$\footnote{$P$ commutes with $Q$ if $P = (P \wedge Q) \vee (P \wedge Q^\bot)$.}, then for any $R$, $P \wedge Q \sqsubseteq R$ if and only if $P \sqsubseteq Q \rightsquigarrow R$. Sasaki conjunction has also been widely explored since it is the skew join of the lattice. It is known that the Sasaki conjunction is left-distributive with respect to the join of subspaces, whereas the Sasaki implication is left-distributive with respect to the meet of subspaces. Here, we list some useful properties for $\rightsquigarrow$ and $\doublecap$.
    \begin{lemma}\label{lem:propsasaki}
      Suppose $A, B, C$ are projectors. Then
      \begin{enumerate}
        \item $(A\sc B)^\bot = A\si\oB,\qquad(A\si B)^\bot = A\sc \oB$;
        \item $A\sc (B_1\vee B_2) = (A\sc B_1)\vee (A \sc B_2),\qquad A\si (B_1\wedge B_2) = (A\si B_1)\wedge (A \si B_2)$;
        \item $A\sc \bigvee_i B_i = \bigvee_i(A\sc B_i),\qquad A\si \bigwedge_i B_i = \bigwedge_i (A\si B_i)$ \quad for a finite or infinite set of $i$;
        \item $B_1\le B_2$ implies $A\sc B_1\le A\sc B_2$ and $A\si B_1\le A\si B_2$;
        \item $A\le B\Leftrightarrow A\si B = I$,\qquad $B\le \oA \Leftrightarrow A\sc B = 0$;
        \item $\oA\le B \Leftrightarrow \oB\le A \Leftrightarrow A\si B = B$,\qquad  $\oA\vee \oB = I \Leftrightarrow A\wedge B = 0 \Leftrightarrow A\si B = \oA$; 
        \item $B\le A \Leftrightarrow \oA\le\oB \Leftrightarrow A\sc B = B$,\qquad  $\oA\vee B = I \Leftrightarrow A \wedge \oB = 0 \Leftrightarrow A\sc B = A$; 
        \item $B\si \left(A\si \left[(A\si B)\si C\right] \right) = A\si \left[(A\si B)\si C\right]$.
      \end{enumerate}
      Further properties are shown in Fig. \ref{fig:sasaki-rules}, c.f. \cite{Gabriels2017}.
      \begin{figure}
        \footnotesize
        {
            \renewcommand{\arraystretch}{1.5}
            \begin{flushleft}\fbox{\textsf{Sasaki conjunction including two variables}}\end{flushleft}
            \begin{align*}
                0 = (A \sc A^\bot) &= (A\sc B)\sc \oA = A\sc (\oA\sc B) \qquad A = A\sc A \\
                A\sc B = A\sc (&A\sc B) = A\sc (B\sc A) = (A\sc B)\sc A = (A\sc B)\sc B = (A\sc B)\sc (A\sc B)= (A\sc B)\sc (B\sc A) \\
                A \sc (B \sc \oA) &= (A \sc B) \sc \oB = (A \sc B) \sc (A \sc \oB) = A \sc (B \sc (A \sc \oB)) = (A \sc B) \sc (B \sc \oA) \\
                &= (A \sc B) \sc (\oB \sc A) = ((A \sc B) \sc \oB) \sc A = ((A \sc B) \sc \oB) \sc B = A \sc (B \sc (\oA \sc \oB)) \\
                &= (A \sc (B \sc \oA)) \sc \oB = (A \sc (B \sc \oA)) \sc B = A \sc ((A \sc B) \sc \oB) = A \sc ((B \sc A) \sc \oA) 
            \end{align*}
            \begin{flushleft}\fbox{\textsf{Sasaki conjunction including three variables}}\end{flushleft}
            \begin{align*}
                (A\sc B)\sc C &= (A\sc B)\sc (A\sc C) = (A\sc B)\sc (C\sc A) = A\sc((A\sc B)\sc C)\\
                &= ((A\sc B)\sc C) \sc A = ((A\sc B)\sc C) \sc B = ((A\sc B)\sc C) \sc C \\
                A\sc(B\sc C) &= (A\sc B)\sc(B\sc C) = (A\sc (B\sc C))\sc A = (A\sc (B\sc C))\sc B\\
                &= A\sc((B\sc A)\sc C) = A\sc((B\sc C)\sc A) = A\sc(A\sc(B\sc C))
            \end{align*}
            \begin{flushleft}\fbox{\textsf{Sasaki implication including two variables}}\end{flushleft}
            \begin{align*}
                & A\si(A\si B) = A\si B && A\si(\oA\si B) = I && A\si(B\si\oA) = A\si\oB && A\si(B\si A) = (A\sc B)\si B\\
                & (A\si B)\si A = A && (A\si B)\si \oA = \oA\vee\oB && (A\si B)\si B = \oA\si B && (A\si B)\si \oB = (\oA\si\oB)\sc\oB
            \end{align*}
            \begin{flushleft}\fbox{\textsf{Sasaki implication and Sasaki conjunction including two variables}}\end{flushleft}
            \begin{align*}
                & A\si(A\sc B) = \oA\vee B && A\sc (A\si B) = A\wedge B  && (A\sc B)\si A = I              && (A\si B)\sc A = A\wedge B \\
                & A\si(\oA\sc B) = \oA     && A\sc (\oA\si B) = A        && (A\sc B)\si \oA = A\si\oB      && (A\si B)\sc \oA = \oA \\
                & A\si(B\sc A) = A\si B    && A\sc (B\si A) = A          && (A\sc B)\si B = A\si\oB        && (A\si B)\sc B = (\oA\si\oB)\si B \\
                & A\si(B\sc \oA) = \oA     && A\sc (B\si \oA) = A\sc \oB && (A\sc B)\si \oB = A\si(B\si A) && (A\si B)\sc \oB = \oA\sc\oB
            \end{align*}
        }
        \caption{Properties of Sasaki conjunction and Sasaki implication.}
        \label{fig:sasaki-rules}
    \end{figure}
    \end{lemma}
    \begin{proof}
        (1) - (7) are trivial; see \cite{beran2012orthomodular,Re13}.
        For (8), it suffices to show $$B \ge A\wedge B = A\sc (A\si B) \ge A\sc[(A\si B)\sc \oC] = (A\si [(A\si B) \si C])^\bot.$$ The result then follows from (5).
        
        We now turn to the properties shown in Fig.~\ref{fig:sasaki-rules}. Those with only two variables can be checked via \cite{hycko2005computation}. The equalities of the Sasaki conjunction with three variables can be easily verified using Theorem 4.1 in \cite{Gabriels2017}, except $A\sc(B\sc C) = A\sc((B\sc A)\sc C)$. To show this one, we first claim that
        $$|\psi\>\in A\sc B \quad\mbox{if and only if}\quad\exists\ |\phi\>\in B, \mbox{ s.t. } |\psi\> = A|\phi\>.$$
        For the $\Rightarrow$ direction, suppose $|\psi\>\in A \wedge (\oA \vee B)$. Then $A|\psi\> = |\psi\>$
          and $\exists\ |\psi_A^\bot\>\in\oA$ and $|\psi_B\>\in B$ such that $|\psi\> = |\psi_A^\bot\> + |\psi_B\>$. Note that $A|\psi\> = |\psi\> = A(|\psi_A^\bot\> + |\psi_B\>) = A|\psi_B\>$. It suffices to take $|\phi\> \define |\psi_B\>$. 
          
         For the opposite direction, we have $|\psi\> = A|\phi\> \in A$. Note that $A|\phi\> = |\phi\> - \oA|\phi\>$, $|\phi\>\in B$, and $\oA|\phi\>\in \oA$. Thus, $|\psi\> = A|\phi\> \in B\vee \oA$, and we conclude that $|\psi\>\in A \wedge (\oA \vee B) = A \doublecap B$.
          
        Using the claim twice, we have $|\psi\>\in A\sc(B\sc C)$ if and only if $\exists\ |\phi\>\in C, |\psi\> = AB|\phi\>$, and so $|\psi\>\in A\sc((B\sc A)\sc C)$ if and only if $\exists\ |\phi\>\in C, |\psi\> = A(B\sc A)|\phi\>$. However, note that $B = (B\wedge \oA)\vee (B\sc A)$ and $(B\wedge \oA) \perp (B\sc A)$. Thus, $AB = A((B\wedge \oA) \vee (B\sc A)) = A(B\sc A)$. 
    \end{proof}

	\section{Deferred proofs}
 \label{sec:deferred-proofs}
		
		\begin{proof}[Proof of Lemma~\ref{lem:densemantics}]
			The only nontrivial case is for the while loop, in which we have to prove that the infinite sum
			$
			\sum_{i=0}^{\infty} \p^\bot_{\bar{q}}\circ \e_i \circ \p_{\bar{q}}\circ \ldots \e_1\circ \p_{\bar{q}}
			$
			always converges for any choice of $\e_1, \e_2, \ldots$ in $\sem{S}$. However, this is direct from the facts that (1) the partial sums are nonincreasing with respect to the partial order $\le$ defined as follows: $\e\le \f$ iff $\f-\e$ is completely positive; and (2) the set of super-operators on $\h_{\qv(S)}$ forms a CPO under $\le$~\cite{ying2016foundations}.
		\end{proof}
		
		\begin{proof}[Proof of Lemma~\ref{lem:equivpcorrectness}]
			The first `iff' follows from the equivalent condition Eq.~\eqref{eq:corequiv} of $\models \ass{P}{S}{Q}$ while the second `iff' from the definition equation Eq.~\eqref{eq:defcor}.
		\end{proof}
		
		\begin{proof}[Proof of Lemma~\ref{lem:wlp}]
			
			Clauses (1)-(3) are easy to verify. 
			For (4), it suffices to prove that\\ $\mathcal{N}(PR^\bot P))=P^\bot\vee (P\wedge R)$. This can be seen from the fact that for any normalized state $|\eta\> = |\eta_1\> + |\eta_2\>$ where $|\eta_1\> \in P^\bot$ and $|\eta_2\>\in P$, 
			\begin{align*}
				\<\eta|P R^\bot P|\eta\> = 0 &\quad \mbox{iff}\quad \<\eta_2|R^\bot|\eta_2\> = 0\\
				&\quad \mbox{iff}\quad |\eta_2\> \in R\\
				&\quad \mbox{iff}\quad |\eta\> \in P^\bot \vee (P\wedge R).
			\end{align*}
			
			For (5), $rhs\le lhs$ follows from the fact that $\supp{\e(P)} \le Q$ iff $P\le \mathcal{N}(\e^\dag(Q^\bot))$. For the reverse part, there are two nontrivial cases to consider:
			\begin{itemize}
			\item Case 1: $Q\le R$ and $R\neq I$. Take $|\phi\>\in R^\bot$ and
				$
				\e \define 
					Set_{\phi}\circ \p^\bot_{\bar{q}} 
				$
				where $Set_{\phi}$ is the super-operator which sets the system $\bar{q}$ in state $|\phi\>$ and $\p^\bot$ denotes the super-operator associated with the projector $P^\bot$.
				Then $\supp{\e(P)} = 0 \le Q$, and so $\e\in  \sem{[P,Q]_{\bar{q}}}$. On the other hand, for any normalized state $|\eta\>\in \h_{\bar{q}}$,
				\begin{align*}
					\<\eta|\e^\dag(R^\bot)|\eta\> = 0 &\quad \mbox{iff}\quad \tr(R^\bot \e(|\eta\>\<\eta|)) = 0\\
					&\quad \mbox{iff}\quad \supp{\e(|\eta\>\<\eta|)} \le R\\
					&\quad \mbox{iff}\quad |\eta\> \in P.
				\end{align*} 
				Thus $wlp.[P,Q]_{\bar{q}}.R \le \sker{\e^\dag(R^\bot)} =  P$.
				\item Case 2: $Q\not\le R$ and $R\neq I$. Thus $Q\neq 0$. Take a normalised state $|\psi\>\in Q$ but $|\psi\>\not \in R$, and 
				$\e \define Set_{\psi}$.
				Then $\supp{\e(P)} \le Q$ and $\tr(R^\bot |\psi\>\<\psi|) >0$, and so $\e\in  \sem{[P,Q]_{\bar{q}}}$. Furthermore, for any normalized state $|\eta\>$, 
				\[
				\<\eta| \e^\dag(R^\bot)|\eta\> = \tr(R^\bot \e(|\eta\>\<\eta|)) = \tr(R^\bot |\psi\>\<\psi|) >0.
				\] Thus $\mathcal{N}(\e^\dag(R^\bot))=0$, and so $wlp.[P,Q]_{\bar{q}}.R =0$ as well.
			\end{itemize}
			Clause (6) follows easily from the fact that $$\sker{p\e_0^\dag(R^\bot)+(1-p)\e_1^\dag(R^\bot)}  = \sker{\e_0^\dag(R^\bot)} \wedge \sker{\e_1^\dag(R^\bot)}$$
			for any $0<p<1$.
			For (7), we compute that 
			\begin{align*}
				wlp.S_0.(wlp.S_1.R) & = \bigwedge\left\{\sker{\e_0^\dag\left(\bigvee\left\{\supp{\e_1^\dag(R^\bot)} : \e_1\in \sem{S_1}\right\}\right)}: \e_0\in \sem{S_0}\right\}\\
				& = \bigwedge\left\{\sker{\e_0^\dag\left(\supp{\e_1^\dag(R^\bot)}\right)} : \e_0\in \sem{S_0}, \e_1\in \sem{S_1}\right\}\\
				&= \bigwedge\left\{\sker{\e_0^\dag\circ \e_1^\dag(R^\bot)} : \e_0\in \sem{S_0}, \e_1\in \sem{S_1}\right\}\\
				&=wlp.(S_0;S_1).R
			\end{align*}
			where the first equality follows from the fact that $\supp{A}^\bot = \sker{A}$ for any hermitian operator $A$, the second from the fact that $\sker{\e(\bigvee_i R_i)} = \bigwedge_i \sker{\e(R_i)}$, and the third from $\sker{\e(\supp{\f(T)})}
			= \sker{\e\circ \f(T)}$, all of which are easy to verify.
			
			For (8), we compute that 
			\begin{align*}
				lhs &  = \bigwedge\left\{\sker{P\e_1^\dag(R^\bot)P + P^\bot\e_0^\dag(R^\bot)P^\bot} : \e_0\in \sem{S_0}, \e_1\in \sem{S_1}\right\}\\
				&=  \bigwedge\left\{\sker{P\e_1^\dag(R^\bot)P} : \e_1\in \sem{S_1}\right\} \wedge  \bigwedge\left\{\sker{P^\bot\e_0^\dag(R^\bot)P^\bot} : \e_0\in \sem{S_0}\right\}\\
				&= wlp.\left(\assert P[\bar{q}]; S_1\right).R\ \wedge\  wlp.\left(\assert P^\bot[\bar{q}]; S_0\right).R\\
				&=wlp.\assert P[\bar{q}].\left(wlp.S_1.R\right)\ \wedge\  wlp.\assert P^\bot[\bar{q}].\left(wlp.S_0.R\right)\\
				&=rhs
			\end{align*}
			
			For~\ref{lemitem:wlprn}, first let $S_0 \define \abort$ and $S_{n+1}\define \iif\ P[\bar{q}]\ \then\ S;S_n\ \eelse\ \sskip\ \pend$ for $n\geq 0$. It is then easy to prove by induction that for any $n$, $wlp.S_n.R =  R_n$ and $\sem{S_n} = \sum_{i=0}^n \sem{(\assert P[\bar{q}]; S)^i; \assert{P^\bot[\bar{q}]}}$. Thus
			\begin{align*}
				lhs &  = \bigwedge\left\{\sker{\sum_{i=0}^\infty \left[\p^\bot_{\bar{q}}\circ \e_i \circ \p_{\bar{q}}\circ \ldots \e_1\circ \p_{\bar{q}}\right]^\dag(R^\bot)} : 
				\forall i\geq 1, \e_i\in \sem{S}\right\}\\
				&=  \bigwedge_{n=0}^\infty \bigwedge\left\{\sker{\sum_{i=0}^n \left[\p^\bot_{\bar{q}}\circ \e_i \circ \p_{\bar{q}}\circ \ldots \e_1\circ \p_{\bar{q}}\right]^\dag(R^\bot)} : \e\in \sem{S}\right\}\\
				&= \bigwedge_{n=0}^\infty wlp.S_n.R = \bigwedge_{n=0}^\infty R_n = rhs. \qedhere
			\end{align*}	
		\end{proof}
		
		\begin{proof}[Proof of Lemma~\ref{lem:spost}]
			
			Again, clauses (1)-(3) are easy to verify. The proof of (4) is similar to Lemma~\ref{lem:wlp}(4) by noting that $\supp{PRP}^\bot = \sker{PRP} $. 			
			For (5), the part that $lhs\le rhs$ is easy. For the reverse part, there are two nontrivial cases to consider:
			\begin{itemize}
				\item Case 1: $R\le P$ and $R\neq 0$. Thus $P\neq 0$. Take $\rho$ such that $\supp{\rho}=Q$, and let $\e \define Set_{\rho}$ be the super-operator which sets the system $\bar{q}$ in state $\rho$.
				Then $Q=\supp{\e(P)} = \supp{\e(R)}$. Thus $\e\in  \sem{[P,Q]_{\bar{q}}}$, and $sp.[P,Q]_{\bar{q}}.R \ge  \supp{\e(R)} = Q$.
				\item Case 2: $R\not\le P$ and $R\neq 0$. Take $\rho$ such that $\supp{\rho}=I$, and 
				$
				\e \define Set_{\rho}\circ \p^\bot_{\bar{q}}
				$
				where $\p^\bot$ is the super-operator corresponding to the projector $P^\bot$.
				Then $\supp{\e(P)} = 0\le Q$, and hence $\e\in  \sem{[P,Q]_{\bar{q}}}$. However, $\supp{\e(R)} = I$. Thus, $sp.[P,Q]_{\bar{q}}.R \ge \supp{\e(R)} =I$.
			\end{itemize}
			(6) follows easily from the fact that for any partial density operators $\rho$ and $\sigma$, $\supp{\rho + \sigma} = \supp{\rho} \vee \supp{\sigma}$.
			
			For (7), we compute that 
			\begin{align*}
				sp.S_1.(sp.S_0.R) & = \bigvee\left\{\supp{\e_1\left(\bigvee\left\{\supp{\e_0(R)} : \e_0\in \sem{S_0}\right\}\right)}: \e_1\in \sem{S_1}\right\}\\
				& = \bigvee\left\{\supp{\e_1(\supp{\e_0(R)})} : \e_0\in \sem{S_0}, \e_1\in \sem{S_1}\right\}\\
				&= \bigvee\left\{\supp{\e_1\circ \e_0(R)} : \e_0\in \sem{S_0}, \e_1\in \sem{S_1}\right\}\\
				&=sp.(S_0;S_1).R
			\end{align*}
			where the second equality follows from the fact that $\supp{\e(\bigvee R_i)} = \bigvee \supp{\e(R_i)}$ while the third equality from $\supp{\e(\supp{\f(R)})}
			= \supp{\e\circ \f(R)}$, both of which are easy to verify.
			
			For (8), we compute that 
			\begin{align*}
				lhs &  = \bigvee\left\{\supp{\e_1(PRP) + \e_0(P^\bot RP^\bot)} : \e_0\in \sem{S_0}, \e_1\in \sem{S_1}\right\}\\
				&=  \bigvee\left\{\supp{\e_1(\supp{PRP})} : \e_1\in \sem{S_1}\right\} \vee  \bigvee\left\{\supp{\e_0(\supp{P^\bot RP^\bot})} : \e_0\in \sem{S_0}\right\}\\
				&=rhs
			\end{align*}
			where the last equality follows from the fact that $\supp{PRP} = P\doublecap R$ which we have proved in (4).
			
			The proof of (9) is quite involved, and it consists of three steps:
			\begin{itemize}
				\item Let $S_n$ be as defined in the proof of Lemma~\ref{lem:wlp}\ref{lemitem:wlprn}. Then 
				\begin{align*}
					lhs &  = \bigvee\left\{\supp{\sum_{i=0}^\infty \p^\bot_{\bar{q}}\circ \e_i \circ \p_{\bar{q}}\circ \ldots \e_1\circ \p_{\bar{q}}(R)} : 
					\forall i\geq 1, \e_i\in \sem{S}\right\}\\
					&=  \bigvee_{n=0}^\infty \bigvee\left\{\supp{\sum_{i=0}^n \p^\bot_{\bar{q}}\circ \e_i \circ \p_{\bar{q}}\circ \ldots \e_1\circ \p_{\bar{q}}(R)} : \e\in \sem{S}\right\}\\
					&= \bigvee_{n=0}^\infty sp.S_n.R.
				\end{align*}
				\item Let $f(X) \define sp.S.(P\doublecap X)$ for any projector $X$, so that $R_{n+1} = R\vee f(R_n)$. Then we have
				\begin{align*}
					f(R_1\vee R_2) &= sp.S.(P\doublecap(R_1\vee R_2)) = 
					sp.S.((P\doublecap R_1)\vee (P\doublecap R_2)) \\
					&= \bigvee\left\{\supp{\e((P\doublecap R_1)\vee (P\doublecap R_2))} : \e\in\sem{S}\right\} \\
					&= \bigvee\left\{\supp{\e(P\doublecap R_1)}\vee \supp{\e(P\doublecap R_2)} : \e\in\sem{S}\right\} \\
					&= \bigvee\left\{\supp{\e(P\doublecap R_1)} : \e\in\sem{S}\right\} \vee
					\bigvee\left\{\supp{\e(P\doublecap R_2)} : \e\in\sem{S}\right\}\\
					&= sp.S.(P\doublecap R_1)\vee sp.S.((P\doublecap R_2)\\
					& = f(R_1)\vee f(R_2).
				\end{align*}
				\item For any projector $X$, let $W_0(X) \triangleq 0,\ W_{n+1}(X) \triangleq X\vee f(W_n(X))$. Note from the second step that $$W_{n+1}(X) = X\vee f(X) \vee \cdots \vee f^n(X) = X\vee W_n(f(X)).$$
				We now prove by induction that for any $n\geq 0$ and projector $T$, $sp.S_n.T = P^\bot\doublecap W_n(T)$. The base case is trivial and the inductive step goes as follows:
				\begin{align*}
					sp.S_{n+1}.T &= 
					sp.S_n.(sp.S.(P\doublecap T)) \vee (P^\bot\doublecap T) & \mbox{(Definitions)}\\
					&= (P^\bot\doublecap W_n(f(T))) \vee (P^\bot\doublecap T) & \mbox{(Induction hypothesis)}\\
					&= P^\bot\doublecap (W_n(f(T)) \vee  T) & \mbox{(Lemma~\ref{lem:propsasaki})} \\
					&= P^\bot\doublecap W_{n+1}(T).  & \mbox{(Proven fact)}
				\end{align*}
				With this, the final result $lhs = rhs$ follows from the first step and the observation that $ R_n = W_n(R)$ for all $n\geq 0$.  \qedhere
			\end{itemize}
		\end{proof}

  		\begin{proof}[Proof of Theorem~\ref{thm:refinewlpspc}]
			From the definition of weakest liberal precondition (Definition~\ref{def:wlp}), $\models \ass{P}{S}{Q}$ iff $P\le wlp.S.Q$. Hence, the first part of the theorem follows straightforwardly. The second part can be reasoned in a similar manner.
		\end{proof}

	\begin{proof}[Poof of Theorem~\ref{thm:psoundpre}]
			The soundness of the structural rules in Fig.~\ref{fig:pstructurerules} and the common rules in Fig.~\ref{fig:prulespres} can be easily proven from Lemma~\ref{lem:wlp} or Lemma~\ref{lem:spost}.
			
			For the rules based on the weakest preconditions in Fig.~\ref{fig:prulespres}, we take only the last two as examples; the others are easier. Consider (\textsc{Wpc-Pres-Cond}) first. Let $R$ be any projector. Then
			\[wlp.rhs.R =  (P\rightsquigarrow wlp.[P_1,Q].R)           \wedge (P^\bot \rightsquigarrow wlp.[P_0,Q].R).
			\]
			There are two nontrivial case to consider:
			\begin{itemize}
				\item Case 1: $Q\le R$ and $R\neq I$. Then 
				\[
				wlp.rhs.R =(P\rightsquigarrow P_1) \wedge (P^\bot \rightsquigarrow P_0)=wlp.lhs.R.
				\]
				\item Case 2: $Q\not\le R$ and $R\neq I$. Then 
				\[	wlp.rhs.R =(P\rightsquigarrow 0)              \wedge (P^\bot \rightsquigarrow 0)=             P^\bot\wedge P = 0 = wlp.lhs.R. 
				\]
			\end{itemize}
			Now, for (\textsc{Wpc-Pres-While}), let $T\define (P\rightsquigarrow Q) \wedge (P^\bot \rightsquigarrow R)$. We claim that for any $n\geq 0$, $T\le R_n$ where $R_0\define 0$ and $R_{n+1} \define \left(P\rightsquigarrow wlp.[Q, T].R_n\right)\wedge\left(P^\bot\rightsquigarrow R\right)$ is defined in a similar way as in Lemma~\ref{lem:wlp}\ref{lemitem:wlprn}.
			The case of $n=0$ is trivial. Furthermore, if $T\le R_n$ then
			\begin{align*}
				R_{n+1}  &=  \left(P\rightsquigarrow wlp.[Q, T].R_n\right) \wedge\left(P^\bot\rightsquigarrow R\right)\\
				&= \begin{cases}
					P^\bot\rightsquigarrow R & \mbox{if } R_n=I\\
					(P\rightsquigarrow Q) \wedge (P^\bot \rightsquigarrow R) & \mbox{o.w.}\\
				\end{cases}\\		
				&\ge T.
			\end{align*}
			With this claim, we have $T\le \bigwedge_{n\geq 0} R_n = wlp.rhs.R$, and the result follows from Theorem~\ref{thm:generalref}. 
			
			For the rules based on the strongest postconditions in Fig.~\ref{fig:prulespres}, again we only take the last two as examples. For (\textsc{Spc-Pres-Cond}), 		
			we compute
				\[
					sp.rhs.P =  \left(sp.[R\doublecap P,Q].(R\doublecap P)\right) \vee \left(sp.[R^\bot \doublecap P, Q].(R^\bot\doublecap P)\right) \le Q
				\]
			and the result follows from Theorem~\ref{thm:generalref}.  
			Now, for (\textsc{Spc-Pres-While}), we have from Lemma~\ref{lem:spost} that $sp.rhs.Inv =  \bigvee_{n\geq 0} \left(P^\bot\doublecap R_n\right)$ where $R_0\define 0$ and for $n\geq 0$, 
			\[
				R_{n+1} \define Inv\vee sp.\left[P\doublecap Inv, Inv\right]_{\bar{q}}.\left(P\doublecap R_n\right).
			\]
			We prove by induction that $R_n \le Inv$. The case of $n=0$ is trivial. For the induction step, note that $R_n \le Inv$ implies $P\doublecap R_n \le P\doublecap Inv$. Then $sp.\left[P\doublecap Inv, Inv\right]_{\bar{q}}.\left(P\doublecap R_n\right)\le Inv$, and so $R_{n+1} \le Inv$ as well. 
			
			Finally, we have $sp.rhs.Inv =  \bigvee_{n\geq 0} \left(P^\bot\doublecap R_n\right) \le P^\bot\doublecap Inv$, and the result follows from Theorem~\ref{thm:generalref}.
		\end{proof}
    
	\begin{proof}[Poof of Theorem~\ref{thm:pcompletepre}]
			From \textsc{(Comm-Pres-Cons)}, it suffices to prove that for any $Q$ and $S$ with $qv(S)\subseteq \bar{q}$,
			\[ 
			\vdash_{\textsc{Wpc}} [wlp.S.Q, Q]_{\bar{q}}\le S \qquad \mbox{and}\qquad \vdash_{\textsc{Spc}} [Q, sp.S.Q]_{\bar{q}}\le S.\]
			This can be proven by induction on the structure of $S$. We only elaborate on the case where $S\define \whilestm{P[\bar{q}]}{S'}$; others are easier.
			
			First, from Lemma~\ref{lem:wlp}, $wlp.S.Q  = \bigwedge_{n\geq 0}R_n$ where $R_0 \define I$ and $$R_{n+1}  \define  \left(P\rightsquigarrow wlp.S'.R_n\right) \wedge\left(P^\bot\rightsquigarrow Q\right).$$	
			Let $T \define wlp.S.Q$. Then it is easy to show from Lemma~\ref{lem:propsasaki} and the fact that $wlp.S'$ is Scott-continuous that $T= \left(P\rightsquigarrow wlp.S'.T\right) \wedge\left(P^\bot\rightsquigarrow Q\right)$. Thus
			\begin{align*}
				[wlp.S.Q, Q]_{\bar{q}} &=  [\left(P\rightsquigarrow wlp.S'.T\right) \wedge\left(P^\bot\rightsquigarrow Q\right), Q]_{\bar{q}} \\
				&\le \while\ P[\bar{q}]\ \ddo\ [wlp.S'.T, T]_{\bar{q}}\ \pend & \textsc{(Wpc-Pres-While)}\\
				&\le \while\ P[\bar{q}]\ \ddo\ S'\ \pend &  \textsc{(Stru-While)},\ \mbox{hypothesis}
			\end{align*}
		
		Similarly, from Lemma~\ref{lem:spost}, $sp.S.Q  = \bigvee_{n\geq 0}  \left(P^\bot\doublecap R_n\right)$ where $R_0\define 0$ and
		$$R_{n+1} \define Q \vee sp.S'.(P\doublecap R_n).$$
		Let $Inv \define \bigvee_{n\geq 0}  R_n$. Then $Q\le Inv$. Furthermore, from Lemma~\ref{lem:propsasaki} and the fact that $sp.S'$ is Scott-continuous, we know $sp.S.Q = P^\bot \doublecap Inv$ and $sp.S'. \left(P\doublecap Inv\right) \le Inv$. Thus 
		\begin{align*}
			[P\doublecap Inv, Inv]_{\bar{q}} &\le [P\doublecap Inv, sp.S'. \left(P\doublecap Inv\right) ]_{\bar{q}} & \textsc{(Comm-Pres-Cons)}  \\
			&\le S' &   \mbox{hypothesis}
		\end{align*}
		and so
		\begin{align*}
			[Q, sp.S.Q]_{\bar{q}} &\le  [Inv, P^\bot \doublecap Inv]_{\bar{q}} & \textsc{(Comm-Pres-Cons)} \\
			& \le  \while\ P[\bar{q}]\ \ddo\ [P\doublecap Inv, Inv]_{\bar{q}}\ \pend & \textsc{(Spc-Pres-While)} \\
			&\le \while\ P[\bar{q}]\ \ddo\ S'\ \pend & \textsc{(Stru-While)} 
		\end{align*}
		\end{proof}
  
    \begin{proof}[Proof of Theorem~\ref{thm:generalassert}]
        First, from Lemma~\ref{lem:equivpcorrectness} and Theorem~\ref{thm:generalref}, we only need to prove (1) $\{P\}; S \equiv \{P\}; S; \{Q\}$ iff $sp.S.P \le Q$ and (2) $S; \{Q\} \equiv \{P\}; S; \{Q\}$ iff $P^\bot \le wlp.S.Q^\bot$.
        
        For (1), let $T$ be a projector. Then $sp.\left(\{P\}; S\right).T = sp.S.(P\doublecap T)$ while
        \[
        sp.\left(\{P\}; S; \{Q\}\right).T = Q\doublecap sp.S.(P\doublecap T).
        \]
        Note that $X\doublecap Y = Y$ iff $Y\le X$ for any projectors $X$ and $Y$. We have
        \begin{align*}
            \{P\}; S \equiv \{P\}; S; \{Q\} & \qquad \mbox{iff}\qquad \forall T,\ sp.S.(P\doublecap T) \le Q\\
            & \qquad \mbox{iff}\qquad sp.S.P \le Q.
        \end{align*}
        Similarly, since $wlp.\left(S; \{Q\}\right).T = wlp.S.(Q\rightsquigarrow T)$ and
        \[
        wlp.\left(\{P\}; S; \{Q\}\right).T = P\rightsquigarrow wlp.S.(Q\rightsquigarrow T).
        \]
        Note that $X\rightsquigarrow Y = Y$ iff $X^\bot \le Y$. Then
        \begin{align*}
            S; \{Q\} \equiv \{P\}; S; \{Q\} & \qquad \mbox{iff}\qquad \forall T,\  P^\bot \le wlp.S.(Q\rightsquigarrow T)\\
            & \qquad \mbox{iff}\qquad P^\bot \le wlp.S.Q^\bot. \qedhere
        \end{align*}
    \end{proof}		
		
		\begin{proof}[Poof of Theorem~\ref{thm:psoundass}]
			We first prove the soundness of the common rules.
			Rule \textsc{(Comm-Asst-True)} is trivial and \textsc{(Comm-Asst-Post)} and \textsc{(Comm-Asst-Pre)} are easy from Theorem~\ref{thm:generalassert}.
			For \textsc{(Comm-Asst-Seq)}, we first note from \textsc{(Stru-Seq)} that 
			\begin{align*}
				\{P\};  S_1 \equiv \{P\}; S_1; \{Q\} & \qquad \mbox{implies} \qquad\{P\};  S_1; S_2\qquad \equiv \{P\}; S_1; \{Q\}; S_2\\
				\{Q\};  S_2 \equiv \{Q\}; S_2; \{R\}& \qquad \mbox{implies} \qquad\{P\}; S_1; \{Q\}; S_2  \equiv \{P\}; S_1; \{Q\}; S_2; \{R\} 
			\end{align*}
			The result then follows from the transitivity of $\equiv$ (or rule \textsc{(Stru-Tran)}).
			
			Next, we prove the soundness of the rules based on the weakest liberal preconditions. 
			
			\textsc{(Wpc-Asst-Pres)}: For any projector $X$, we compute
			\begin{align*}
				wlp.\left( \{P\}; [P\rightsquigarrow Q, R]\right).T & = P\rightsquigarrow wlp. [P\rightsquigarrow Q, R].T\\
				&= \begin{cases}
					P\rightsquigarrow I &\mbox{ if } T=I\\
					P\rightsquigarrow ( P\rightsquigarrow Q) & \mbox{ if } R\le T \mbox{ and } T\neq I\\
					P\rightsquigarrow 0 & \mbox{o.w.}
				\end{cases}\\
				&= P\rightsquigarrow wlp. [Q, R].T\\
				&=wlp.\left( \{P\}; [Q, R]\right).T 
			\end{align*}
			where the third equality follows from the fact that $P\rightsquigarrow ( P\rightsquigarrow Q) = P\rightsquigarrow Q$.
			
			 \textsc{(Wpc-Asst-Cond)}: For simplicity, we denote $P_1 \define P$ and $P_0 \define P^\bot$. Then from the assumptions, 
				\begin{align*}
					wlp.S_i.(R\rightsquigarrow X) &= (P_i\rightsquigarrow Q)\rightsquigarrow wlp.S_i.(R\rightsquigarrow X).
				\end{align*}
				Thus for any projector $T$,
				\begin{align*}
					&wlp.\left(\{Q\}; \iif\ P[\bar{q}]\ \then\ \{P\rightsquigarrow Q\}; S_1\ \eelse\ \{P^
					\bot\rightsquigarrow Q\}; S_0\ \pend; \{R\}\right).T\\
					& = Q\rightsquigarrow \bigwedge_{i=0}^1 \left(P_i\rightsquigarrow \left[(P_i\rightsquigarrow Q)\rightsquigarrow wlp.S_i.(R\rightsquigarrow T)\right]\right)\\
					&= \bigwedge_{i=0}^1  \left(Q\rightsquigarrow \left(P_i\rightsquigarrow \left[(P_i\rightsquigarrow Q)\rightsquigarrow wlp.S_i.(R\rightsquigarrow T)\right]\right)\right)\\
					&= \bigwedge_{i=0}^1  \left(P_i\rightsquigarrow \left[(P_i\rightsquigarrow Q)\rightsquigarrow wlp.S_i.(R\rightsquigarrow T)\right]\right)\\
					&= \bigwedge_{i=0}^1  \left(P_i\rightsquigarrow wlp.S_i.(R\rightsquigarrow T)\right)\\
					&=wlp.\left(\iif\ P[\bar{q}]\ \then\ S_1\ \eelse\ S_0\ \pend; \{R\}\right).T
				\end{align*}
				where  the second and third equalities follow from Lemma~\ref{lem:propsasaki}.

			\textsc{(Wpc-Asst-While)}: For any projector $T$, let  $R_0 =  R_0' =I$ and 
			\begin{align*}
				R_{n+1} &= (P\rightsquigarrow wlp.S.R_n) \wedge (P^\bot \rightsquigarrow ((P^\bot \rightsquigarrow Q)\rightsquigarrow T)), \\
				R_{n+1}' &= \left(P\rightsquigarrow \left[(P\rightsquigarrow Q) \rightsquigarrow wlp.S.(Q \rightsquigarrow R_n')\right]\right) \wedge (P^\bot \rightsquigarrow ((P^\bot \rightsquigarrow Q)\rightsquigarrow T))
			\end{align*}
			for any $n\geq 0$.
			Note that $P^\bot\rightsquigarrow (P^\bot\rightsquigarrow T) = P^\bot\rightsquigarrow T$.
			Then from Lemma~\ref{lem:wlp}\ref{lemitem:wlprn},
			\begin{align*}
				&wlp.\left(\{Q\}; \while\ P[\bar{q}]\ \ddo\ \{P\rightsquigarrow Q\}; S; \{Q\}\ \pend; \{P^\bot \rightsquigarrow Q\} \right).T\\
				& = Q \rightsquigarrow \bigwedge_{n\geq 0}R_n' = \bigwedge_{n\geq 0} \left(Q \rightsquigarrow R_n'\right)
			\end{align*}
			where the last equality follows from Lemma~\ref{lem:propsasaki}.
			We claim that for any $n\geq 0$, $Q \rightsquigarrow R_n' = R_n.$ The case where $n=0$ is easy since $(Q \rightsquigarrow I) = I$. Now we compute that
			\begin{align*}
				&Q \rightsquigarrow \left(P\rightsquigarrow \left[(P\rightsquigarrow Q) \rightsquigarrow wlp.S.(Q \rightsquigarrow R_n')\right]\right)  \\
				&= Q \rightsquigarrow \left(P\rightsquigarrow \left[(P\rightsquigarrow Q) \rightsquigarrow wlp.S.R_n\right]\right)  \\
				&= P\rightsquigarrow \left[(P\rightsquigarrow Q) \rightsquigarrow wlp.S.R_n\right]\\
				& = P\rightsquigarrow wlp.S.R_n
			\end{align*}
			where the first equality follows from the induction hypothesis, the second from Lemma~\ref{lem:propsasaki}, and the third from the assumption that $ S; \{Q\} \equiv \{P\rightsquigarrow Q\}; S; \{Q\}$. Furthermore, using Lemma~\ref{lem:propsasaki} again, we have
			\begin{align*}
				&Q \rightsquigarrow \left(P^\bot \rightsquigarrow \left[(P^\bot \rightsquigarrow Q)\rightsquigarrow T\right]\right) = P^\bot \rightsquigarrow \left[(P^\bot \rightsquigarrow Q)\rightsquigarrow T\right].
			\end{align*}
			Combining the above two derivatives, we have $Q \rightsquigarrow R_{n+1}' = R_{n+1}$ from Lemma~\ref{lem:propsasaki}.
			With the claim, we have $wlp.rhs.T = \bigwedge_{n\geq 0} R_n = wlp.lhs.T$.

			Finally, we prove the soundness of rules based on the strongest postconditions.

			 \textsc{(Spc-Asst-Post)}: The first equivalence follows directly from the first part of Theorem~\ref{thm:generalassert}. For the second equivalence, we note that $\left(wlp.\{P\}.Q^\bot\right)^\bot = \left(P\rightsquigarrow Q^\bot\right)^\bot = P\doublecap Q$.

			\textsc{(Spc-Asst-Pres)}:  For any projector $T$, we compute
			\begin{align*}
				sp.\left([Q, P\doublecap R]; \{P\}\right).T & = P\doublecap sp. [Q, P\doublecap R].T\\
				&= \begin{cases}
					P\doublecap 0 &\mbox{ if } T=0\\
					P\doublecap (P\doublecap R) & \mbox{ if } T\le Q \mbox{ and } T\neq 0\\
					P\doublecap I & \mbox{o.w.}
				\end{cases}\\
				&= P\doublecap sp. [Q, R].T\\
				&=sp.\left([Q, R];\{P\}\right).T
			\end{align*}
			where the third equality follows from the fact that $P\doublecap ( P\doublecap Q) = P\doublecap Q$.
			
			\textsc{(Spc-Asst-Cond)}: For simplicity, we denote $P_1 \define P$ and $P_0 \define P^\bot$. Then from the assumptions, 
				\begin{align*}
					sp.S_i.((P_i\doublecap Q)\doublecap X) &= R\doublecap sp.S_i.((P_i\doublecap Q)\doublecap X).
				\end{align*}
				Thus for any projector $T$,
				\begin{align*}
					&sp.\left(\{Q\}; \iif\ P[\bar{q}]\ \then\ \{P\doublecap Q\}; S_1\ \eelse\ \{P^
					\bot\doublecap Q\}; S_0\ \pend; \{R\} \right).T\\
					& = R\doublecap \left(\bigvee_{i=0}^1 sp.S_i.\left((P_i\doublecap Q)\doublecap (P_i\doublecap (Q\doublecap T))\right)\right)\\
					&= \bigvee_{i=0}^1  \left(R\doublecap sp.S_i.\left((P_i\doublecap Q)\doublecap (P_i\doublecap (Q\doublecap T))\right)\right)\\
					&= \bigvee_{i=0}^1  sp.S_i.\left((P_i\doublecap Q)\doublecap (P_i\doublecap (Q\doublecap T))\right)\\
					&= \bigvee_{i=0}^1  sp.S_i.\left(P_i\doublecap (Q\doublecap T)\right)\\
					&=sp.\left(\{Q\}; \iif\ P[\bar{q}]\ \then\ S_1\ \eelse\ S_0\ \pend\right).T
				\end{align*}
				where the second equality follows from Lemma~\ref{lem:propsasaki}, and the fourth from the fact that $P_i\doublecap (Q\doublecap T) \le P_i\doublecap Q$.
				
				\textsc{(Spc-Asst-While)}: Let $S_n$ be defined as in the proof of Lemma~\ref{lem:wlp}~\ref{lemitem:wlprn},  $S'_0 \define \abort$, and $$S'_{n+1}\define \iif\ P[\bar{q}]\ \then\  \{P\doublecap Inv\}; S; \{Inv\}; S'_n\ \eelse\ \sskip\ \pend$$ for $n\geq 0$. Then for any projector $T$,
				\begin{align}
					&sp.\left(\{Inv\}; \while\ P[\bar{q}]\ \ddo\ \{P\doublecap Inv\}; S; \{Inv\}\ \pend; \{P^\bot\doublecap Inv\} \right).T \notag\\
					& = (P^\bot\doublecap Inv)\doublecap \left(\bigvee_{n\geq 0} sp.S'_n.(Inv\doublecap T)\right)\notag\\
					& = \bigvee_{n\geq 0} (P^\bot\doublecap Inv)\doublecap \left(sp.S'_n.(Inv\doublecap T)\right) \tag{*}
				\end{align}
				where the last equality follows from Lemma~\ref{lem:propsasaki}.
				We claim that for any $n\geq 0$ and any $X$,
				\begin{equation*}
					(P^\bot\doublecap Inv)\doublecap \left(sp.S'_n.(Inv\doublecap X)\right) = sp.S_n.(Inv\doublecap X).
				\end{equation*}
				The case where $n=0$ is easy, since $S_0 = S'_0 = \abort$. Now we compute that
				\begin{align*}
					&(P^\bot\doublecap Inv)\doublecap \left(sp.S'_{n+1}.(Inv\doublecap X)\right)\\
					& = (P^\bot\doublecap Inv)\doublecap \left( sp.S'_n.\left[Inv\doublecap \{sp.S.\left[(P\doublecap Inv)\doublecap (P\doublecap (Inv\doublecap X))\right]\}\right]\vee \left(P^\bot \doublecap (Inv\doublecap X)\right)\right)\\
					& = \left( (P^\bot\doublecap Inv)\doublecap  sp.S'_n.\left[Inv\doublecap \{sp.S.\left[(P\doublecap Inv)\doublecap (P\doublecap (Inv\doublecap X))\right]\}\right]\right)\vee \left(P^\bot \doublecap (Inv\doublecap X)\right)\\
					& = sp.S_n.\left[Inv\doublecap \{sp.S.\left[(P\doublecap Inv)\doublecap (P\doublecap (Inv\doublecap X))\right]\}\right]\vee \left(P^\bot \doublecap (Inv\doublecap X)\right)\\
					& = sp.S_n.\left(sp.S.\left[(P\doublecap Inv)\doublecap (P\doublecap (Inv\doublecap X))\right]\right)\vee \left(P^\bot \doublecap (Inv\doublecap X)\right)\\
					& = sp.S_n.\left(sp.S.\left[P\doublecap (Inv\doublecap X)\right]\right)\vee \left(P^\bot \doublecap (Inv\doublecap X)\right)\\
					& = sp.S_{n+1}.(Inv\doublecap X)
				\end{align*}
				where the second equality follows from Lemma~\ref{lem:propsasaki}, the third from the induction hypothesis, and the fourth from the assumption that $\{P\doublecap Inv\};  S \equiv \{P\doublecap Inv\}; S; \{Inv\}$.
				
				With the claim, we have
				\begin{align*}
					(*)& =\bigvee_{n\geq 0} sp.S_n.(Inv\doublecap T)=sp.\left(\{Inv\}; \pwstm\right).T. \qedhere
				\end{align*}
		\end{proof}

		\section{Constructions for full quantum-to-quantum Bernoulli factory}
		\label{app:q2q-bernoulli}
		{
			\renewcommand{\arraystretch}{1.3}
			\begin{figure}
				\footnotesize
				\begin{tabular}{rl}
					$\mathit{Sum}[S_1,\cdots, S_n]\triangleq$ & \hspace{-0.3cm} $\mathit{Add}[\mathit{Add}[\cdots \mathit{Add}[S_1, S_2], \cdots], S_n]$ \\
					$\mathit{Exp}_n\triangleq$ & \hspace{-0.3cm} $\underbrace{\mathit{Mul}[\mathit{Mul}[\cdots \mathit{Mul}[\mathit{Init}_p, \mathit{Init}_p], \cdots], \mathit{Init}_p]}_{\mathrm{repeat\ }n \mathrm{\ times}}$ \quad $(n\geq 2)$\\
					$\mathit{Poly}_g\triangleq$ & \hspace{-0.3cm} $\mathit{Sum}[\mathit{Cst}_{\alpha_0}, \mathit{Mul}[\mathit{Cst}_{\alpha_1},\mathit{Init}_p], \mathit{Mul}[\mathit{Cst}_{\alpha_2},\mathit{Exp}_2], \cdots,  \mathit{Mul}[\mathit{Cst}_{\alpha_n},\mathit{Exp}_n]]$\\
					$\mathit{Simul}_h\triangleq$ & \hspace{-0.3cm} $\mathit{Add}[ \mathit{Mul}[ \mathit{Mul}[ \mathit{Poly}_{g_1}, \mathit{Inv}[\mathit{Poly}_{g_2}] ], \qcoin] , \mathit{Mul}[ \mathit{Poly}_{g_3}, \mathit{Inv}[\mathit{Poly}_{g_4}] ] ]$
				\end{tabular}
				\caption{Further constructs for quantum-to-quantum Bernoulli factory, where
					$g(p) = \sum_i\alpha_ip^i$ with $\alpha_i\in\mathbb{C}$ is polynomial of $p$, $h(p) = \frac{g_1(p)}{g_2(p)}\sqrt{\frac{p}{1-p}} + \frac{g_3(p)}{g_4(p)}$ is any quantum-to-quantum simulable function\cite{JZS18} with $g_i$ being polynomials of $p$.
				}
				\label{fig:q2q-bernoulli-A}
			\end{figure}
		}

		Based on the constructs presented in Fig. \ref{fig:q2q-bernoulli}, we further construct the programs $\mathit{Sum}$, $\mathit{Exp}_n$, $\mathit{Poly}_g$ and $\mathit{Simul}_h$ such that
		{
			\small
			\begin{align*}  
				&[I, |h_1+\cdots+h_n\>]_q \le \mathit{Sum}[F_{h_1},\cdots,F_{h_n}](q)\quad 
				&&[I, |p^n\>]_q \le \mathit{Exp}_n(q)  \\     
				&[I, |g\>]_q \le \mathit{Poly}_g(q)  \quad  
				&&[I, |h\>]_q \le \mathit{Simul}_h(q).
			\end{align*}
		}       
		
		$\mathit{Sum}$ and $\mathit{Exp}_n$ are refined by calling $\mathit{Init}_p$, $\mathit{Mul}$ and $\mathit{Add}$ repeatedly\footnote{We can also introduce an auxiliary variable that counts the iteration number rather than repeatedly calling subroutines in program scheme. Since $n$ and $h_i$ are regarded as parameters, we keep them in these simple forms.}:
		\begin{align*}
			[I, |h_1+\cdots+h_n\>] &\le \mathit{Add}[\mathit{Add}[\cdots \mathit{Add}[F_{h_1}, F_{h_2}], \cdots], F_{h_n}]\\
			[I, |p^n\>] &\le \mathit{Mul}[F_{p^{n-1}}, F_{p}] \le \mathit{Mul}[\mathit{Mul}[\cdots \mathit{Mul}[F_{p}, F_{p}], \cdots], F_{p}] \\
			&\le \underbrace{\mathit{Mul}[\mathit{Mul}[\cdots \mathit{Mul}[\mathit{Init}_p, \mathit{Init}_p], \cdots], \mathit{Init}_p]}_{\mathrm{repeat\ }n \mathrm{\ times}}.
		\end{align*}
		Then for any polynomial $g(p) = \alpha_0 + \alpha_1p + \sum_{i=2}^n\alpha_{i}p^{i}$, we give $\mathit{Poly}_g$ that refines $[I, |g\>]$:
		\begin{align*}
			[I, |g\>] &\le \mathit{Sum}[\mathit{Cst}_{\alpha_0}, \mathit{Mul}[\mathit{Cst}_{\alpha_1},\mathit{Init}_p], 
			\mathit{Mul}[\mathit{Cst}_{\alpha_2},\mathit{Exp}_2], \cdots,  \mathit{Mul}[\mathit{Cst}_{\alpha_n},\mathit{Exp}_n]].
		\end{align*}
		Finally, we construct the program $\mathit{Simul}_h$ that refine $[I, |h\>]$, i.e., it produces $|h\>$, where $h$ is an arbitrary simulable function of the form
		$h(p) = \frac{g_1(p)}{g_2(p)}\sqrt{\frac{p}{1-p}} + \frac{g_3(p)}{g_4(p)}$ with $g_i(p)$ being polynomials of $p$:
		\begin{align*}
			[I, |h\>] &\le \mathit{Add}[ \mathit{Mul}[ \mathit{Mul}[ \mathit{Poly}_{g_1}, \mathit{Inv}[\mathit{Poly}_{g_2}] ], \qcoin] , \mathit{Mul}[ \mathit{Poly}_{g_3}, \mathit{Inv}[\mathit{Poly}_{g_4}] ] ].
		\end{align*}
		
		Now, we finish building all the remaining ingredients presented in Fig. \ref{fig:q2q-bernoulli-A}, as well as the final program $\mathit{Simul}_h$ that produces $|h\>$ for any simulable function $h$.

		\section{\texttt{Quire}: Rz rotation example}
        \label{sec:Rz in Quire}
		This section demonstrates the refinement of the Rz rotation example in Sec.~\ref{subsec:zrotation}.

		\subsubsection*{Start the server.}
		\texttt{Quire} is operated by commands. We can either invoke its interface in Python code or start a server for interactive proving. For the latter approach, the Python boot script is as follows:

		\lstset{style=py}
		\begin{lstlisting}[language=Python]
    import numpy as np
    from quire import quire_server, predefined
    opts = { "Rz" : predefined.Rz(np.arccos(3/5)) }
    quire_server("./code", "./output", opts)
        \end{lstlisting}
		Here the parameters \textcolor{PineGreen}{\textbf{\texttt{"./code"}}} and \textcolor{PineGreen}{\textbf{\texttt{"./output"}}} specify the file for \texttt{Quire} code input and the information output. The server will monitor the modifications made in \textcolor{PineGreen}{\textbf{\texttt{"./code"}}} and update the \textcolor{PineGreen}{\textbf{\texttt{"./output"}}} file. \texttt{opts} is a dictionary for string-matrix pairs which introduce user-defined operators into the prover. In our example, it is the target quantum gate \texttt{Rz}.
		
		\subsubsection*{Construct and test operators.}
		We start with defining new constants in the \textcolor{PineGreen}{\textbf{\texttt{"./code"}}} file. The following code constructs operators $P^{00}$ and $P^{\neq 00}$ and outputs their matrix representation. 

\begin{tikzpicture}
  \node[fill=gray!20, inner xsep=5pt, inner ysep =5pt] (box) {
    \begin{minipage}{0.96\textwidth}
      \lstset{style=Quire}
      \begin{lstlisting}
Def P00 := [|00$\>$]. // The projector on |00>
Def Pnot00 := P00^$\bot$.
Eval P00.
      \end{lstlisting}\end{minipage}
    };
\end{tikzpicture}

		\texttt{Quire} is able to calculate projective predicates and test the order between operators. As an example, the following command verifies the statement that $P^{00}_{p, q} \wedge |+\>_{p}\<+| \sqsubseteq \boldsymbol{0}$. 

\begin{tikzpicture}
  \node[fill=gray!20, inner xsep=5pt, inner ysep =5pt] (box) {
    \begin{minipage}{0.96\textwidth}
      \lstset{style=Quire}
      \begin{lstlisting}
Test P00[p q] $\wedge$ Pp[p] <= c0[]. // [p ... q] denotes quantum registers
        \end{lstlisting}
    \end{minipage}
    };
\end{tikzpicture}
		
		\subsubsection*{Construct programs and proofs.}
We define the program \graycode{pCircuit} that will be used in the refinement.

\begin{tikzpicture}
  \node[fill=gray!20, inner xsep=5pt, inner ysep =5pt] (box) {
    \begin{minipage}{0.96\textwidth}
      \lstset{style=Quire}
      \begin{lstlisting}
Def pCircuit := Prog
    H[q0]; H[q1];
    CCX[q0 q1 t]; S[t]; CCX[q0 q1 t];
    H[q0]; H[q1];
    if Pnot00[q0 q1] then Z[t] else skip end.
    \end{lstlisting}
      \end{minipage}
    };
\end{tikzpicture}

		The refinement process can also be constructed in the form of a program statement. The following command checks $[|0\>_q\<0|, |1\>_q\<1|] \sqsubseteq q \apply X$ and defines the program \graycode{ex} as \graycode{X[q]} with the history of refinement.
  
\begin{tikzpicture}
  \node[fill=gray!20, inner xsep=5pt, inner ysep =5pt] (box) {
    \begin{minipage}{0.96\textwidth}
      \lstset{style=Quire}
\begin{lstlisting}
Def ex := Prog < P0[q], P1[q] > <= X[q].
\end{lstlisting}
      \end{minipage}
    };
\end{tikzpicture}

		\subsection{Stepwise refinement of \texttt{Rz}}
		Similar to other interactive theorem provers, \texttt{Quire} supports the tactic mode for stepwise refinement. The following code refines the target prescription $[\Omega_{t,t'}, \mathcal{R}_z (\theta)_t(\Omega_{t,t'})] \equiv [\Omega_{t,t'}, R_z(\theta)_t \Omega_{t, t'} R_z(\theta)_t^\dagger]$.

\begin{tikzpicture}
  \node[fill=gray!20, inner xsep=5pt, inner ysep =5pt] (box) {
    \begin{minipage}{0.96\textwidth}
      \lstset{style=Quire}
\begin{lstlisting}
Refine pf : < Omega[t t'], Rz[t] * Omega[t t'] * Rz[t]$\dagger$ >.
    Step Seq Pnot00[q0 q1] * Omega[t t'].
    Step [q0 q1] :=0; X[q0].
    Def Inv0 := (Pnot00[q0 q1] $\otimes$ Omega[t t']) $\vee$ (P00[q0 q1] $\otimes$ (Rz[t] * Omega[t t'] * Rz[t]$\dagger$)).
    Step While Pnot00[q0 q1] Inv IQOPT Inv0.
    Step Seq P00[q0 q1] $\otimes$ Omega[t t'].
    Step [q0 q1] :=0.
    Step proc pCircuit.
End.
\end{lstlisting}      
\end{minipage}
    };
\end{tikzpicture}

		See Fig. \ref{fig:refine illustration} for an illustration.
		As we proceed with every step, the updated information for the current goals is shown in the output file. As in Sec. \ref{subsec:zrotation}, the first step is to insert an intermediate predicate using \graycode{\textcolor{NavyBlue}{\textbf{Step Seq}}} command. 
		It turns the original prescription into two subgoals. 
		The next step finishes the first goal using \graycode{[q0 q1] :=0; X[q0]}. Regarding the second goal, we define the invariant \graycode{Inv0} as in Sec. \ref{subsec:zrotation}, then use command \graycode{\textcolor{NavyBlue}{\textbf{Step While}}} to refine the goal prescription to a while program according to \textsc{(Spc-Pres-While)}. The following steps complete the goal with an initialization \graycode{[q0 q1] :=0} and the subprocedure \graycode{pCircuit}.
		
		In the following code, the proof \texttt{pf} is printed, the refinement result \graycode{S0} is extracted as a quantum program, and the classical simulation of $\llbracket S0 \rrbracket (|+\>_t\<+|)$ is calculated. 

\begin{tikzpicture}
  \node[fill=gray!20, inner xsep=5pt, inner ysep =5pt] (box) {
    \begin{minipage}{0.96\textwidth}
      \lstset{style=Quire}
\begin{lstlisting}
Show pf.
Def S0:= Extract pf. Show S0.
Def rho := [[proc S0]](Pp[t]). Show rho.
\end{lstlisting}
\end{minipage}
    };
\end{tikzpicture}

		\section{Documentation of \texttt{Quire}}
		\label{sec: doc_quire}
		This section explains the Python interface, commands and syntax of \texttt{Quire}. For further details, please refer to the GitHub repository at: \url{https://github.com/LucianoXu/Quire}.
		
		\subsection{Python interface}
		\begin{itemize}
			\item \color{blue} \texttt{quire\_code(input\_code, opts)} \\
			\color{black} Reset the prover with extra operators in \texttt{opts}, and process the \texttt{Quire} commands in the string \texttt{input\_code}.
			
			\item \color{blue} \texttt{quire\_file(input\_path, opts)}\\
			\color{black} Reset the prover with extra operators in \texttt{opts}, and process the \texttt{Quire} commands in the file at \texttt{input\_path}.
			
			\item \color{blue} \texttt{quire\_server(input, output, opts)}\\
			\color{black} Reset the prover with extra operators in \texttt{opts}, and start an interactive server which processes \texttt{Quire} commands in the file \texttt{input} and outputs information in the file \texttt{output}.\\
			Note: Modify and save the \texttt{input} file to update the input, and use \texttt{Ctrl+C} to close the server.
		\end{itemize}

		In the following introduction of syntax, $C$ denotes an identifier for a constant, $stm$ denotes a quantum program, $qvar$ denotes a quantum register, $o$ denotes an unlabelled operator, and $oi$ denotes a labelled quantum operator.
		
		\subsection{\texttt{Quire} commands.}
		The \texttt{Quire} tool is manipulated by a simple imperative language, which consists of a sequence of commands. 
		
		The possible commands include: 
		\begin{itemize}
			\item \color{blue} $\texttt{Def}\ C\ \texttt{:=}\ o \texttt{.}$ \\
			\color{black} Define the constant $C$ as the unlabelled operator $o$.
			
			\item \color{blue} $\texttt{Def}\ C\ \texttt{:=}\ oi \texttt{.}$\\
			\color{black} Define the constant $C$ as the labelled operator $oi$. Here and below, labelled operators mean quantum operators operating on the specified quantum systems.
			
			\item \color{blue} $\texttt{Def}\ C\ \texttt{:= [[}stm \texttt{]](} oi \texttt{).} $\\
			\color{black} Define the constant $C$ as the classical simulation result of executing the quantum program $stm$ on the quantum state $oi$.
			
			\item \color{blue} $\texttt{Def}\ C\ \texttt{:=}\ \texttt{Prog}\ stm \texttt{.}$ \\
			\color{black} Define the constant $C$ as the quantum program $stm$.
			
			\item \color{blue} $\texttt{Def}\ C_1\ \texttt{:=}\ \texttt{Extract}\ C_2 \texttt{.}$\\
			\color{black} Define the constant $C_1$ as the program extracted from the program/proof $C_2$.
			
			\item \color{blue} $\texttt{Refine}\ C : pres \texttt{.}$\\
			\color{black} Define $C$ as the prescription $pres$ and start the stepwise refinement mode on it.
			
			\item \color{blue} $\texttt{Step}\ stm \texttt{.}$\\
			\color{black} (refinement mode) Try to refine the current goal with the program $stm$. 
			
			\item \color{blue} $\texttt{Step}\ \texttt{Seq}\ oi \texttt{.}$\\
			\color{black} (refinement mode) Apply the refinement rule \textsc{(Comm-Pres-Cons)} $[P, Q] \sqsubseteq [P, R]; [R, Q]$ to the current goal, where $R$ is specified by $oi$.
			
			\item \color{blue} $\texttt{Step}\ \texttt{If}\ oi \texttt{.}$ \\
			\color{black} (refinement mode) Apply the refinement rule \textsc{(Spc-Pres-Cond)}: 
			$$
			[P,Q]_{\bar{q}}\equiv \iif\ R[\bar{q}]\ \then\ [R\doublecap P, Q]_{\bar{q}}\ \eelse\ [R^\bot\doublecap P, Q]_{\bar{q}}
			$$ 
			to the current goal, where $R$ is specified by $oi$.
			
			\item \color{blue} $\texttt{Step}\ \texttt{While}\ oi_1\ \texttt{Inv}\ oi_2 \texttt{.}$\\
			\color{black} (refinement mode) Apply the refinement rule \textsc{(Spc-Pres-While)}
			$$
			[Inv, P^\bot \doublecap Inv]_{\bar{q}}\le \while\ P[\bar{q}]\ \ddo\ [P\doublecap Inv, Inv]_{\bar{q}}\ \pend
			$$
			to the current goal, where $P$, $Inv$ are specified by $oi_1$, $oi_2$ respectively.

                \item \color{blue} $\texttt{WeakenPre}\ oi \texttt{.}$\\
                \color{black} (refinement mode) Weaken the precondition of the goal to $oi$.

                \item \color{blue} $\texttt{StrengthenPost}\ oi \texttt{.}$\\
                \color{black} (refinement mode) Strengthen the postcondition of the goal as $oi$.
			
			\item \color{blue} $\texttt{Choose}\ N \texttt{.}$\\
			\color{black} (refinement mode) Chose the $N$-th goal as the current goal.
			
			\item \color{blue} $\texttt{End} \texttt{.}$\\
			\color{black} (refinement mode) Complete the refinement when all goals are clear.
			
			\item \color{blue} $\texttt{Pause} \texttt{.}$\\
			\color{black} (interactive server) Pause the parsing of the input file so that the current information of the prover can be shown in the output file.
			
			\item \color{blue} $\texttt{Show}\ \texttt{Def} \texttt{.}$\\
			\color{black} Print all the names of definitions in the environment.
			
			\item \color{blue} $\texttt{Show}\ C \texttt{.}$\\
			\color{black} Print the definition of $C$.
			
			\item \color{blue} $\texttt{Eval}\ C \texttt{.}$\\
			\color{black} Evaluate the definition $C$ (e.g. operator expressions) and print the value.
			
			\item \color{blue} $\texttt{Test}\ o_1\ \texttt{=}\ o_2 \texttt{.}$\\
			\color{black} Test whether $o_1 = o_2$ for the unlabelled operators $o_1$ and $o_2$.
			
			\item \color{blue} $\texttt{Test}\ o_1\ \texttt{<=}\ o_2 \texttt{.}$\\
			\color{black} Test whether $o_1 \sqsubseteq o_2$ for the unlabelled operators $o_1$ and $o_2$.
			
			\item \color{blue} $\texttt{Test}\ oi_1\ \texttt{=}\ oi_2 \texttt{.}$\\
			\color{black} Test whether $oi_1 = oi_2$ for the labelled operators $o_1$ and $o_2$.
			
			\item \color{blue} $\texttt{Test}\ oi_1\ \texttt{<=}\ oi_2 \texttt{.}$\\
			\color{black} Test whether $oi_1 \sqsubseteq oi_2$ for the labelled operators $o_1$ and $o_2$.
		\end{itemize}
		
		\subsection{Program Syntax}
		A quantum program $stm$ is constructed in the following ways:
		\begin{itemize}
			\item {\color{blue} $\texttt{abort}$}\\
			The abort program.
			\item {\color{blue} $\texttt{skip}$}\\
			The skip program.
			\item {\color{blue} $qvar \texttt{:=0}$}\\
			The initialization program.
			\item {\color{blue} $oi$}\\
			The unitary transformation.
			
			\item {\color{blue} $\texttt{assert}\ oi$}\\
			The assertion program.
			\item {\color{blue} $\texttt{< } oi_1\texttt{, } oi_2 \texttt{ >}$}\\
			The program prescription.
			
			\item {\color{blue} $ stm_1 \texttt{; }stm_2 $}\\
			The sequential composition. Right associative.
			
			\item {\color{blue} $ \texttt{(} stm_1\ \texttt{[} p\ \oplus\texttt{]}\ stm_2 \texttt{)}$\\
				$ \texttt{(} stm_1\ \texttt{[} p\ \texttt{\textbackslash oplus}\texttt{]}\ stm_2 \texttt{)}$}\\
			The probabilistic composition. $p$ is a real number.
			
			\item {\color{blue} $ \texttt{if}\ oi\ \texttt{then}\ stm_1\ \texttt{else}\ stm_0\ \texttt{end}$}\\
			The if program.
			
			\item {\color{blue} $ \texttt{while}\ oi\ \texttt{do}\ stm\ \texttt{end}$}\\
			The while program.
			
			\item {\color{blue} $ \texttt{proc}\ C$}\\
			The program of subprocedure $C$.
			
			\item {\color{blue} $ pres\ \texttt{<=}\ stm$}\\
			The program representing a refinement from prescription $prec$ to program $stm$. In other words, it is the program $stm$ with the refinement history.  An error will be reported if the refinement is not valid.
			
		\end{itemize}
		
		\subsection{Operator Syntax}
            In Quire, operators are categorized into unlabelled operators and labelled ones. Calculations involving unlabelled operators are limited to those with corresponding data types (qubit numbers). Additionally, due to automatic extensions, calculations between labelled operators are always permitted.
		
		The grammar for unlabelled operators is:
		$$
		\begin{aligned}
			o ::=\ &C\ |\ [v]\ |\ - o\ |\ o + o\ |\ o - o\\
			& |\ c * o\ |\ c\ o\ \\
			& |\ o * o\ |\ o\dagger \\
			& |\ o \otimes o\ \\
			& |\ o \vee o\ |\ o \wedge o\ |\ o \verb|^| \bot \\
			& |\ o \rightsquigarrow o\ |\ o \Cap o.
		\end{aligned}
		$$
  
        The operator $\texttt{[} v \texttt{]}$ corresponds to the projector $|v\>\<v|$. The syntax for $v$ is:
        $$
        v ::=\ \ket{\texttt{<bit string>}}\ |\ v + v\ |\ c * v\ |\ c\ v.
        $$
        		
		The grammar for labelled operators is:
		$$
		\begin{aligned}
			oi ::= \ & \text{\texttt{IQOPT}}\ C\ |\ o\ qvar \\
			& |\ -oi\ |\ oi + oi\ |\ oi - oi\\
			& |\ c*oi\ |\ c\ oi\\
			& |\ oi * oi\ |\ oi\dagger\\
			& |\ oi \otimes oi\\
			& |\ oi \vee oi\ |\ oi \wedge oi\ |\ oi\ \verb|^|\bot\\
			& |\ oi \rightsquigarrow oi\ |\ oi \Cap oi.
		\end{aligned}
		$$
		The Unicode characters can be replaced by ASCII strings: $\dagger$ by \texttt{\textbackslash dagger}, $\otimes$ by \texttt{\textbackslash otimes}, $\vee$ by \texttt{\textbackslash vee}, $\wedge$ by \texttt{\textbackslash wedge}, $\bot$ by \texttt{\textbackslash bot}, $\rightsquigarrow$ by \texttt{\textbackslash SasakiImply} and $\Cap$ by \texttt{\textbackslash SasakiConjunct}.
		
	\end{document}

%% file: lstlisting.tex
\usepackage{listings}

\lstdefinestyle{py}{
    language=Python,
    basicstyle=\ttfamily,
    breaklines=true,
    keywordstyle=\bfseries\color{NavyBlue},
    morekeywords={as},
    emph={self},
    emphstyle=\bfseries\color{Rhodamine},
    commentstyle=\itshape\color{black!50!white},
    stringstyle=\bfseries\color{PineGreen!90!black},
    columns=flexible,
}

\lstdefinestyle{Quire}{
    columns=fullflexible,
    mathescape = true,
    sensitive=true,
    basicstyle=\small,
    breaklines=true,
    morekeywords=[1]{Def, Test, Eval, End, Refine, Extract, Show, Pause, Prog, Step, Choose, Inv, Seq, If, While, WeakenPre},
    morekeywords=[2]{abort, skip, assert, if, then, else, end, while, do, proc, IQOPT}, 
    morekeywords=[3]{Goal, Clear},
    keywordstyle=[1]\bfseries\color{NavyBlue},
    keywordstyle=[2]\bfseries\color{black}, 
    keywordstyle=[3]\bfseries\color{black}, 
    morecomment=[l]{//},
    commentstyle=\itshape\color{black!50!white},
    string=[b]"
}

\lstdefinestyle{Quire-tiny}{
    columns=fullflexible,
    mathescape = true,
    sensitive=true,
    basicstyle=\footnotesize,
    breaklines=true,
    morekeywords=[1]{Def, Test, Eval, End, Refine, Extract, Show, Pause, Prog, Step, Choose, Inv, Seq, If, While, WeakenPre},
    morekeywords=[2]{abort, skip, assert, if, then, else, end, while, do, proc, IQOPT}, 
    morekeywords=[3]{Goal, Clear},
    keywordstyle=[1]\bfseries\color{NavyBlue},
    keywordstyle=[2]\bfseries\color{black}, 
    keywordstyle=[3]\bfseries\color{black}, 
    morecomment=[l]{//},
    commentstyle=\itshape\color{black!50!white},
    string=[b]"
}

%% file: pmymacro.tex
\newcommand \sker[1] {\mathcal{N}\left(#1\right)}
\newcommand {\empstr} {\Lambda}
\def\qcoin{\mathbf{QC}}

\newcommand {\qcf}[1] {{\sf{#1}}}

\newcommand {\qc}[1] {{\sf{#1}}}
\def\>{\ensuremath{\rangle}}
\def\<{\ensuremath{\langle}}
\def\sl {\ensuremath{\llparenthesis}}
\def\sr{\ensuremath{\rrparenthesis}}
\def\-{\ensuremath{\textrm{-}}}
\def\ott{t}
\def\otu{u}
\def\ots{s}
\def\apply{\mathrel{*\!\!=}}

\def\comm{\ensuremath{\leftrightarrow^*}}
\def\reach{\ensuremath{\rightarrow^*}}
\newcommand \alert[1] {{\color{red} #1}}
\newcommand{\pcom}[1]{\ {}_{#1}\!\!\oplus}
\def\ctp{P}
\def\ctq{Q}

\def\change{\ensuremath{\mathit{change}}}

\def\qVar{\ensuremath{\mathit{qv}}}
\def\qv{\ensuremath{\mathit{qv}}}
\def\cVar{\ensuremath{\mathit{cv}}}
\def\QVar{\ensuremath{\mathit{V}}}
\def\CVar{\ensuremath{\mathit{cVar}}}
\def\Var{\ensuremath{\mathit{var}}}
\def\Chan{\mathit{chan}}
\def\cChan{\mathit{cChan}}
\def\qChan{\mathit{qChan}}
\def\BExp{\mathit{BExp}}
\def\Exp{\mathbb{E}}

\def\fdmu{\Delta}
\def\fdnu{\dnu}
\def\fdomega{\domega}

\def\dmu{\mu}
\def\dnu{\nu}
\def\domega{\omega}
\def\expect{\mathbb{E}}
\def\preexpect{\mathrm{pre}\mathbb{E}}

\def\rassign{:=_{\$}}
\def\fpi{\widehat{\pi}}
\def\h{\ensuremath{\mathcal{H}}}
\def\p{\ensuremath{\mathcal{P}}}
\def\l{\ensuremath{\mathcal{L}}}
\def\g{\ensuremath{\mathcal{G}}}
\def\lh{\ensuremath{\mathcal{L(H)}}}
\def\dh{\ensuremath{\mathcal{D(H})}}
\def\dhv{\ensuremath{\d(\h_V)}}
\def\shv{\ensuremath{\s(\h_V)}}
\def\q{\bold Q}
\def\Q{\ensuremath{\mathbb Q}}
\def\P{\ensuremath{\mathbb P}}
\def\SO{\ensuremath{\mathcal{SO}}}
\def\HP{\ensuremath{\mathcal{HP}}}
\def\hpe{\ensuremath{\mathcal{\e}}}

\def\r{\ensuremath{\mathcal{R}}}
\def\R{\ensuremath{\mathbb{R}}}
\def\m{\ensuremath{\mathcal{M}}}
\def\u{\ensuremath{\mathcal{U}}}
\def\k{\ensuremath{\mathcal{K}}}
\def\K{\ensuremath{\mathfrak{K}}}
\def\S{\ensuremath{\mathfrak{S}}}
\def\s{\ensuremath{\mathcal{S}}}
\def\t{\ensuremath{\mathcal{T}}}
\def\u{\ensuremath{\mathcal{U}}}
\def\U{\ensuremath{\mathfrak{U}}}
\def\L{\ensuremath{\mathfrak{L}}}
\def\x{\ensuremath{\mathcal{X}}}
\def\y{\ensuremath{\mathcal{Y}}}
\def\z{\ensuremath{\mathcal{Z}}}

\def\st{\ensuremath{\mathfrak{t}}}
\def\su{\ensuremath{\mathfrak{u}}}
\def\ss{\ensuremath{\mathfrak{s}}}

\def\ra{\ensuremath{\rightarrow}}
\def\a{\ensuremath{\mathcal{A}}}
\def\b{\ensuremath{\mathcal{B}}}
\def\c{\ensuremath{\mathcal{C}}}

\def\e{\ensuremath{\mathcal{E}}}
\def\f{\ensuremath{\mathcal{F}}}
\def\l{\ensuremath{\mathcal{L}}}
\def\X{\mbox{\bf{X}}}
\def\N{\mathbb{N}}
\def\sreal{\mathbb{R}}
\def\Z{\mathbb{Z}}

\def\qzz{\ensuremath{|0\>_q\<0|}}
\def\qoo{\ensuremath{|1\>_q\<1|}}
\def\qzo{\ensuremath{|0\>_q\<1|}}
\def\qoz{\ensuremath{|1\>_q\<0|}}
\def\qii{\ensuremath{|i\>_q\<i|}}
\def\qiz{\ensuremath{|i\>_q\<0|}}
\def\qzi{\ensuremath{|0\>_q\<i|}}

\def\quzz{\ensuremath{|0\>_{\bar{q}}\<0|}}
\def\quoo{\ensuremath{|1\>_{\bar{q}}\<1|}}
\def\quzo{\ensuremath{|0\>_{\bar{q}}\<1|}}
\def\quoz{\ensuremath{|1\>_{\bar{q}}\<0|}}
\def\quii{\ensuremath{|i\>_{\bar{q}}\<i|}}
\def\quiz{\ensuremath{|i\>_{\bar{q}}\<0|}}
\def\quzi{\ensuremath{|0\>_{\bar{q}}\<i|}}

\DeclarePairedDelimiter{\ceil}{\lceil}{\rceil}

\def\d{\ensuremath{\mathcal{D}}}
\def\dh{\ensuremath{\mathcal{D(H)}}}
\def\lh{\ensuremath{\mathcal{L(H)}}}
\def\le{\ensuremath{\sqsubseteq}}
\def\ge{\ensuremath{\sqsupseteq}}
\def\lt{\ensuremath{\sqsubset}}
\def\gt{\ensuremath{\sqsupset}}
\def\eval{\ensuremath{{\psi}}}
\def\aeq{\ensuremath{{\ \equiv\ }}}
\def\osnt{\ensuremath{\sl \ott, \e\sr}}
\def\snt{\st}
\def\snti{\ensuremath{\sl \ott_i, \e_i\sr}}
\def\osnu{\ensuremath{\sl \otu, \f\sr}}
\def\osns{\ensuremath{\sl s, \g\sr}}
\def\snu{\su}
\def\fdist{\ensuremath{\d ist_\h}}
\def\dist{\ensuremath{Dist}}
\def\wtx{\ensuremath{\widetilde{X}}}

\def\bv{1{v}}
\def\bV{\mathbf{V}}
\def\bf{\mathbf{f}}
\def\bw{\mathbf{w}}
\def\zo{\mathbf{0}}
\def\bX{\mathbf{X}}
\def\bDelta{\mathbf{\Delta}}
\def\bdelta{\boldsymbol{\delta}}
\def\next{\mathcal{X}}
\def\until{\mathcal{U}}

\def\leqI{\ensuremath{\mathcal{SI}(\h)}}
\def\leqIq{\ensuremath{\mathcal{SI}_{\eqsim}(\h)}}
\def\oact{\ensuremath{\alpha}}
\def\oactb{\ensuremath{\beta}}
\def\sact{\ensuremath{\gamma}}
\def\fpi{\ensuremath{\widehat{\pi}}}
\newcommand{\supp}[1]{\ensuremath{\left\lceil{#1}\right\rceil}}
\newcommand{\support}[1]{\lceil{#1}\rceil}

\newcommand{\abis}{\stackrel{\lambda}\approx}
\newcommand{\abisa}[1]{\stackrel{#1}\approx}
\newcommand {\qbit} {\mbox{\bf{new}}}

\renewcommand{\theenumi}{(\arabic{enumi})}
\renewcommand{\labelenumi}{\theenumi}
\newcommand{\tr}{{\rm tr}}
\newcommand{\rto}[1]{\stackrel{#1}\rightarrow}
\newcommand{\orto}[1]{\stackrel{#1}\longrightarrow}
\newcommand{\srto}[1]{\stackrel{#1}\longmapsto}
\newcommand{\sRto}[1]{\stackrel{#1}\Longmapsto}

\newcommand{\ass}[3]{\left\{#1\right\}#2\left\{#3\right\}}
\newcommand{\iass}[3]{\left[#1\right] #2 \left[#3\right]}
\newcommand{\oass}[3]{\left\<#1\right\> #2 \left\<#3\right\>}

\newcommand{\andor}{\ \&\ }

\newcommand {\true} {\ensuremath{{\mathbf{true}}}}
\newcommand {\false} {\ensuremath{{\mathbf{false}}}}
\newcommand {\abort}{\ensuremath{{\mathbf{abort}}}}
\newcommand {\sskip} {\mathbf{skip}}

\newcommand {\then} {\ensuremath{\mathbf{then}}}
\newcommand {\eelse} {\ensuremath{\mathbf{else}}}
\newcommand {\while} {\ensuremath{\mathbf{while}}}
\newcommand {\ddo} {\ensuremath{\mathbf{do}}}
\newcommand {\pend} {\ensuremath{\mathbf{end}}}
\newcommand {\inv} {\ensuremath{\mathbf{inv}}}

\newcommand {\mymeas} {\mathbf{meas}}

\newcommand \assert[1] {\mathbf{assert}\ #1}
\newcommand {\fail} {\mathbf{fail}}
\newcommand {\iif} {\mathbf{if}}
\newcommand {\fii} {\mathbf{fi}}
\newcommand {\od} {\mathbf{od}}
\newcommand {\irepeat} {\mathbf{repeat}}
\newcommand {\iuntil} {\mathbf{until}}
\newcommand {\blocal} {\mathbf{begin\ local}}
\def\mstm{\iif\ b\ \then\ S_1\ \eelse\ S_0\ \pend}
\def\wstm{\while\ b\ \ddo\ S\ \pend}

\newcommand\pmeasstm[3]{\iif\ #1\ \then\ #2\ \eelse\ #3\ \pend}
\def\pmstm{\iif\ P[\bar{q}] \ \then\ S_1\ \eelse\ S_0\ \pend}
\def\pwstm{\while\ P[\bar{q}] \ \ddo\ S\ \pend}

\newcommand\measstm[3]{\iif\ #1\ra\ #2\ \square\ \neg (#1)\ra #3\ \fii}

\newcommand\whilestm[2]{\while\ #1\ \ddo\ #2\ \pend}

\newcommand\repstm[2]{\irepeat\ #1\ \iuntil\ #2}

\newcommand\alterex{\iif\ B_1\ra S_1 \square\ldots\square B_n\ra S_n\ \fii}

\newcommand\altercom{\iif\ \square_{i=1}^n B_i\ra S_i\ \fii}

\newcommand\repex{\ddo\ B_1\ra S_1 \square\ldots\square B_n\ra S_n\ \od}

\newcommand\repcom{\ddo\ \square_{i=1}^n B_i\ra S_i\ \od}

\newcommand\seqcom{\ddo\ \square_{i=1}^n B_i;\alpha_i\ra S_i\ \od}

\newcommand {\spann} {\mathrm{span}}

\newcommand{\rrto}[1]{\xhookrightarrow{#1}}
\newcommand{\con}[3]{\iif\ {#1}\ \then\ {#2}\ \eelse\ {#3}\ \pend}

\newcommand{\Rto}[1]{\stackrel{#1}\Longrightarrow}
\newcommand{\nrto}[1]{\stackrel{#1}\nrightarrow}

\newcommand{\Rhto}[1]{\stackrel{\widehat{#1}}\Longrightarrow}
\newcommand{\define}{\ensuremath{\triangleq}}
\newcommand{\rsim}{\simeq}
\newcommand{\obis}{\approx_o}
\newcommand{\sbis}{\ \dot\approx\ } 
\newcommand{\stbis}{\ \dot\sim\ } 
\newcommand{\nssbis}{\ \dot\nsim\ } 

\newcommand{\bis}{\sim}
\newcommand{\rat}{\rightarrowtail}
\newcommand{\wbis}{\approx}
\newcommand{\id}{\mathcal{I}}
\newcommand{\stet}[1]{\{ {#1}  \}  } 
\newcommand{\unw}[1]{\stackrel{{#1}}\sim}
\newcommand{\rma}[1]{\stackrel{{#1}}\approx}

\def\step{\textsf{step}}
\def\obs{\textsf{obs}}
\def\dom{\textsf{dom}}
\def\purge{\textsf{ipurge}}
\def\source{\textsf{sources}}
\def\cnt{\textsf{cnt}}
\def\read{\textsf{read}}
\def\alter{\textsf{alter}}
\def\dirac#1{\delta_{#1}}

\def \srho {\sqrt{\rho}}
\def\tybool{\ensuremath{\mathbf{Boolean}}}
\def\tyint{\ensuremath{\mathbf{Integer}}}
\def\tyqubit{\ensuremath{\mathbf{Qubit}}}
\def\tyqudit{\ensuremath{\mathbf{Qudit}}}
\def\tyqureg{\ensuremath{\mathbf{Qureg}}}
\def\tyunitreg{\ensuremath{\mathbf{Unitreg}}}
\def\type{\ensuremath{\mathit{type}}}

\def\qconc{\mathcal{Q}}
\def\ps{\rm{PS}}
\def\is{\rm{IN}}
\def\aS{\sem{S}^{\#}}
\def\amap#1{\sem{#1}^{\#}}
\def\qstate{\rho}
\def\qassert{\Theta}
\def\qassertp{\Psi}
\def\casserts{\a}
\def\cstate{S}
\def\cstates{\prog}
\def\cassert{p}
\def\emptydis{\bot}
\def\qset{Q}
\def\qsetp{R}

\newcommand {\tot} {\mathit{tot}}
\newcommand {\pal} {\mathit{par}}

\def\supoprset{\mathbb{E}}

\def\leinf{\le_{\mathit{inf}}}
\def\geinf{\ge_{\mathit{inf}}}
\def\qstates{\s_V}
\def\qasserts{\a_V}
\def\qstatesh#1{\d(\mathcal{H}_{#1})}

\def\qassertsh#1{\mathcal{A}_{#1}}

\def\qstatesp{\mathcal{S}(\h')}

\newcommand\prog{\mathit{Prog}}
\def\ph{\ensuremath{\mathcal{P}(\h)}}
\def\phv{\ensuremath{\mathcal{P}(\h_V)}}

\def\l{\mathcal{L}}
\def\k{\mathcal{K}}
\def\qmc {\color{red}}
\def\dtmc {\color{black}}
\newcommand{\ysim}[1]{\stackrel{#1}\sim}
\def\z{\mathbf{0}}
\newcommand{\TRANDA}[3]{#1\xrightarrow{#2}_{{\sf D}}#3}
\def\pdist{\mathit{pDist}}

\def\C{\mathbb{C}}

\newcommand{\subs}[2]{{#2}/{#1}}

\def \Rm#1{\mbox{\rm #1}}
\def \lsem      {\llbracket}
\def \rsem      {\rrbracket}
\def \sem#1{\lsem #1 \rsem}

\def \asemp#1{\mbox{\lsem$#1$\rsem$_1^\#$}}
\def \asems#1{\mbox{\lsem$#1$\rsem$_\sigma^\#$}}
\def \asemsr#1{\mbox{\lsem$#1$\rsem$_{\sigma,\r}^\#$}}

\newtheorem{remark}{Remark}